\documentclass[11pt,reqno,tikz]{amsart}
\usepackage{amsmath,amsthm,amssymb,bm,bbm,dsfont,braket,esint,nicefrac}
\usepackage[mathscr]{eucal}
\usepackage{amsaddr}
\usepackage{upgreek}
\usepackage[left=2cm,right=2cm,top=2cm,bottom=2cm]{geometry}
\usepackage{mathtools}\mathtoolsset{centercolon}
\usepackage[nospace,noadjust]{cite}

\usepackage{pifont} 
\usepackage{adforn} 

\usepackage[shortlabels]{enumitem}
\usepackage{setspace}
\usepackage{graphicx}
\usepackage{verbatim}
\usepackage{float}
\usepackage{placeins}
\usepackage{array}
\usepackage{booktabs}
\usepackage{multicol,multirow,tabularx}
\usepackage{threeparttable}
\usepackage[update,prepend]{epstopdf}
\usepackage{amsfonts,amssymb,dsfont,xfrac,comment}
\usepackage[abs]{overpic}

\makeatletter
\def\adl@drawiv#1#2#3{%
	\hskip0
	\tabcolsep
	\xleaders#3{#2 0\@tempdimb #1{1}#2 0.5\@tempdimb}%
	#2\z@ plus1fil minus1fil\relax
	\hskip0\tabcolsep}
\newcommand{\cdashlinelr}[1]{%
	\noalign{\vskip\aboverulesep
		\global\let\@dashdrawstore\adl@draw
		\global\let\adl@draw\adl@drawiv}
	\cdashline{#1}
	\noalign{\global\let\adl@draw\@dashdrawstore
		\vskip\belowrulesep}}
\makeatother

\usepackage[dvipsnames]{xcolor}
\usepackage{colortbl}
\usepackage{arydshln}
\usepackage[hidelinks,linktocpage=true,pagebackref]{hyperref}
\hypersetup{
	unicode=false,          
	pdftoolbar=true,        
	pdfmenubar=true,        
	pdffitwindow=false,     
	pdfstartview={FitH},    
	pdfnewwindow=true,      
	colorlinks=true,        
	linkcolor=Red,          
	citecolor=ForestGreen,  
	filecolor=Magenta,      
	urlcolor=BlueViolet,    
}
\renewcommand*{\backref}[1]{}
\renewcommand*{\backrefalt}[4]{%
    \ifcase #1%
          \or [Cited on page~#2.]%
          \else [Cited on pages~#2.]%
    \fi%
    }

\usepackage{doi}
\usepackage{url}
\usepackage{caption, subcaption}
\usepackage{enumitem}
\usepackage{shadethm}
\usepackage{tikz,pgf}

\usepackage{cleveref}

\makeatletter
\ifx\@NODS\undefined%

\let\mathbb=\mathds
\else%
\fi
\makeatother

\def\d{{\text {\rm d}}}

\DeclarePairedDelimiter{\ceil}{\lceil}{\rceil}
\DeclareMathOperator*{\argmax}{arg\,max}

\DeclareMathOperator{\Tr}{Tr}
\DeclareMathOperator{\Log}{Log}
\DeclareMathOperator{\Arg}{Arg}
\DeclareMathOperator{\spec}{spec}
\DeclareMathOperator{\sign}{sign}
\DeclareMathOperator{\I}{\mathbf{1}}
\DeclareMathOperator{\supp}{supp}

\def\X{\mathsf{X}}
\def\Y{\mathsf{Y}}
\def\Z{\mathsf{Z}}
\def\A{\mathsf{A}}
\def\B{\mathsf{B}}
\def\C{\mathsf{C}}

\def\E{\mathsf{E}}
\def\R{\mathsf{R}}
\def\Q{\mathsf{Q}}

\newcommand{\proj}[1]{\left\{#1\right\}} 

\newcommand{\be}{{\mathbf e}}

\newcommand{\norm}[2]{\parallel \! #1 \! \parallel_{#2}}
\newcommand{\Renyi}{{}R\'{e}nyi{ }}

\def\0{{\mathbf{0}}}
\def\1{{\mathbf{1}}}
\def\2{{\mathbf{2}}}
\def\3{{\mathbf{3}}}
\def\4{{\mathbf{4}}}
\def\5{{\mathbf{5}}}
\def\6{{\mathbf{6}}}

\def\7{{\mathbf{7}}}
\def\8{{\mathbf{8}}}
\def\9{{\mathbf{9}}}

\def\be{\begin{equation}}
\def\ee{\end{equation}}
\def\bea{\begin{eqnarray}}
\def\eea{\end{eqnarray}}
\def\eps{\varepsilon}

\newcommand{\id}{\operatorname{id}}

\theoremstyle{plain}
\newshadetheorem{theo}{Theorem}[section]
\newtheorem{conj}{Conjecture} 
\newshadetheorem{prop}[theo]{Proposition}
\newshadetheorem{lemm}[theo]{Lemma}
\newtheorem*{fact}{Fact} 

\theoremstyle{definition}
\newtheorem{defn}[theo]{Definition} 

\theoremstyle{remark}
\newtheorem{remark}[theo]{Remark}

\makeatletter
\newcommand{\opnorm}{\@ifstar\@opnorms\@opnorm}
\newcommand{\@opnorms}[1]{%
	$\left|\mkern-1.5mu\left|\mkern-1.5mu\left|
	#1
	\right|\mkern-1.5mu\right|\mkern-1.5mu\right|$
}
\newcommand{\@opnorm}[2][]{%
	\mathopen{#1|\mkern-1.5mu#1|\mkern-1.5mu#1|}
	#2
	\mathclose{#1|\mkern-1.5mu#1|\mkern-1.5mu#1|}
}
\makeatother

\usepackage{tikz}
\usetikzlibrary{arrows.meta} 
\tikzset{>={Latex[length=4,width=4]}} 
\usetikzlibrary{calc}

\colorlet{mylightblue}{blue!5!white}
\colorlet{mydarkblue}{blue!30!black}
\colorlet{myblue}{blue!50!black}
\colorlet{myred}{red!50!black}
\colorlet{mydarkred}{red!30!black}
\colorlet{mydarkgreen}{green!30!black}

\newcommand{\sh}{\kern-0.08em$^\textbf{\#}$\hspace{-3pt}}
\renewcommand{\b}{\kern-0.06em$\flat$}

\begin{document}

\let\origmaketitle\maketitle
\def\maketitle{
	\begingroup
	\def\uppercasenonmath##1{} 
	\let\MakeUppercase\relax 
	\origmaketitle
	\endgroup
}

\title{\bfseries \Large{ 
Error Exponents for Quantum Packing Problems\\
via an Operator Layer Cake Theorem
		}}

\author{ \normalsize 
{Hao-Chung Cheng}$^{1\text{--}5}$
and
{Po-Chieh Liu}$^{1,2}$
}
\address{\small  	
$^1$Department of Electrical Engineering and Graduate Institute of Communication Engineering,\\ National Taiwan University, Taipei 106, Taiwan (R.O.C.)\\
$^2$Department of Mathematics, National Taiwan University\\
$^3$Center for Quantum Science and Engineering, National Taiwan University\\
$^4$Hon Hai (Foxconn) Quantum Computing Center, New Taipei City 236, Taiwan (R.O.C.)\\
$^5$Physics and Mathematics Divisions, National Center for Theoretical Sciences, Taipei 10617, Taiwan (R.O.C.)\\
}

\email{\href{mailto:haochung.ch@gmail.com}{haochung.ch@gmail.com}}

\date{\today}

\begin{abstract}
We establish a one-shot random coding bound for classical-quantum channel coding with a universal dimension-independent prefactor, resolving Burnashev and Holevo's 1998 conjecture.
The bound holds for arbitrary input distributions; after optimization, its asymptotic exponent matches the reliability function of classical-quantum channels above the critical rate, including channels with infinite-dimensional output systems.
This error analysis framework naturally extends to a broad class of quantum packing problems, encompassing classical communication over arbitrary quantum channels, with or without entanglement assistance; entanglement-assisted quantum communication; constant-composition codes; and classical data compression with quantum side information under fixed- and variable-length coding.

Our central tool is an operator layer cake theorem, expressing the directional derivative of the operator logarithm as an integral of spectral projections.
This representation identifies the integral pretty-good measurement with a randomized Holevo--Helstrom measurement, giving an operational explanation of why the pretty-good measurement is pretty good.
\end{abstract}

\maketitle
\vspace{-2.8em}
\tableofcontents

\section{Introduction} \label{sec:introduction}

Shannon's celebrated noisy-channel coding theorem \cite{Sha48} provides an information-theoretic characterization of the ultimate communication capability of a noisy channel, i.e., a stochastic map $p_{\Y\mid\X}$ from an input alphabet $\X$ to an output alphabet $\Y$.
The \emph{achievability} theorem provides codes for sending $k$ bits of information over $n$ uses of the channel such that the associated probability of erroneous decoding at the receiver behaves like
\begin{align} \label{eq:vanishing_error}
\varepsilon_n \rightarrow 0, \quad \text{as} \;\; n\to \infty,
\end{align}
if the ratio $\nicefrac{k}{n}$ (whose limit is called the transmission rate) is below the channel capacity of $p_{\Y\mid\X}$.
The (weak) converse states that, otherwise, vanishing error is not possible.
Shannon's asymptotic result is profoundly simple, yet it has driven major research developments and technological advances in the information era, and later in the field of quantum Shannon theory as well.
	
One unaddressed question was how fast the error in \eqref{eq:vanishing_error} approaches $0$ as $n$ increases.
The convergence rate of error is operationally relevant because, in practice, one may ask how many channel uses (known as the \emph{query complexity} \cite{sample_complexity_25, NW25}) are needed to achieve a prescribed error tolerance, say $10^{-6}$ .
Pioneered by Feinstein \cite{Fei54}, Shannon \cite{Sha59b} and later refined by Gallager \cite{Gal65}, these coding arguments and error bounds have been the focus of extensive research efforts.
Notably, Gallager \cite[Theorem 1]{Gal65} (see also \cite[Theorem 5.6.2]{Gal68}) established a mathematically elegant error estimate in terms of a power-mean expression via random coding with maximum-likelihood decoding: for any input distribution $p_{\X}$,
\begin{align} \label{eq:Gallager}
    \varepsilon_n
&\leq (2^k)^{\frac{1-\alpha}{\alpha}} \left( 
\sum_{y\in\Y}
\left( \sum_{x\in\X} p_{\X}(x) (p_{\Y|\X}(y| x))^{\alpha} \right)^{\nicefrac{1}{\alpha}}
\right)^n, \quad \forall\, \alpha \in [\nicefrac{1}{2},1], \,n\in\mathds{N}.
\end{align}
It was later shown that the exponential decay rate (the so-called \emph{random-coding exponent}) obtained after choosing the best $p_{\X}$ is optimal for rates above a certain critical value \cite{SGB67, SGB67b}.
More importantly, Gallager's random coding bound holds for \emph{any} blocklength $n$ (including short blocklengths), unlike the asymptotic result in \eqref{eq:vanishing_error}.

If the underlying physical medium of the communication channel is quantum mechanical, the channel is a quantum evolution that transforms the state of a quantum system into another state as output.
The so-called HSW theorem \cite{SW97, Hol98} extends Shannon's noisy-channel coding theorem to sending classical information over a quantum channel with asymptotic error behavior as in \eqref{eq:vanishing_error}.
Inspired by Gallager's random coding bound in \eqref{eq:Gallager}, Burnashev and Holevo studied one of the simplest forms of quantum channels---\emph{classical-quantum channels}---that send a codeword $x_1 x_2 \ldots x_n$ of length $n$ to a product state $\rho_{\B_1}^{x_1} \otimes \rho_{\B_2}^{x_2} \otimes \cdots \otimes \rho_{\B_n}^{x_n}$ at the output quantum system $\B^n$ \cite{BH98}, and made the following conjecture for classical-quantum channels.

\smallskip
\noindent\fbox{\begin{minipage}{1\textwidth}
\begin{conj}[Burnashev and Holevo 1998 {\cite{BH98}, \cite[(5)]{Hol00}}] \label{conjecture}
For any classical-quantum channel $x\mapsto \rho_{\B}^x$ and any input distribution $p_{\X}$, the random coding error satisfies 
\begin{align} \label{eq:conjecture}
\varepsilon_n
&\leq c \cdot (2^k)^{\frac{1-\alpha}{\alpha}} \left( 
\Tr\left[
\left( \sum_{x\in\X} p_{\X}(x) (\rho_{\B}^x)^{\alpha} \right)^{\nicefrac{1}{\alpha}}
\right]
\right)^n, \quad \forall\, \alpha \in [\nicefrac{1}{2},1], \,n\in\mathds{N}
\end{align}
for some constant $c$.
\end{conj}
\end{minipage}}
\smallskip

Using the quantum Sibson identity \cite{KW09, SW12, HT14, CGH18}, one can rewrite \eqref{eq:conjecture} in the form:
\begin{align} \label{eq:bound_Renyi}
\varepsilon_n\leq
c \cdot 2^{  -n \sup_{\alpha \in [\nicefrac{1}{2},1]} \frac{1-\alpha}{\alpha} \left[ I_{\alpha}(\X:\B)_{\rho} - \nicefrac{k}{n} \right] },
\end{align}
where $D_{\alpha}(\rho\Vert\sigma) \!\coloneqq\! \frac{1}{\alpha-1}\log_2 \Tr[\rho^{\alpha}\sigma^{1-\alpha}]$ is the Petz--\Renyi divergence, and 
$I_{\alpha}(\X\!:\!\B)_{\rho}
\!\coloneqq\!\! \inf\limits_{\text{state} \;\sigma_{\B}} \! D_{\alpha}\left( \rho_{\X\B} \Vert \rho_{\X} \!\otimes\! \sigma_{\B} \right)$
is the \Renyi information with respect to the joint input-output classical-quantum state $\rho_{\X\B}$.

Burnashev and Holevo's conjectured random coding bound is of theoretical and practical significance for the following reasons.
\begin{enumerate}[(I)]
    \item 
    Although the bound \eqref{eq:conjecture} is written as an $n$-shot expression, it is actually a \emph{one-shot bound} (i.e., $n=1$), and hence it holds also for non-stationary channels and even for arbitrary channels that do not possess any independent and identically distributed (i.i.d.)~structure.
    This consideration is especially prominent for the quantum scenario as imposing a technical i.i.d.~assumption on quantum channels is not always realistic.

    \item
    The error probability $\varepsilon_n$ decays exponentially not only in the asymptotic limit $n\to \infty$, but also for \textbf{any} short and moderate blocklength $n$. 
    Unlike channel capacity, which is in general not achievable in finite blocklength, the random coding bound of the form \eqref{eq:Gallager} is achievable in one-shot.\footnote{Even with the \emph{channel dispersion} back-off term of capacity, large blocklengths of order $\mathcal{O}(\frac{1}{\varepsilon^2})$ (where $\varepsilon$ is the target error probability to achieve) are still required to achieve the second-order rate.}
    At present, realizing a large-scale quantum device (e.g.,~the collective measurement for decoding) remains challenging.

    \item 
    The bound \eqref{eq:conjecture} holds for \emph{any} input distribution $p_{\X}$.
    Hence, without knowing the optimal $p_{\X}$, the guarantee of the exponential decay holds for any suboptimal one.
    This feature is useful from the coding-theoretic perspective because computationally finding the optimal input distributions is generally hard, and practically implementing such an optimal random block code could be challenging.

    \item 
    The constant $c$ is independent of the dimension of the output quantum system $\B$, which allows an accurate estimate of the error probability even for large quantum systems.
\end{enumerate}

At that time, Burnashev and Holevo proved Conjecture~\ref{conjecture} for the special case of pure-state channels with $c=2$.
Later, Hayashi proved an exponential-decay bound with a weaker error exponent, but for any mixed-state channels as well \cite{Hay07}.
Dalai established an asymptotic sphere-packing bound, which improves on Winter's converse bound \cite{Win99} (in the so-called Haroutunian form), and 
matches the error exponent in \eqref{eq:conjecture} for \emph{optimal} input distributions and for rates above the critical rate \cite{Dal13}.
Later, the sphere-packing bound was refined to any constant-composition code \cite{DW14, CHT19} and for finite blocklengths \cite{CHT19}.
Burnashev and Holevo's result for pure-state channels and Hayashi's bound were slightly improved to tighter one-shot bounds \cite{Cheng_simple}.
In ~\cite{BT24}, Beigi and Tomamichel proved \eqref{eq:bound_Renyi} with $c=1$ but with a measured \Renyi information \cite{Don86, HP91}, which recovers Gallager's result in the commuting case.
Substantial progress was made by Renes, who showed Conjecture~\ref{conjecture} with a dimension-dependent constant $c$ and with $p_{\X}$ being the uniform distribution \cite{Ren23}.
This matches Dalai's sphere-packing bound for \emph{symmetric} classical-quantum channels.
Essentially, Renes bypassed Burnashev and Holevo's conjecture but provided a new route via the duality relation to achieve the optimal error exponent for finite-dimensional symmetric classical-quantum channels (together with Hayashi's achievability bound on privacy amplification in the dual domain \cite[Theorem 1]{Hay15_PA}).
Recently, Renes employed Gallager's shaping method \cite{Gal68, Ren25} and, concurrently, Li and Yang utilized the method of types \cite{LY25} to asymptotically approximate the optimal exponent-achieving input distribution so as to match Dalai's sphere-packing bound.
We refer the reader to~\cite{MHU18} for other possible input-shaping methods and their practical overhead.

Although tremendous progress has been made in quantum Shannon theory over the past decades --- see, e.g.~recent developments in error exponent analysis of other quantum information-theoretic tasks \cite{CHT19, CHDH-2018, CHDH2-2018, Dup21, KL21, LY21a, LY21b, LY24a, LY24b, Cheng_simple, CG22, CG22b, SGC22b, AB24, ATB24, OCC+24, CDG24} --- Burnashev and Holevo's conjecture remains a long-standing open question in the field.\footnote{Burnashev and Holevo's conjecture was publicly mentioned in Holevo's \emph{2016 Claude E.~Shannon Award Lecture} at \textit{2016 IEEE International Symposium on Information Theory}, Barcelona, Spain \cite{Holevo2016Shannon}.}
In the following, we provide possible reasons why this problem is so challenging.
\begin{enumerate}[(i)]
    \item 
    Historically, the field lacked a systematic and sharp analytical tool for quantum state discrimination.
    Essentially, one can resort to the maximum-likelihood decoding as in Gallager's random coding bound \eqref{eq:Gallager}.
    Unfortunately, there is no maximum-likelihood decoder in the quantum setting due to noncommutativity.
    The optimal quantum Bayesian decoder does not have a closed-form expression.

    \item 
    One of the key ingredients in large deviation analysis is \emph{tilting} (i.e., a Markov-type inequality in probability theory).
    Although the tilting question has been solved in binary quantum hypothesis testing, it is nontrivial how to extend it to general quantum information-theoretic problems.

    \item
    Due to noncommutativity, there are different proposed quantum \Renyi divergences.
    It has been shown that the Petz--\Renyi divergence \cite{Pet86} has operational meaning in some scenarios \cite{ANS+08}, whereas the sandwiched \Renyi divergence \cite{MDS+13, WWY14} has operational meaning in others \cite{MO15, MO14, LY24a}.
    However, the governing principles are not clear.

    \item 
    A powerful technique in quantum information theory, called \emph{pinching}, was developed by Hayashi \cite{Hay02} to force operators to be commutative.
    How to apply pinching is unclear here.
    
    \item 
    Csisz{\'a}r--K{\"o}rner's random coding based on the method of types \cite[Theorem 10.2]{CK11} might not be applicable here.
    As we will show later in Remark~\ref{remark:CK}, the so-called \emph{dual-domain expression} of the exponent (which naturally appears via the method of types) corresponds to another \emph{larger} quantum \Renyi divergence, which should not be achievable.

    \item 
    In certain circumstances, it is useful to employ additional resources (i.e.,~stronger or even non-physical correlations) to assist the task, yielding a lower error. 
    Then, one uses the assisted performance to estimate the unassisted performance via the \emph{rounding technique} 
    \cite{FSS19, BFO24, AB24, ATB24}.
    The rounding technique works in several settings, but its applicability here is unclear.

\end{enumerate}

\medskip

We prove Burnashev and Holevo's Conjecture~\ref{conjecture} in the affirmative.
We establish the bound 
\eqref{eq:conjecture} with a dimension-independent constant $c< 1.102$ by constructing an \emph{integral pretty-good measurement} defined in terms of the directional derivative of the operator logarithm.
(See Section~\ref{sec:QHT} for more detailed definitions.)
Such measurements were recently analyzed by Beigi and Tomamichel \cite{BT24}.

Moreover, the established random coding bound applies beyond classical-quantum channels to other \emph{quantum packing-type problems}, including classical communication over \emph{any} fully quantum channel with or without entanglement assistance, quantum communication with entanglement assistance, constant-composition codes (for which the established error exponent matches the sphere-packing bound for \emph{any} composition \cite{DW14, CHT19}), and classical data compression for both fixed-length coding and variable-length coding.
Table~\ref{table:main} summarizes the main error-exponent results established here.
Table~\ref{table:comparison} provides a comparison with the prior work on Burnashev and Holevo's conjecture.
Tables~\ref{table:survey_packing} and \ref{table:survey_covering} collect some known exponent results in quantum information theory.

Our key technical contribution is to show that the above-mentioned two-outcome integral pretty-good measurement admits an extremal decomposition into Holevo--Helstrom measurements.
In other words, it is actually equivalent to randomized optimal measurements.
With this, we can directly apply the known results of binary quantum hypothesis testing (e.g.,~the information spectrum method).
Moreover, the above interpretation also provides an intuitive explanation of why the pretty-good measurement is near-optimal.
Indeed, the Holevo--Helstrom measurement can achieve the optimal error exponent even under mismatched priors, which only contribute a constant multiplicative factor; 
integrating the resulting constants over the integral PGM  induces at most a constant multiplicative cost to the optimal error.

The extremal decomposition is a special case of the established \emph{operator layer cake theorem} (Theorem~\ref{theo:Dlog_formula}): For $A>0$ and $B\geq 0$,
\begin{align}
    \lim_{t\to0} \frac{\log(A+tB) - \log A}{t} 
    =
    \int_0^\infty \proj{uA < B} \, \d u
    - \int_{-\infty}^0 \proj{uA>B}\, \d u,
\end{align}
(here $\proj{uA < B}$ denotes the projection onto the positive spectral subspace of $B-uA$).

\medskip
This paper is organized as follows.
Section~\ref{sec:notation} introduces necessary notation.
Section~\ref{sec:QHT} establishes the extremal decomposition of the integral
pretty-good measurement (Theorem~\ref{theo:extremal}) and the tilting inequality
(Proposition~\ref{prop:key}) that underlies all subsequent results.
Section~\ref{sec:CQ} presents a solution to Burnashev and Holevo's conjecture for classical-quantum channels.
Section~\ref{sec:CC} considers constrained codebooks and constant-composition codes.
Section~\ref{sec:CQSW} studies classical data compression with quantum side information under fixed-length coding and variable-length coding.
Sections~\ref{sec:unassisted} and \ref{sec:EA} investigate unassisted and entanglement-assisted classical communication over arbitrary quantum channels, respectively.
Section~\ref{sec:EAQ} extends to entanglement-assisted quantum communication over quantum channels.
We conclude the paper in Section~\ref{sec:conclusions}.
The appendices collect properties of the operator logarithm (Appendix~\ref{sec:log})
and prove the operator layer cake theorem (Appendix~\ref{sec:layer-cake}).


\subsection{Notation} \label{sec:notation}

Let quantum systems (or quantum registers) $\A, \B, \ldots$ be associated with Hilbert spaces $\mathcal{H}_{\A}, \mathcal{H}_{\B}, \ldots $, respectively.
We use $|\A|$ to denote the dimension of $\A$. Unless explicitly stated otherwise, Hilbert spaces are finite-dimensional.
The quantum state of a system $\A$ is represented by a density operator $\rho_{\A}$, i.e., a positive semidefinite operator with unit trace on $\mathcal{H}_{\A}$.
The set of quantum states on $\mathcal{H}_{\A}$ is denoted by $\mathcal{S}(\A)$.
We denote by $\I_{\A}$ the identity operator on $\mathcal{H}_{\A}$, i.e., $\I_{\A} = \sum_{i} |i\rangle\langle i|_{\A}$ for any orthonormal basis $\{|i\rangle_{\A}\}_i$ of $\mathcal{H}_{\A}$.
Sometimes we omit the system subscript if we do not specify it.
Given a bipartite state $\rho_{\A\B} \in \mathcal{S}(\A\B)$, the marginal state of system $\A$ is denoted by $\rho_{\A}$, obtained by tracing out system $\B$: $\rho_{\A} = \Tr_{\B}\left[ \rho_{\A\B}\right] \coloneqq\sum_{i} (\I_{\A} \otimes \langle i \rvert_{\B})\rho_{\A\B} (\I_{\A}\otimes \lvert i\rangle_{\B})$, where $\{\lvert i\rangle_{\B}\}_i$ is any orthonormal basis of $\mathcal{H}_{\B}$. 
If we do not specify the subscript of $\Tr$, we mean that all quantum systems are traced out.
We denote the adjoint by $\dagger$.

For a scalar-valued function $f$ and a normal operator $A$ with spectral decomposition $A = \sum_i \lambda_i |i\rangle\langle i|$ whose spectrum $\{\lambda_i\}_i$ (denoted by $\spec (A)$) lies in the domain of $f$, we define the operator $f(A)$ via functional calculus: $f(A) \coloneqq\sum_i f(\lambda_i) \lvert i\rangle \langle i\rvert$.
Note that logarithms and negative powers are taken on the relevant supports; $A^0$ means the projection operator onto $\supp(A)$.
The operator norm and trace norm of $A$ are denoted by 
$\|A\|_{\infty} \coloneqq \sup_{\|v\|=1} \|Av\|$,
and $\|A\|_1 \coloneqq\Tr[\sqrt{A^\dagger A}]$, respectively.
We adopt L\"owner's partial order on self-adjoint operators; $A>B$ (resp.~$A\geq B$) for self-adjoint $A$ and $B$ means that $A-B>0$ (resp.~$A-B\geq 0$) is a positive definite (resp.~positive semidefinite) operator.
For a self-adjoint operator $X$ with spectral decomposition $X = \sum_i \lambda_i |i\rangle \langle i|$, we define the orthogonal projection onto its positive support by
$\left\{ X > 0 \right\} 
\coloneqq\sum_{i\colon \lambda_i>0}  |i\rangle \langle i|$.
Similarly, $\left\{ X \geq 0 \right\} 
\coloneqq\sum_{i\colon \lambda_i\geq0}  |i\rangle \langle i|$ and $\left\{ X = 0 \right\} 
\coloneqq\sum_{i\colon \lambda_i=0}  |i\rangle \langle i|$.


\begin{table}[t!]
	\centering
	\resizebox{1\columnwidth}{!}{
		\begin{tabular}{@{}cccr@{}}

        \toprule
		\addlinespace

        \textbf{Task} & \textbf{Codebook} & \multicolumn{2}{c}{\textbf{Error upper bound}}

        \\
        
        \midrule
        
        \multirow{4}{4cm}{\centering Classical-quantum channel coding}
        & \multirow{2}{*}{$\forall$ i.i.d.~codebook $p_{\X}$} & \multirow{2}{*}{$c_{\alpha}\cdot 2^{-n \frac{1-\alpha}{\alpha} \left[ I_{\alpha}(\X:\B)_{\rho} - R \right] }$} & \multirow{2}{*}{(Theorem~\ref{theo:CQ})}

        \\

         &  &  &

        \\

        \arrayrulecolor{black!50} \cmidrule{2-2} \cmidrule(lr){3-4}
        
        &  constant-composition & \multirow{2}{*}{$\mathcal{O}(n^{|\X|})\cdot 2^{-n \frac{1-\alpha}{\alpha} \left[ {I}^{\text{Aug}}_{\alpha}(q;\,\mathscr{N}) - R \right] }$ } &\multirow{2}{*}{(Theorem~\ref{theo:CC})}

        \\

        & codes ($\forall$ $n$-type $q_{\X}$) &  &

        \\

        \arrayrulecolor{black}\midrule

        \multirow{6}{4.5cm}{\centering Source coding with quantum side information}
        & i.i.d.~sources $\rho_{\X\B}$ & \multirow{2}{*}{$c_{\alpha}\cdot 2^{-n \frac{1-\alpha}{\alpha} \left[ R - H_{\alpha}(\X
        \,\mid\, \B)_{\rho} \right] }$} & \multirow{2}{*}{(Theorem~\ref{theo:CQSW_iid})}

        \\

        & (fixed-length) & &

        \\

        \arrayrulecolor{black!50}\cmidrule{2-2} \cmidrule(lr){3-4}

         & constant-type $q_{\X}$ & \multirow{2}{*}{$\mathcal{O}(n^{|\X|})\cdot 2^{-n \frac{1-\alpha}{\alpha} \left[ R - H(\X)_{q} + {I}^{\text{Aug}}_{\alpha}(q;\,x\mapsto \rho_{\B}^x)  \right] }$ }& \multirow{2}{*}{(Theorem~\ref{theo:CQSW_cc})}

        \\

        & (fixed-length) & &

        \\

        \arrayrulecolor{black!50}\cmidrule{2-2} \cmidrule(lr){3-4}

        & i.i.d.~source $\rho_{\X\B}$& \multirow{2}{*}{$\mathcal{O}(n^{|\X|})\cdot 2^{-n \frac{1-\alpha}{\alpha} \left[ \bar{R} - H(\X)_{p} + {I}^{\text{Aug}}_{\alpha}(p;\,x\mapsto \rho_{\B}^x)  \right] }$ } &\multirow{2}{*}{(Theorem~\ref{theo:CQSW_variable})}

        \\

        & (variable-length) & &

        \\

        \arrayrulecolor{black}\midrule

        \multirow{3}{4cm}{\centering Unassisted classical communication over quantum channels} & \multirow{3}{*}{$\forall$ ensemble $\rho_{\X^n\A^n}$} & \multirow{3}{*}{$c_{\alpha} \cdot
        2^{
        -n\frac{1-\alpha}{\alpha} \left[ \frac{1}{n} I_{\alpha} (\X^n : \B^n)_{\mathscr{N}^{\otimes n}(\rho)} - R  \right]
        }$ } & \multirow{3}{*}{(Theorem~\ref{theo:CQ_unassisted})}

        \\

        & & &

        \\

        & & & 

        \\

        \arrayrulecolor{black}\midrule

        \multirow{3}{4.5cm}{\centering Entanglement-assisted classical communication over quantum channels} & \multirow{3}{*}{$\forall$ entanglement $\theta_{\mathsf{R}^n\A^n}$} & \multirow{3}{*}{$c_{\alpha} \cdot
        2^{
        -n\frac{1-\alpha}{\alpha} \left[ \frac{1}{n}I_{\alpha} (\mathsf{R}^n : \B^n)_{\mathscr{N}^{\otimes n}(\theta)} - R  \right]
        }$ } & \multirow{3}{*}{(Theorem~\ref{theo:EA})}

        \\

        & & &

        \\

        & & & 

        \\

        \arrayrulecolor{black}\midrule

        \multirow{3}{4.5cm}{\centering Entanglement-assisted quantum communication over quantum channels} & \multirow{3}{*}{$\forall$ entanglement $\theta_{\mathsf{R}^n\A^n}$} & \multirow{3}{*}{$c_{\alpha} \cdot
        2^{
        -n\frac{1-\alpha}{\alpha} \left[ \frac{1}{n}I_{\alpha} (\mathsf{R}^n : \B^n)_{\mathscr{N}^{\otimes n}(\theta)} - 2R  \right]
        }$ } & \multirow{3}{*}{(Theorem~\ref{theo:EAQ})}

        \\

        & & &

        \\

        & & & 

        \\

        \arrayrulecolor{black}\bottomrule
        
        \end{tabular}
	}
	\caption{
		\small 
        Summary of the established finite-blocklength error exponents for various quantum packing-type problems.
        All error (upper) bounds hold for all $\alpha \in [\nicefrac{1}{2},1]$ and all blocklengths $n\in\mathds{N}$; the prefactor $c_{\alpha}$ is uniformly bounded by $1.102$.
        For fixed-length coding, the rate is defined as $R\coloneqq\frac{1}{n}\log_2 |\mathsf{M}|$.
        For variable-length source coding, the average rate is defined in \eqref{eq:rate_CQSW_variable}.
	}	\label{table:main}	
\end{table}

\begin{table}[ht!]
	\centering
	\resizebox{1\columnwidth}{!}{
		\begin{tabular}{@{}lccccccccccc@{}}

        &  
        \multicolumn{5}{c}{\cellcolor{gray!10} Burnashev--Holevo’s 1998 conjecture for c-q channels}
        & &
        \multicolumn{5}{c}{\cellcolor{gray!10} Beyond Burnashev--Holevo’s conjecture}

        \\
        \addlinespace
        
        \toprule
		\addlinespace

         & \multirow{2}{*}{One-shot} & Asymptotically  & \multirow{2}{*}{$\forall$ Inputs} & Infinite &\multirow{2}{*}{Prefactor} & & c.~c.~ & CQSW & CQSW & Fully quantum & Entanglement


         \\

        & & tight$^*$ & & dimension & & & codes & fixed-length & variable-length & channels & assistance

        \\
        
        \cmidrule(r){1-7} \cmidrule{8-12}
        \addlinespace

        {\small{Burnashev--Holevo \!\cite{BH98}}} & {\color{ForestGreen}\ding{51}} & Pure-state c-q & {\color{ForestGreen}\ding{51}} & {\color{ForestGreen}\ding{51}} & $2$ & & & & & & 

        \\

        \arrayrulecolor{black!35}
        \cmidrule(rr){1-7} \cmidrule{8-12}
        
        Hayashi \cite{Hay07} & {\color{ForestGreen}\ding{51}} & {\color{Red}\ding{53}} & {\color{ForestGreen}\ding{51}} & {\color{Red}\ding{53}} & $4$ & & & & & &

        \\

        \arrayrulecolor{black!35}
        \cmidrule(rr){1-7} \cmidrule{8-12}
        
        Cheng \cite{Cheng_simple} & {\color{ForestGreen}\ding{51}} & {\color{Red}\ding{53}} & {\color{ForestGreen}\ding{51}} & {\color{ForestGreen}\ding{51}} & $1$ & & & & & &

        \\

        \cmidrule(rr){1-7} \cmidrule{7-12}

        Renes \cite{Ren23} & {\color{ForestGreen}\ding{51}} & Symmetric c-q & {\color{Red}\ding{53}} & {\color{Red}\ding{53}} & $ \frac{\alpha \cdot \nu_{\B}^{ \frac{\alpha+1}{\alpha}  } }{1-\alpha}$ & &  & {\color{ForestGreen}\ding{51}} & & &

        \\

        \cmidrule(rr){1-7} \cmidrule{8-12}

        {\footnotesize{Beigi--Tomamichel \cite{BT24}}} & {{\color{ForestGreen}\ding{51}}} & {{\color{Red}\ding{53}}} & {{\color{ForestGreen}\ding{51}}} & {\color{Red}\ding{53}} &{$1$} & & & & & & 
        
        \\

        \cmidrule(rr){1-7} \cmidrule{8-12}

        Renes \cite{Ren25}, & \multirow{2}{*}{{\color{Red}\ding{53}}} & \multirow{2}{*}{{\color{ForestGreen}\ding{51}}} & \multirow{2}{*}{{\color{ForestGreen}\ding{51}}} & \multirow{2}{*}{{\color{Red}\ding{53}}} & \multirow{2}{*}{$\textrm{poly}(n)$} & & & & & & 

        \\

        Li--Yang \cite{LY25} & & & & & & & & & & &

        \\
    
        \cmidrule(rr){1-7} \cmidrule{8-12}

        \cellcolor{gray!15}& \cellcolor{gray!15}& \cellcolor{gray!15}& \cellcolor{gray!15}& \cellcolor{gray!15}& \cellcolor{gray!15}& \cellcolor{gray!15}& \cellcolor{gray!15}& \cellcolor{gray!15}& \cellcolor{gray!15}& \cellcolor{gray!15}& \cellcolor{gray!15}
        
        \\

        \cellcolor{gray!15} \multirow{-2}{*}{\textbf{This Work}} & \cellcolor{gray!15} \multirow{-2}{*}{{\color{ForestGreen}\ding{51}}} & \cellcolor{gray!15} \multirow{-2}{*}{{\color{ForestGreen}\ding{51}}} & \cellcolor{gray!15} \multirow{-2}{*}{{\color{ForestGreen}\ding{51}}} &
        \cellcolor{gray!15} \multirow{-2}{*}{{\color{ForestGreen}\ding{51}}} & \cellcolor{gray!15} \multirow{-2}{*}{$c_{\alpha} \!<\! 1.102$}  & \cellcolor{gray!15} & \cellcolor{gray!15} \multirow{-2}{*}{{\color{ForestGreen}\ding{51}}} & \cellcolor{gray!15} \multirow{-2}{*}{{\color{ForestGreen}\ding{51}}} & \cellcolor{gray!15} \multirow{-2}{*}{{\color{ForestGreen}\ding{51}}} & \cellcolor{gray!15} \multirow{-2}{*}{{\color{ForestGreen}\ding{51}}} & \cellcolor{gray!15} \multirow{-2}{*}{{\color{ForestGreen}\ding{51}}}

        \\
        
        
        \arrayrulecolor{black}\bottomrule
        
        \end{tabular}
	}
	\caption{
		\small 
        Comparison with prior work on Burnashev and Holevo's conjecture.
        The left block shows classical-quantum (c-q) channels, and the right block shows extensions of Burnashev and Holevo's conjecture to other quantum packing-type problems.
        Here, \emph{asymptotic tightness}$^*$ is with respect to certain critical-rate regions.
        The coefficient $\nu_{\B}$ is the number of distinct eigenvalues of an operator on system $\B$.
	}	\label{table:comparison}	
\end{table}

\begin{table}[htbp]
	\centering
	\resizebox{1\columnwidth}{!}{
		\begin{tabular}{@{}cccccc@{}}

        \toprule

         & \multicolumn{2}{c}{\cellcolor{gray!10} \textbf{Error Exponent}} & & \multicolumn{2}{c}{\cellcolor{gray!10} \textbf{Strong Converse Exponent}} 

        \\

        \arrayrulecolor{black!55}
        \cmidrule{2-3}\cmidrule{5-6}

        & Achievability & Converse & & Achievability & Converse

        \\

        \arrayrulecolor{black}\cmidrule{1-1}\cmidrule{2-3} \cmidrule{5-6} 

        
        \cellcolor{orange!5}& \multicolumn{2}{c}{\multirow{2}{*}{$\sup\limits_{\alpha \in [0,1]} (1-\alpha) D_{\alpha}(\rho\Vert\sigma)$}}
        &
        & \multicolumn{2}{c}{\multirow{2}{*}{$0$}}

        \\

        \cellcolor{orange!5}& & & & &
        
        \\

        \multirow{-3}{3.5cm}{\cellcolor{orange!5}\centering Symmetric binary hypothesis testing}& \cite{ACM+07, ANS+08, JOP+12, beigi2025some, LHC25_layer_cake} & \cite{NS09, ANS+08} & & &

        \\
        
        \arrayrulecolor{black!45}\cdashlinelr{1-6}

        \cellcolor{orange!5}& \multicolumn{2}{c}{\multirow{2}{*}{$\min_{i\neq j}\sup\limits_{\alpha \in [0,1]} (1-\alpha) D_{\alpha}(\rho_i \Vert \rho_j)$}}
        &
        & \multicolumn{2}{c}{\multirow{2}{*}{$0$}}

        \\

        \cellcolor{orange!5}& & & & &
        
        \\

        \multirow{-3}{3.5cm}{\cellcolor{orange!5}\centering Symmetric multiple hypothesis testing}& \cite{Li16} & \cite{NS09, ANS+08} & & &

        \\

        \arrayrulecolor{black!45}\cdashlinelr{1-6}

        \cellcolor{orange!5}& \multicolumn{2}{c}{\multirow{2}{*}{$\sup\limits_{\alpha \in [0,1]} \frac{1-\alpha}{\alpha} \left[ D_{\alpha}(\rho\Vert\sigma) - R \right]$}}
        &
        & \multicolumn{2}{c}{\multirow{2}{*}{$\sup\limits_{\alpha > 1} \frac{\alpha-1}{\alpha} \left[ R - \widetilde{D}_{\alpha}(\rho\Vert\sigma) \right]$}}

        \\

        \cellcolor{orange!5}& & & & &
        
        \\

        \multirow{-3}{3.5cm}{\cellcolor{orange!5}\centering Asymmetric binary hypothesis testing}& \cite{ACM+07, ANS+08, JOP+12, beigi2025some, LHC25_layer_cake} & \cite{NS09, ANS+08} & & \multicolumn{2}{c}{\cite{MO14}}

        \\
        
        \arrayrulecolor{black!45}\cdashlinelr{1-6}

        \cellcolor{orange!5}& \multirow{2}{*}{$\sup\limits_{\alpha \in [\nicefrac{1}{2},1]} \frac{1-\alpha}{\alpha}\left[ I_{\alpha}(\X:\B)_{\rho} -R \right]$}
        & \multirow{3}{*}{\textbf{?}}
        &
        & \multirow{3}{*}{\textbf{?}}
        & \multirow{2}{*}{$\sup\limits_{\alpha > 1} \frac{\alpha-1}{\alpha}\left[ R - \widetilde{I}_{\alpha}(\X:\B)_{\rho} \right]$}

        \\

        \cellcolor{orange!5}& & & & &
        
        \\

        \multirow{-3}{3.5cm}{\cellcolor{orange!5}\centering Classical-quantum channel coding (i.i.d.~codes)}& (Theorem~\ref{theo:CQ}) & & & & \cite{CG24}
        
        \\

        \arrayrulecolor{black!45}\cdashlinelr{1-6}
        
        \cellcolor{orange!5}& \multirow{2}{*}{$\sup\limits_{\alpha \in [\nicefrac{1}{2},1]} \frac{1-\alpha}{\alpha}\left[ {I}^{\text{Aug}}_{\alpha}(p;\mathscr{N}) -R \right]$}
        & \multirow{2}{*}{$\sup\limits_{\alpha \in (0,1]} \frac{1-\alpha}{\alpha}\left[ {I}^{\text{Aug}}_{\alpha}(p;\mathscr{N}) -R \right]$}
        &
        & \multicolumn{2}{c}{\multirow{2}{*}{$\sup\limits_{\alpha > 1} \frac{\alpha-1}{\alpha}\left[ R - \widetilde{I}^{\text{Aug}}_{\alpha}(p;\mathscr{N}) \right]$}}

        \\

        \cellcolor{orange!5}& & & & &
        
        \\

        \multirow{-3}{3.5cm}{\cellcolor{orange!5}\centering Classical-quantum channel coding (c.~c.~codes)}& (Theorem~\ref{theo:CC}) &  \cite{DW14,CHT19,SC25} & & \cite{MO18} &{\cite{CHDH2-2018}}
        
        \\

        \arrayrulecolor{black!45}\cdashlinelr{1-6}
        
        \cellcolor{orange!5}& \multirow{2}{*}{$\sup\limits_{\alpha \in [\nicefrac{1}{2},1]} \frac{1-\alpha}{\alpha}\left[ I_{\alpha}(\mathscr{N}_{\X\to\B}) -R \right]$}
        & \multirow{2}{*}{$\sup\limits_{\alpha \in (0,1]} \frac{1-\alpha}{\alpha}\left[ I_{\alpha}(\mathscr{N}_{\X\to\B}) -R \right]$}
        &
        & \multicolumn{2}{c}{\multirow{2}{*}{$\sup\limits_{\alpha > 1} \frac{\alpha-1}{\alpha}\left[ R - \widetilde{I}_{\alpha}(\mathscr{N}_{\X\to\B})  \right]$}}

        \\

        \cellcolor{orange!5}& & & & &
        
        \\

        \multirow{-3}{3.5cm}{\cellcolor{orange!5}\centering Classical-quantum channel coding (optimal codes)} & (Thms.~\ref{theo:CQ}, \ref{theo:CC}) \cite{Ren25, LY25} &  \cite{Dal13, DW14,CHT19,SC25} & & \cite{MO14, MO18} & \cite{MO14, WWY14, CG24}
        
        \\

        \arrayrulecolor{black!45}\cdashlinelr{1-6}
        
        \cellcolor{orange!5}& \multicolumn{2}{c}{\multirow{2}{*}{$\sup\limits_{\alpha \in (0,1]} \frac{1-\alpha}{\alpha}\left[ I_{\alpha}(\mathscr{N}_{\X\to\B}) -R \right]$ \;(w/ \!activation)}}
        &
        & \multicolumn{2}{c}{\multirow{2}{*}{$\sup\limits_{\alpha > 1} \frac{\alpha-1}{\alpha}\left[ R - \widetilde{I}_{\alpha}(\mathscr{N}_{\X\to\B})  \right]$}}

        \\

        \cellcolor{orange!5}& & & & &
        
        \\

        \multirow{-3}{3.5cm}{\cellcolor{orange!5}\centering Non-signaling (NS) classical-quantum channel coding} & \multicolumn{2}{c}{\cite{ATB24}} & & \multicolumn{2}{c}{\cite{AB24}}
        
        \\        

        \arrayrulecolor{black!45}\cdashlinelr{1-6}
        
        \cellcolor{orange!5}& \multirow{2}{*}{$\sup\limits_{\alpha \in [\nicefrac{1}{2},1]} \frac{1-\alpha}{\alpha}\left[ \chi_{\alpha}^{\text{reg}}(\mathscr{N}_{\A\to\B}) -R \right]$}
        & \multirow{3}{*}{\textbf{?}}
        & 
        & \multirow{3}{*}{\textbf{?}} & \multirow{3}{*}{\textbf{?}}

        \\

        \cellcolor{orange!5}& & & & &
        
        \\

        \multirow{-3}{3.5cm}{\cellcolor{orange!5}\centering Unassisted classical communication over quantum channels} & (Theorem~\ref{theo:CQ_unassisted})  &   & & &
        
        \\

        \arrayrulecolor{black!45}\cdashlinelr{1-6}
        
        \cellcolor{orange!5}& \multirow{2}{*}{$\sup\limits_{\alpha \in [\nicefrac{1}{2},1]} \frac{1-\alpha}{\alpha}\left[ I^{\text{reg}}_{\alpha}(\mathscr{N}_{\A\to\B}) -R \right]$}
        & \multirow{3}{*}{\textbf{?}}
        & 
        & \multicolumn{2}{c}{\multirow{2}{*}{$\sup\limits_{\alpha > 1} \frac{\alpha-1}{\alpha}\left[ R - \widetilde{I}_{\alpha}(\mathscr{N}_{\A\to\B}) \right]$}}

        \\

        \cellcolor{orange!5}& & & & &
        
        \\

        \multirow{-3}{3.5cm}{\cellcolor{orange!5}\centering EA classical communication over quantum channels} & (Theorem~\ref{theo:EA})  &   & & \cite{LY24b, AB24} & \cite{GW14}
        
        \\

        \arrayrulecolor{black!45}\cdashlinelr{1-6}
        
        \cellcolor{orange!5}& \multirow{2}{*}{$\sup\limits_{\alpha \in [\nicefrac{1}{2},1]} \frac{1-\alpha}{\alpha}\left[ I^{\text{reg}}_{\alpha}(\mathscr{N}_{\A\to\B}) - 2R \right]$}
        & \multirow{3}{*}{\textbf{?}}
        & 
        & \multicolumn{2}{c}{\multirow{2}{*}{$\sup\limits_{\alpha > 1} \frac{\alpha-1}{\alpha}\left[ 2R - \widetilde{I}_{\alpha}(\mathscr{N}_{\A\to\B}) \right]$}}

        \\

        \cellcolor{orange!5}& & & & &
        
        \\

        \multirow{-3}{3.5cm}{\cellcolor{orange!5}\centering EA quantum communication over quantum channels} & (Theorem~\ref{theo:EAQ})  &   & & \cite{LY24b, AB24} & \cite{GW14}
        
        \\        

        \arrayrulecolor{black!45}\cdashlinelr{1-6}
        
        \cellcolor{orange!5}& \multirow{3}{*}{$\sup\limits_{\alpha \in (0,1]} \!\!\frac{1-\alpha}{\alpha}\!\left[ I^{\text{reg}}_{\alpha}(\mathscr{N}_{\A\to\B}) \!-\!R \right]$ (w/ \!activation)}
        & \multirow{3}{*}{\textbf{?}}
        & 
        & \multicolumn{2}{c}{\multirow{2}{*}{$\sup\limits_{\alpha > 1} \frac{\alpha-1}{\alpha}\left[ R - \widetilde{I}_{\alpha}(\mathscr{N}_{\A\to\B}) \right]$}}

        \\

        \cellcolor{orange!5}& & & & &
        
        \\

        \multirow{-3}{3.5cm}{\cellcolor{orange!5}\centering NS classical communication over quantum channels} & \cite{ATB24} &   & & \multicolumn{2}{c}{\cite{AB24}}
        
        \\        

        \arrayrulecolor{black!45}\cdashlinelr{1-6}
        
        \cellcolor{orange!5}& \multirow{2}{*}{$\sup\limits_{\alpha \in [\nicefrac{1}{2},1]} \frac{1-\alpha}{\alpha}\left[ R- H_{\alpha}(\X\mid\B)_{\rho}  \right]$}
        & \multirow{2}{*}{$\sup\limits_{\alpha \in (0,1]} \frac{1-\alpha}{\alpha}\left[ R-H_{\alpha}(\X\mid\B)_{\rho} \right]$}
        &
        & \multicolumn{2}{c}{\multirow{2}{*}{$\sup\limits_{\alpha > 1} \frac{\alpha-1}{\alpha}\left[ \widetilde{H}_{\alpha}(\X\mid\B)_{\rho} - R\right]$}}

        \\

        \cellcolor{orange!5}& & & & &
        
        \\

        \multirow{-3}{3.5cm}{\cellcolor{orange!5}\centering Source coding with QSI (i.i.d.~sources)} & (Theorem~\ref{theo:CQSW}) \cite{Ren23} &  \cite{CHDH-2018} & & \multicolumn{2}{c}{\cite{CHDH-2018}}
        
        \\

        \arrayrulecolor{black!45}\cdashlinelr{1-6}
        
        \cellcolor{orange!5}& \multirow{2}{*}{$\sup\limits_{\alpha \in [\nicefrac{1}{2},1]} \frac{1-\alpha}{\alpha}\left[ R - H(\X)_q + {I}^{\text{Aug}}_{\alpha}(q;\mathscr{N}) \right]$}
        & \multirow{2}{*}{$\sup\limits_{\alpha \in (0,1]} \frac{1-\alpha}{\alpha}\left[ R - H(\X)_q + {I}^{\text{Aug}}_{\alpha}(q;\mathscr{N}) \right]$}
        &
        & \multicolumn{2}{c}{\multirow{2}{*}{$\sup\limits_{\alpha > 1} \frac{\alpha-1}{\alpha}\left[ H(\X)_q - \widetilde{I}^{\text{Aug}}_{\alpha}(q;\mathscr{N}) - R  \right]$}}

        \\

        \cellcolor{orange!5}& & & & &
        
        \\

        \multirow{-3}{3.5cm}{\cellcolor{orange!5}\centering Source coding with QSI (constant type)} & (Theorem~\ref{theo:CQSW_cc})  &  \cite{CHDH2-2018} & & \multicolumn{2}{c}{\cite{CHDH2-2018}}
        
        \\

        \arrayrulecolor{black!45}\cdashlinelr{1-6}
        
        \cellcolor{orange!5}& \multirow{2}{*}{$\sup\limits_{\alpha \in [\nicefrac{1}{2},1]} \frac{1-\alpha}{\alpha}\left[ \bar{R} - H(\X)_p + {I}^{\text{Aug}}_{\alpha}(p;\mathscr{N}) \right]$}
        & \multirow{2}{*}{$\sup\limits_{\alpha \in (0,1]} \frac{1-\alpha}{\alpha}\left[ \bar{R} - H(\X)_p + {I}^{\text{Aug}}_{\alpha}(p;\mathscr{N}) \right]$}
        &
        & \multicolumn{2}{c}{\multirow{2}{*}{$0$}}

        \\

        \cellcolor{orange!5}& & & & &
        
        \\

        \multirow{-3}{3.5cm}{\cellcolor{orange!5}\centering Source coding with \!QSI (variable \!length)} & (Theorem~\ref{theo:CQSW_variable})  &  \cite{CHDH2-2018} & & \multicolumn{2}{c}{\cite{CHDH2-2018}}
        
        \\

                \arrayrulecolor{black!45}\cdashlinelr{1-6}

        \cellcolor{orange!5}& \multicolumn{2}{c}{\multirow{3}{*}{$\lim_{n\to\infty} \frac{1}{n} \sup_{\alpha\in(0,1]} \frac{1-\alpha}{\alpha}\left[\inf_{\sigma_n \in \text{SEP}_n} D_{\alpha}\left(\rho^{\otimes n} \Vert \sigma_n \right) -  R \right]$}}
        &
        & \multicolumn{2}{c}{\multirow{3}{*}{$\lim_{n\to\infty} \frac{1}{n} \sup_{\alpha>1} \frac{\alpha-1}{\alpha}\left[R - \inf_{\sigma_n \in \text{SEP}_n} \widetilde{D}_{\alpha}\left(\rho^{\otimes n} \Vert \sigma_n \right) \right]$}}

        \\

        \cellcolor{orange!5}& & & & &
        
        \\

        \cellcolor{orange!5}& & & & &
        
        \\

        \multirow{-4}{3.2cm}{\cellcolor{orange!5}\centering Non-entangling-assisted entanglement distillation} & \multicolumn{2}{c}{\cite{lin2026entanglement}} & & \multicolumn{2}{c}{\cite{lin2026entanglement}}

        \\
        
    \arrayrulecolor{black}\bottomrule
    \end{tabular}
	}
	\caption{
    Summary of the exponent analysis for various quantum packing-type problems (highlighted in \colorbox{orange!10}{orange})
    in quantum information theory.
    The regularized quantities are denoted by ``\text{reg}'' in the superscript (see e.g.~\eqref{eq:Holevo-reg}).
    All information quantities with ``$\widetilde{{\color{white}tt}}$'' are defined with respect to the sandwiched \Renyi divergence \cite{MDS+13,WWY14}.
    Assistance with entanglement (resp.~non-signaling) is abbreviated by EA (resp.~NS).
    We omit expurgated bounds because of space limitations.
    We denote by the ``\textbf{?}'' results that are not yet established to the best of our knowledge.
	}	\label{table:survey_packing}	
\end{table}

\begin{table}[htbp]
	\centering
	\resizebox{1\columnwidth}{!}{
		\begin{tabular}{@{}cccccc@{}}

        \toprule

         & \multicolumn{2}{c}{\cellcolor{gray!10} \textbf{Error Exponent}} & & \multicolumn{2}{c}{\cellcolor{gray!10} \textbf{Strong Converse Exponent}} 

        \\

        \arrayrulecolor{black!55}
        \cmidrule{2-3}\cmidrule{5-6}

        & Achievability & Converse & & Achievability & Converse

        \\

        \arrayrulecolor{black}\cmidrule{1-1}\cmidrule{2-3} \cmidrule{5-6} 

        \cellcolor{violet!5}& \multirow{2}{*}{$\sup\limits_{\alpha \in [1,2]} \frac{\alpha-1}{\alpha}\left[ R -  \widetilde{I}_{\alpha}(\X:\B)_{\rho} \right]$}
        & \multirow{3}{*}{\textbf{?}}
        &
        & \multirow{3}{*}{\textbf{?}}
        & \multirow{2}{*}{$\sup\limits_{\alpha \in [0, 1]} (1-\alpha)\left[ I_{\alpha}(\X:\B)_{\rho} -R \right]$}

        \\

        \cellcolor{violet!5}& & & & &
        
        \\

        \multirow{-3}{3.2cm}{\cellcolor{violet!5}\centering Classical-quantum soft covering, i.i.d. (trace distance)} & \cite{CG22} & & & & \cite{CG22}
        
        \\

        \arrayrulecolor{black!45}\cdashlinelr{1-6}
        
        \cellcolor{violet!5}& \multirow{2}{*}{$\sup\limits_{\alpha \in [1,2]} \frac{\alpha-1}{\alpha}\left[ R -  \widetilde{I}^{\text{Aug}}_{\alpha}(p;\mathscr{N}) \right]$}
        & \multirow{3}{*}{\textbf{?}}
        &
        & \multirow{3}{*}{\textbf{?}}
        & \multirow{2}{*}{$\sup\limits_{\alpha \in [0, 1]} (1-\alpha)\left[ {I}^{\text{Aug}}_{\alpha}(p;\mathscr{N}) -R \right]$}

        \\

        \cellcolor{violet!5}& & & & &
        
        \\

        \multirow{-3}{3.2cm}{\cellcolor{violet!5}\centering Classical-quantum soft covering, c.~c. (trace distance)} & \cite{CG22} & & & & \cite{CG22}
        
        \\

        \arrayrulecolor{black!45}\cdashlinelr{1-6}

        \cellcolor{violet!5}& \multirow{2}{*}{$\sup\limits_{\alpha \in [1,2]} \frac{\alpha-1}{2}\left[ R -  \widetilde{D}_{\alpha}(\rho_{\X\B}\Vert \rho_{\X}\otimes \rho_{\B}) \right]$}
        & \multirow{3}{*}{\textbf{?}}
        &
        & \multirow{3}{*}{\textbf{?}}
        & \multirow{3}{*}{\textbf{?}}

        \\

        \cellcolor{violet!5}& & & & &
        
        \\

        \multirow{-3}{3.5cm}{\cellcolor{violet!5}\centering Classical-quantum soft \!covering \!(purified distance)\!\!\!} & \cite{sharp25} & & & & 
        
        \\

        \arrayrulecolor{black!45}\cdashlinelr{1-6}
        
        \cellcolor{violet!5}& \multirow{2}{*}{$\sup\limits_{\alpha \in [1,2]} \frac{\alpha-1}{\alpha}\left[ R -  \inf_{\sigma_{\B}} \widetilde{D}_{\alpha}(\rho_{\A\B}\Vert \tau_{\A}\otimes \sigma_{\B}) \right]$}
        & \multirow{3}{*}{\textbf{?}}
        &
        & \multirow{3}{*}{\textbf{?}}
        & \multirow{2}{*}{$\sup\limits_{\alpha \in [0,1]} \!(1-\alpha)\left[   D_{\alpha}(\rho_{\A\B}\Vert \tau_{\A}\otimes \rho_{\B} ) \!-\! R \right]$}

        \\

        \cellcolor{violet!5}& & & & &
        
        \\

        \multirow{-3}{3cm}{\cellcolor{violet!5}\centering Convex splitting (trace distance)}& \cite{CG22b} & & & & \cite{CG22b}
        
        \\

        \arrayrulecolor{black!45}\cdashlinelr{1-6}
        
        \cellcolor{violet!5}& \multirow{2}{*}{$\sup\limits_{\alpha \in [1,2]} \frac{\alpha-1}{2}\left[ R -  \widetilde{D}_{\alpha}(\rho_{\A\B}\Vert \tau_{\A}\otimes \rho_{\B}) \right]$}
        & \multirow{3}{*}{\textbf{?}}
        &
        & \multirow{3}{*}{\textbf{?}}
        & \multirow{3}{*}{\textbf{?}}

        \\

        \cellcolor{violet!5}& & & & &
        
        \\

        \multirow{-3}{3cm}{\cellcolor{violet!5}\centering Convex splitting \!\!(purified distance)\!}& \cite{sharp25} & & & & 
        
        \\

        \arrayrulecolor{black!45}\cdashlinelr{1-6}
        
        \cellcolor{violet!5}& \multirow{2}{*}{$\sup\limits_{\alpha \in [1,2]} \frac{\alpha-1}{\alpha}\left[  \widetilde{H}_{\alpha}(\X\mid\E)_{\rho} - R \right]$}
        & \multirow{3}{*}{\textbf{?}}
        &
        & \multirow{3}{*}{\textbf{?}}
        & \multirow{3}{*}{$\sup\limits_{\alpha \in [0, 1]} (1-\alpha) \left[ R - H_{\alpha}^{\downarrow}(\X\mid\E)_{\rho} \right]$}

        \\

        \cellcolor{violet!5}& & & & &
        
        \\

        \multirow{-3}{3.5cm}{\cellcolor{violet!5}\centering Privacy~amplification against QSI (trace distance)}& \cite{Dup21,CDG24,regula2026rethinking} & & & & \cite{SGC22a}
        
        \\

        \arrayrulecolor{black!45}\cdashlinelr{1-6}

        \cellcolor{violet!5}& \multirow{3}{*}{$\sup\limits_{\alpha \in [1,2]} \frac{\alpha-1}{2}\left[  \widetilde{H}_{\alpha}^{\downarrow}(\X\mid\E)_{\rho} - R \right]$}
        & \multirow{3}{*}{$\sup\limits_{\alpha>1} \frac{\alpha-1}{2}\left[ \widetilde{H}_{\alpha}^{\downarrow}(\X\mid\E)_{\rho} - R \right]$}
        &
        & \multicolumn{2}{c}{\multirow{3}{*}{$\sup\limits_{\alpha \in [\nicefrac{1}{2}, 1]} \frac{1-\alpha}{\alpha} \left[ R - \widetilde{H}_{\alpha}(\X\mid\E)_{\rho} \right]$ (squared fidelity)}}

        \\

        \cellcolor{violet!5}& & & & &

        \\

        \cellcolor{violet!5}& & & & &
        
        \\

         \multirow{-4}{3.4cm}{\cellcolor{violet!5}\centering Privacy~amplification against QSI (purified distance) (non-composable)}& \cite{Hay15_PA, KL21, sharp25} & \cite{KL21}& & \multicolumn{2}{c}{\cite{LY24a}} 
        
        \\

        \arrayrulecolor{black!45}\cdashlinelr{1-6}

        \cellcolor{violet!5}& \multirow{2}{*}{$\sup\limits_{\alpha \in [1,2]} \frac{\alpha-1}{2}\left[  \widetilde{H}_{\alpha}^{\downarrow}(\X\mid\E)_{\rho} - R \right]$}
        & \multirow{2}{*}{$\sup\limits_{\alpha>1} \frac{\alpha-1}{2}\left[\widetilde{H}_{\alpha}^{\downarrow}(\X\mid\E)_{\rho} - R \right]$}
        &
        & \multicolumn{2}{c}{\multirow{2}{*}{$\sup\limits_{\alpha \in [\nicefrac{1}{2}, 1]} \frac{1-\alpha}{\alpha} \left[ R - \widetilde{H}_{\alpha}^{ \frac{1-2\alpha}{1-\alpha} }(\X\mid\E)_{\rho} \right]$ (squared fidelity)}}

        \\

        \cellcolor{violet!5}& & & & &
        
        \\

         \multirow{-3}{3.4cm}{\cellcolor{violet!5}\centering Privacy~amplification against QSI (purified distance)}& \cite{Hay15_PA, KL21, sharp25} & \cite{KL21}& & \multicolumn{2}{c}{\cite{rubboli2026sc}} 
        
        \\

        \arrayrulecolor{black!45}\cdashlinelr{1-6}
        
        \cellcolor{violet!5}& \multirow{2}{*}{$\sup\limits_{\alpha \in [1,2]} \frac{\alpha-1}{\alpha}\left[ 2R - \log|\A| + \widetilde{H}_{\alpha}(\A\!\mid\!\E)_{\rho} \right]$}
        & \multirow{3}{*}{\textbf{?}}
        &
        & \multirow{3}{*}{\textbf{?}}
        & \multirow{2}{*}{$\sup\limits_{\alpha \in [0, 1]} (1\!-\!\alpha) \left[ \log |\A|  - H_{\alpha}(\A\!\mid\!\E)_{\rho}\! - 2R \right]$}

        \\

        \cellcolor{violet!5}& & & & &

        \\

        \multirow{-3}{3cm}{\cellcolor{violet!5}\centering Quantum\! decoupling \,\,\,(trace distance)}& \cite{Dup21,CDG24,regula2026rethinking} & & & & \cite{CDG24}
        
        \\

        \arrayrulecolor{black!45}\cdashlinelr{1-6}

        \cellcolor{violet!5}& \multirow{2}{*}{$\sup\limits_{\alpha\in[1,2]} \frac{\alpha-1}{2}\left[ 2R - \widetilde{I}_{\alpha}(\E\!:\!\A)_{\rho}  \right]$}
        & \multirow{2}{*}{$\sup \limits_{\alpha>1} \frac{\alpha-1}{2}\left[ 2R- \widetilde{I}_{\alpha}(\E\!:\!\A)_{\rho}  \right]$}
        &
        & \multicolumn{2}{c}{\multirow{2}{*}{$\sup\limits_{\alpha \in[\nicefrac{1}{2},1] } \!\!\frac{1-\alpha}{\alpha}\!\left[  \widetilde{I}^{\downarrow\downarrow,\text{reg}}_{\alpha}(\E\!:\!\A)_{\rho} \!-\! 2R \right]$ (squared fidelity)}}

        \\

        \cellcolor{violet!5}& & & & &
        
        \\

        \multirow{-3}{3.2cm}{\cellcolor{violet!5}\centering Catalytic quantum decoupling (purified distance)} & \cite{LY21a, sharp25} & \cite{LY21a}& & \multicolumn{2}{c}{\cite{LY24b}} 
        
        \\

        \arrayrulecolor{black!45}\cdashlinelr{1-6}

        \cellcolor{violet!5}& \multirow{2}{*}{$\sup\limits_{\alpha\in[1,2]} \frac{\alpha-1}{2}\left[ 2R - \log |\A| + \widetilde{H}_{\alpha}^{\downarrow}(\A\!\mid\!\E)_{\rho}  \right]$}
        & \multirow{2}{*}{$\sup \limits_{\alpha>1} \frac{\alpha-1}{2}\left[ 2R - \log |\A| + \widetilde{H}_{\alpha}^{\downarrow}(\A\!\mid\!\E)_{\rho} \right]$}
        &
        & \multirow{2}{*}{\textbf{?}} & \multirow{2}{*}{\textbf{?}}

        \\

        \cellcolor{violet!5}& & & & &
        
        \\

        \multirow{-3}{3.2cm}{\cellcolor{violet!5}\centering Quantum decoupling (purified distance)} & \cite{LY21a, sharp25} & \cite{LY21a}& & \multicolumn{2}{c}{\cite{LY24b}} 
        
        \\

        \arrayrulecolor{black!45}\cdashlinelr{1-6}

        \cellcolor{violet!5}& \multicolumn{2}{c}{\multirow{2}{*}{{\color{white}tttttttt}$\sup\limits_{\alpha>1} \frac{\alpha-1}{2}\left[ R - \widetilde{I}_{\alpha}(\mathscr{N}_{\X\to\B})  \right]$ \;(\text{no critical rate})}}
        &
        & \multicolumn{2}{c}{\multirow{2}{*}{$\sup\limits_{\alpha\in[\nicefrac{1}{2},1]} \frac{1-\alpha}{\alpha}\left[  \widetilde{I}_{\alpha}(\mathscr{N}_{\X\to\B}) -R \right]$}}

        \\

        \cellcolor{violet!5}& & & & &
        
        \\

        \multirow{-3}{3.2cm}{\cellcolor{violet!5}\centering EA/NS c-q channel simulation (purified \!distance)} & \cite{LY21b, AYB24} & \cite{AYB24} & &  \multicolumn{2}{c}{\cite{AYB24}}

        \\

        \arrayrulecolor{black!45}\cdashlinelr{1-6}

        \cellcolor{violet!5}& \multirow{2}{*}{$\sup\limits_{\alpha\in[1,2]} \frac{\alpha-1}{2}\left[ R - \widetilde{I}_{\alpha}(\mathscr{N}_{\A\to\B})  \right]$}
        & \multirow{2}{*}{$\sup\limits_{\alpha>1} \frac{\alpha-1}{2}\left[ R - \widetilde{I}_{\alpha}(\mathscr{N}_{\A\to\B})  \right]$}
        &
        & \multirow{3}{*}{\textbf{?}}
        & \multirow{3}{*}{\textbf{?}}

        \\

        \cellcolor{violet!5}& & & & &
        
        \\

        \multirow{-3}{3.2cm}{\cellcolor{violet!5}\centering EA quantum channel simulation (purified distance)} & \cite{LY21b, sharp25} & \cite{LY21b}& &  &

        \\

        \arrayrulecolor{black!45}\cdashlinelr{1-6}

        \cellcolor{violet!5}& \multirow{3}{*}{$\sup\limits_{\alpha\in[\nicefrac{1}{2},1)}\!\frac{1-\alpha}{2\alpha}\!\left[ \lim\limits_{n\to\infty} \frac{1}{n} \!\sup\limits_{\mathscr{L}:\A^n:\B^n\to\mathsf{C}:\mathsf{D} } {I}_{\alpha} (\mathsf{C}\rangle\mathsf{D})_{\mathscr{L}(\rho_{\A\B}^{\otimes n})} - R \right]$}
        & \multirow{4}{*}{\textbf{?}}
        &
        & \multirow{4}{*}{\textbf{?}}
        & \multirow{4}{*}{\textbf{?}}

        \\

        \cellcolor{violet!5}& & & & &
        
        \\

        \cellcolor{violet!5}& & & & &
        
        \\

        \multirow{-4}{3.2cm}{\cellcolor{violet!5}\centering One-way LOCC entanglement distillation} & \cite{berta2026any-shot} & & &  &

        \\

        \arrayrulecolor{black!45}\cdashlinelr{1-6}

        \cellcolor{violet!5}& \multirow{3}{*}{$\sup\limits_{\alpha\in[\nicefrac{1}{2},1)}\! \frac{1-\alpha}{2\alpha}\!\left[ \lim\limits_{n\to\infty} \frac{1}{n} \!\sup\limits_{\phi_{\underline{\A}^n \A^n }} {I}_{\alpha}(\underline{\A}^n\rangle\B^n)_{\mathscr{N}^{\otimes n}(\phi_{\underline{\A}^n \A^n })} \!-\! R \right]$}
        & \multirow{3}{*}{\textbf{?}}
        &
        & \multirow{4}{*}{\textbf{?}}
        & \multirow{4}{*}{\textbf{?}}

        \\

        \cellcolor{violet!5}& & & & &
        
        \\

        \cellcolor{violet!5}& & & & &
        
        \\

        \multirow{-4}{3.2cm}{\cellcolor{violet!5}\centering Unassisted quantum communication} & \cite{berta2026any-shot} & & &  &

        \\
        
        \arrayrulecolor{black}\bottomrule
        \end{tabular}
	}
	\caption{
    Summary of the exponent analysis for various quantum covering-type problems (highlighted in \colorbox{violet!10}{violet})
    in quantum information theory.
    All quantities with ``$\widetilde{{\color{white}tt}}$'' are defined with respect to the sandwiched \Renyi divergence \cite{MDS+13,WWY14}.
    The rate for privacy amplification (resp.~quantum decoupling) is with respect to the remaining (resp.~discarded) system.
    Note the following different notions of conditional \Renyi entropies:
    $H_{\alpha}^{\downarrow}(\A\!\mid\!\E)_{\rho} \coloneqq - {D}_{\alpha}(\rho_{\A\E}\Vert\!\I_{\A}\otimes \rho_{\E})$ (and its sandwiched version is defined similarly);
    $\widetilde{H}_{\alpha}^{\lambda}(\A\!\mid\!\E)_{\rho}$ was defined in \cite{rubboli2024lambda} (see also \cite{li2025two-parameter}).
    Assistance with entanglement (resp.~non-signaling) is abbreviated by EA (resp.~NS).
    We omit expurgated bounds because of space limitations.
    We denote by the ``\textbf{?}'' results that are not yet established to the best of our knowledge.
	}	\label{table:survey_covering}	
\end{table}

\newpage
\section{Quantum Hypothesis Testing and Tilting} \label{sec:QHT}

\emph{Quantum hypothesis testing}, or equivalently, \emph{quantum state discrimination}, lies at the core of quantum information science, as it serves as a primitive tool for characterizing various fundamental quantum information-processing tasks.
The problem setup is the following.
The state of the quantum system is modeled by a density operator $\rho_{\B}^x \in \mathcal{S}(\B)$ on a Hilbert space $\mathcal{H}_{\B}$ with prior probability $p_{\X}(x)$.
The aim is to guess the true index $x \in \X$ with high probability by performing quantum measurements, which are modeled by positive operator-valued measures (POVMs) $\{ \Lambda_{\B}^x \}_{x\in\X}$, i.e., $\Lambda_{\B}^x \geq 0$ and $\sum_{x\in\X} \Lambda_{\B}^x = \I_{\B}$.
The minimum error probability of discrimination among the states is given by the optimization:
\begin{align}
    \varepsilon(\X\mid\B)_{\rho}
    \coloneqq1 -
    \sup_{ \text{POVM } \{ \Lambda_{\B}^x \}_{x\in\X} } \sum_{x\in\X} p_{\X}(x) \Tr\left[ \rho_{\B}^x \Lambda_{\B}^x \right],
\end{align}
where $\rho_{\X\B} = \sum_{x\in\X} p_{\X}(x) |x\rangle \langle x |\otimes \rho_{\B}^x$ is the joint classical-quantum state that describes the problem.

From the early developments by Holevo and Helstrom \cite{Hel67, Hol72, Hol78}, it is well known that the above optimization can be expressed as a semidefinite program \cite{YKL75, KRS09, AM14, Wat18}.
Unfortunately, closed-form expressions for the optimal measurements and the maximum success probability are generally unavailable, except for binary hypotheses (i.e., $|\X|=2$).
This makes the performance analysis of quantum state discrimination difficult.
Moreover, the computational cost of finding the optimal performance becomes quite high even for dozens of qubits, which makes practical implementations challenging.\footnote{
If the underlying states $\{\rho_{\B}^x\}_{x\in\X}$ are $n$-fold identical product states, then there is an efficient algorithm for computing the optimal success probabilities by employing the state symmetry; see e.g.~\cite[\S B of Supplementary Material]{sample_complexity_25}.
}

Recent work on quantum hypothesis testing often uses suboptimal measurements.
One important example is the conventional pretty-good measurement (PGM) with respect to $\left\{ p_{\X}(x), \rho_{\B}^x \right\}_{x\in\X}$ \cite{Bel75, HW94}:
\begin{align}
    {\Pi}_{\B}^{x} \coloneqq
    \left( {\rho}_{\B} \right)^{-\nicefrac{1}{2}} p_{\X}(x)\rho_{\B}^x  \left( {\rho}_{\B} \right)^{-\nicefrac{1}{2}}
    + p_{\X}(x) \I_{\B^{\perp}}, \quad \forall\, x\in \X,
\end{align}
where $\rho_{\B} = \sum_{x\in\X} p_{\X}(x) \rho_{\B}^x$ is the marginal state on system $\B$ of $\rho_{\X\B}$, and $\I_{\B^{\perp}}$ denotes the projection onto the kernel of the ensemble, i.e., $p_{\X}(x)\rho_{\B}^x \I_{\B^{\perp}} = 0$ for all $x\in\X$.
In the commuting case, PGM corresponds to the \emph{stochastic likelihood decoder}, in which the probability of deciding a hypothesis (given an observation) is proportional to the \textit{a posteriori distribution}.
To ensure resolution of unity (i.e., $\sum_{x\in\X} \Pi_{\B}^x = \I_{\B}$), PGM includes the sandwiched term $({\rho}_{\B})^{-\nicefrac{1}{2}}\cdot ({\rho}_{\B})^{-\nicefrac{1}{2}}$, which acts as the denominator of a quotient.
One may also define a PGM by replacing each positive
operator by its $\alpha$-power via functional calculus, i.e., $A\leftarrow A^{\alpha}$ and $B\leftarrow B^{\alpha}$, $\alpha > 0$; we call this the ``$\alpha$-PGM''.
PGMs are useful in quantum information theory because the resulting suboptimal error is at most twice the optimal error $\varepsilon(\X\mid\B)_{\rho}$ \cite{BK02}, \cite[Theorem 3.10]{Wat18}.
We refer the reader to~\cite{Cheng_simple} for more details.

Due to noncommutativity, there is no unique way to define the \emph{quotient} in the quantum setting.
To the best of our knowledge, Lieb observed that the directional derivative of the logarithmic function at a positive definite operator $X>0$ in a Hermitian direction $Y = Y^{\dagger}$ also serves as a noncommutative quotient~\cite{Lie73}:
\begin{align} \label{eq:quotient}
\frac{Y}{X} \equiv \mathrm{D} \log [X] (Y) 
\coloneqq\lim_{t\to 0} \frac{\log(X+tY) - \log X}{t}.
\end{align}
Hence, one may use it to define a variant of the PGM:
\begin{align} \label{eq:integral-PGM-multiple}  
    \mathring{\Pi}_{\B}^{x} \coloneqq\frac{ p_{\X}(x)\rho_{\B}^x  }{ \rho_{\B}  } + p_{\X}(x)\I_{\B^{\perp}}, \quad \forall\, x\in \X.
\end{align}
Some integral representations are known for the quotient $\frac{Y}{X}$ (Fact~\ref{item:Dlog_integral}).
We therefore call \eqref{eq:integral-PGM-multiple} the ``integral PGM'' (and the ``integral $\alpha$-PGM'' for the $\alpha$-powered version).
See Appendix~\ref{sec:log} for more details.
Recently, Beigi and Tomamichel derived several useful properties and promoted the use of \eqref{eq:integral-PGM-multiple} \cite{BT24}; see also \cite{Umlaut, Qumlaut}.

The success probabilities using the conventional PGM $\{{\Pi}_{\B}^{x}\}_{x\in\X}$ and the integral PGM $\{\mathring{\Pi}_{\B}^{x}\}_{x\in\X}$ can be written as
\begin{align}
    \sum_{x\in \X} p_{\X}(x) \Tr\left[  \rho_{\B}^x \cdot \Pi_{\B}^x \right]
    &= \widetilde{Q}_2 \left( \rho_{\X\B} \Vert \I_{\X}\otimes \rho_{\B} \right);
    \\
    \sum_{x\in \X} p_{\X}(x) \Tr\left[  \rho_{\B}^x \cdot \mathring{\Pi}_{\B}^x \right]
    &= \mathring{Q}_2 \left( \rho_{\X\B} \Vert \I_{\X}\otimes \rho_{\B} \right),
\end{align}
where 
\begin{align}
    \widetilde{Q}_2(A\Vert B) \coloneqq\Tr\left[ \left( B^{-\nicefrac{1}{4}} A B^{-\nicefrac{1}{4}} \right)^2 \right],
    \quad
    \mathring{Q}_2(A\Vert B) \coloneqq\Tr\left[ A \frac{A}{B} \right].
\end{align}
The former is called the \emph{quasi-collision divergence} \cite{Ren05}, and the latter is defined similarly using the logarithmic-derivative quotient (see e.g.~\cite{HP13}).

It is known in the quantum information geometry community that $\widetilde{Q}_2 \geq \mathring{Q}_2 $ \cite{LR99,TKR+10} (we provide a proof in Proposition~\ref{prop:PGMs} below for completeness), which implies that the error probability of using conventional PGM is no larger than that of the integral one.
In other words, a tight error estimate for the integral PGM also ensures performance guarantees for the conventional PGM.

\medskip
\noindent\fbox{\begin{minipage}{1\textwidth}
In practice, conventional PGMs may be adopted as they yield better performance and can be implemented by applying known quantum algorithms \cite{GLM+22}, whereas the integral PGMs serve as a convenient proxy for error analysis below.
\end{minipage}}
\smallskip

\begin{prop}[Relation between conventional and integral PGMs] \label{prop:PGMs}
For any $A\geq 0$ and $C>0$,
\begin{align}
    \Tr\left[ A C^{-\nicefrac{1}{2}} A C^{-\nicefrac{1}{2}} \right] \geq 
    \Tr\left[ A \cdot \mathrm{D} \log [C] (A) \right].
\end{align}

In particular, for any classical-quantum state $\rho_{\X\B}$,
\begin{align}
    \widetilde{Q}_2 \left( \rho_{\X\B} \Vert \I_{\X}\otimes \rho_{\B} \right)
    \geq 
    \mathring{Q}_2 \left( \rho_{\X\B} \Vert \I_{\X}\otimes \rho_{\B} \right).
\end{align}
\end{prop}
\begin{proof}
Let $C = \sum_i \lambda_i |i\rangle \langle i |$ be the spectral decomposition of $C$.
By inspection, we have $C^{-\nicefrac{1}{2}} A C^{-\nicefrac{1}{2}} = \sum_{i,j} \frac{1}{\sqrt{\lambda_i \lambda_j}} \langle i|A|j \rangle \cdot |i\rangle \langle j|$.
Then, the left-hand side is 
\begin{align}
\Tr\left[ A C^{-\nicefrac{1}{2}} A C^{-\nicefrac{1}{2}} \right] 
&= \sum_{i,j}  \frac{1}{ \sqrt{\lambda_i \lambda_j}  }  \left| \langle i |A|j\rangle \right|^2.
\end{align}

On the other hand, via Lieb's integral formula for $\mathrm{D} \log [\,\cdot\,](\,\cdot\,)$ (Fact~\ref{item:Dlog_integral}), we have
\begin{align}
\mathrm{D} \log [C] (A) 
=\sum_{i,j} \int_0^\infty \frac{\d t}{(\lambda_i+t)(\lambda_j + t) } \langle i|A|j \rangle \cdot |i\rangle \langle j|
=\sum_{i,j} \frac{\log \lambda_i - \log \lambda_j}{\lambda_i - \lambda_j }  \langle i | A | j \rangle \cdot |i\rangle \langle j|.
\end{align}
The right-hand side is then
\begin{align}
    \Tr\left[ A \cdot \mathrm{D} \log [C] (A) \right]
    &=
    \sum_{i,j}  \frac{\log \lambda_j - \log \lambda_i}{\lambda_j - \lambda_i }
    \left|    \langle i |A | j \rangle \right|^2.
\end{align}
The proof is concluded because the logarithmic mean is greater than the geometric mean (see e.g.~\cite[Example 5.22]{HP14}), and taking reciprocals reverses the inequality:
\begin{align}
    \frac{1}{\sqrt{\lambda_i \lambda_j} }  
    \geq
    \frac{\log \lambda_i - \log \lambda_j}{\lambda_i - \lambda_j },
    \quad \forall\, \lambda_i, \lambda_j >0
\end{align}
(for $\lambda_i = \lambda_j$, we write $\nicefrac{(\log \lambda_i - \log \lambda_j)}{(\lambda_i - \lambda_j)} \equiv \nicefrac{1}{\lambda_i}$, and both sides are equal.)
\end{proof}

\medskip
Next, let us turn our attention to binary hypothesis testing because a general recipe for analyzing one-shot quantum information processing is to reduce a task to the binary setting.
Suppose the null hypothesis and the alternative hypothesis are described by positive semidefinite operators $A$ and $B$, respectively.
Here, we consider a more general scenario without imposing the unit-trace constraint; one may think of $A$ and $B$ as density operators weighted by the prior probabilities (e.g., $p_{\X}(x) \rho_{\B}^x$).
The minimum error for distinguishing them is 
\begin{align}
\varepsilon^{\star}_{A \Vert B}
\coloneqq\inf_{ 0\leq T \leq \I} \Tr\left[A (\I-T)\right] + \Tr\left[ B T \right],
\end{align}
where the operator $0\leq T\leq \I$ is called a \emph{test} for deciding the null hypothesis.
Equivalently, $\{T, \I-T\}$ forms a two-outcome POVM.

Holevo and Helstrom showed that the minimum error is 
\begin{align} \label{eq:noncommutative_minimum}
    \varepsilon^{\star}_{A \Vert B}
    = \Tr\left[ A \wedge B \right]
    = \frac{ \Tr[A+B] }{2} - \frac{\left\| A - B \right\|_1}{2}
\end{align}
(here, $A \wedge B \coloneqq\argmax_{H=H^\dagger} \left\{\Tr[H]: H\leq A, H\leq B\right\}  = \frac12\left[ A+B-|A-B|\right]$ \cite{Cheng_simple}),
and the optimal test $T^{\star}_{A\Vert B}$ is achieved by the \emph{Holevo--Helstrom measurement}:\footnote{In the commuting case, the Holevo--Helstrom measurement is equal to the classical Neyman--Pearson test.}
\begin{align} \label{eq:HH}
    T^{\star}_{A\Vert B}
    = \left\{ A > B \right\} + \delta \cdot \left\{ A = B \right\},
\end{align}
where $\delta \in [0,1]$ is arbitrary and used to break ties at random.

Later, Audenaert \textit{et al.} proved that the optimal error admits a quantum Chernoff bound \cite[Theorem 1]{ACM+07}:
\begin{align} \label{eq:Chernoff}
\varepsilon^{\star}_{A \Vert B} \leq 2^{- (1-\alpha) D_{\alpha}\left(A\Vert B \right)}, \quad \forall\, \alpha \in [0,1],
\end{align}
in which the (non-normalized) order-$\alpha$ Petz--\Renyi divergence \cite{Pet86} is defined by 
\begin{align} \label{eq:Petz}
D_{\alpha}\left(A\Vert B \right)
\coloneqq\frac{1}{\alpha-1}\log_2 \Tr\left[ A^{\alpha} B^{1-\alpha}\right], \quad \forall\, \alpha \in (0,1)
\end{align}
for $A \not\perp B$ and is defined to be infinite otherwise; the order-$0$ and order-$1$ cases are defined via continuous extensions.
The Chernoff bound is remarkable because it applies to any $n$-fold independent and identically distributed (i.i.d.) scenario.
Namely, for $A\leftarrow p \rho^{\otimes n}$ and $B\leftarrow q \sigma^{\otimes n}$, where $n$ is the number of copies, $(p,q)$ are fixed ($n$-independent) prior probabilities, and $(\rho,\sigma)$ is a pair of states, then \eqref{eq:Chernoff} implies the following upper bound on the error probability for \emph{any} number of copies:
\[
\varepsilon^{\star}_{p \rho^{\otimes n} \Vert q \sigma^{\otimes n}} \leq 2^{- n \cdot\sup_{\alpha \in [0,1]} (1-\alpha) D_{\alpha}(\rho\Vert \sigma) }, \quad \forall\, n\in\mathds{N}.
\]
Later, the error exponent $\sup_{\alpha \in [0,1]} (1-\alpha) D_{\alpha}(\rho\Vert \sigma)$ was shown to be asymptotically optimal \cite{NS09, ANS+08}.
Moreover, the above results tell us that the prior probabilities contribute only an $n$-independent multiplicative factor $\sup_{\alpha\in[0,1]} p^{\alpha} q^{1-\alpha} \leq 1$ to $\varepsilon^{\star}$, and hence, they do not affect the error exponent $\sup_{\alpha \in [0,1]} (1-\alpha) D_{\alpha}(\rho\Vert \sigma)$.

We analyze the error probability using the integral PGM:
\begin{align} \label{eq:integral-PGM}
\left\{ \frac{A}{A+B}, \frac{B}{A+B} \right\}
\end{align}
because its error bound plays a fundamental role in characterizing various quantum packing-type problems below.
(Here, without loss of generality, we may restrict the Hilbert space to the joint support of $A$ and $B$ such that $A+B>0$.)

The main contribution of this section is to show that the integral PGM in \eqref{eq:integral-PGM}
admits an extremal decomposition in terms of Holevo--Helstrom measurements so that standard hypothesis-testing tools apply. 
The extremal decomposition is a special case of the \emph{operator layer cake theorem} for the derivatives of the operator logarithm (Theorem~\ref{theo:Dlog_formula}), detailed in Appendix~\ref{sec:layer-cake}.\footnote{
The extremal decomposition of an operator is in general not unique. 
For example, one immediately has
$\frac{B}{A+B} = \int_0^1 \{ u \I < \frac{B}{A+B} \} \, \d u$ by the classical layer cake representation \cite{LM01}.
However, the representation in Theorem~\ref{theo:extremal} will be needed for our purposes.
}

\begin{theo}[Extremal decomposition] \label{theo:extremal}
For all positive semidefinite operators $A$ and $B$ such that $A+B>0$, the following extremal decomposition holds for the quotient defined in \eqref{eq:quotient}:
\begin{align}
\frac{B}{A+B} 
=
\int_0^1 \left\{ u A < (1-u) B \right\} \d u.
\end{align}
\end{theo}
\begin{proof}
    Theorem~\ref{theo:Dlog_formula} with $A\leftarrow A+B$ and $B\leftarrow B$ shows that
    \begin{align}
        \frac{B}{A+B}
        = \int_0^{\infty} \proj{ u(A+B) < B} \d u
        - \int_{-\infty}^0 \proj{ u(A+B) > B} \d u.
    \end{align}
    Since $B\geq 0$ and $A+B>0$, we have $u(A+B) - B < 0$ for all $u<0$.
    Hence, the projection $\proj{u(A+B) > B}$ in the second integration is the zero operator.
    On the other hand, 
    $B - u(A+B) = (1-u) B - u A < 0$ for all $u > 1$.
    The projection $\proj{u(A+B) < B}$ in the first integration is the zero operator for all $u>1$, and hence we can restrict the integration interval from $[0,\infty)$ to $[0,1]$. 
\end{proof}

Theorem~\ref{theo:extremal} demonstrates an operational interpretation  of the integral PGM---one draws a parameter $u\in[0,1]$ uniformly at random and applies the Holevo--Helstrom measurement $\left\{ T_{uA\Vert(1-u)B}^{\star}, \I - T_{uA\Vert(1-u)B}^{\star} \right\}$ with the corresponding priors $(u,1-u)$ as illustrated in Figure~\ref{fig:Integral-PGM} below, i.e.,
\begin{align*}
\left\{ \frac{A}{A+B}, \frac{B}{A+B} \right\}
= \left\{  \int_{0}^1 T^{\star}_{uA\Vert (1-u)B} \d u, \int_{0}^1 \left( \I - T^{\star}_{uA\Vert (1-u)B} \right) \d u \right\}
= 
\int_{0}^1 \left\{ T^{\star}_{uA\Vert (1-u)B}, \I - T^{\star}_{uA\Vert (1-u)B} \right\} \! \d u.
\end{align*}

\begin{figure}[ht!]
	\centering
	\resizebox{1\columnwidth}{!}{
	\includegraphics{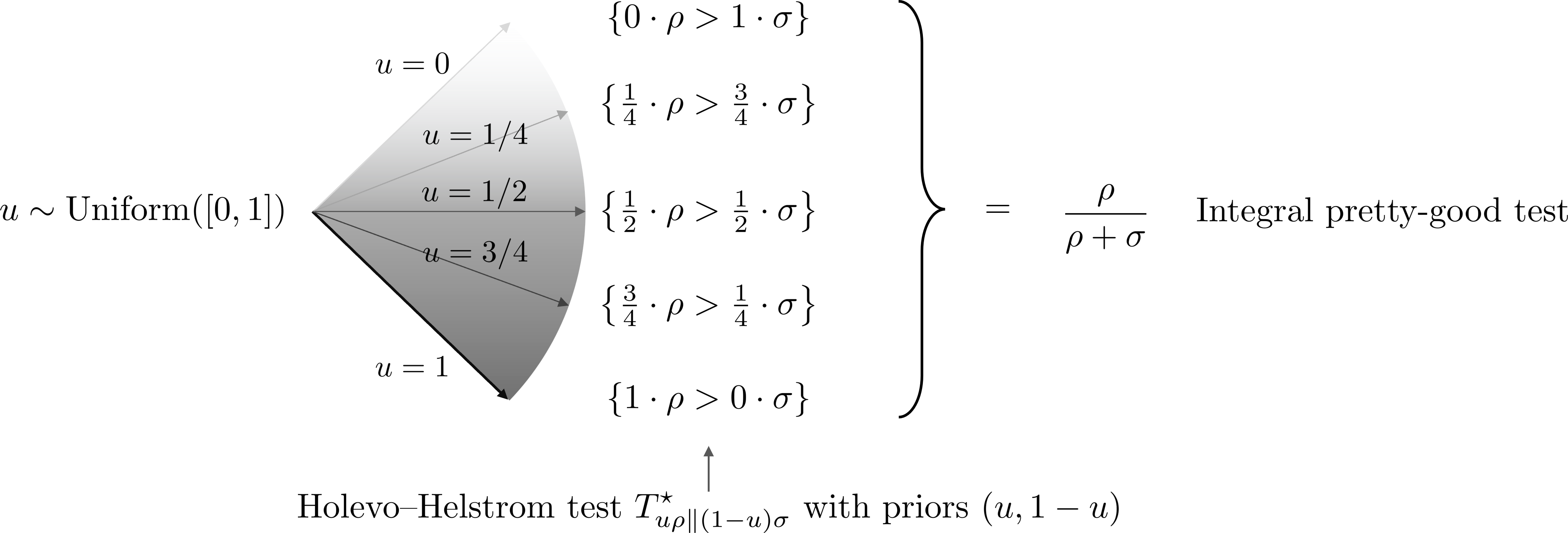}     	
	}
	\caption{\small
    A schematic illustration of the integral pretty-good test in \eqref{eq:integral-PGM} in terms of a randomized Holevo--Helstrom test \eqref{eq:HH} with the priors $(u,1-u)$, where $u$ is drawn uniformly at random.}
	\label{fig:Integral-PGM}
\end{figure}

\smallskip
\noindent\fbox{\begin{minipage}{1\textwidth}
The interpretation given in Theorem~\ref{theo:extremal} also provides an intuitive explanation of why both conventional and integral pretty-good measurements are pretty good (cf.~Proposition~\ref{prop:PGMs}).
Indeed, the Holevo--Helstrom measurement can achieve the optimal error exponent even under mismatched priors, because the wrong priors contribute only a constant multiplicative factor;
integrating the resulting constants over the integral PGM induces at most a constant multiplicative cost to the optimal error.
\end{minipage}}
\smallskip

With this observation, we obtain the following tilting inequality in Proposition~\ref{prop:key}, which serves as the main technical tool in our later analysis.

\begin{prop}[Tilting inequality] \label{prop:key}
	For all finite-dimensional positive semidefinite operators $A$ and $B$, the following holds:
	\begin{align}
	\Tr\left[ A \frac{B^{\alpha}}{A^{\alpha} + B^{\alpha}} \right]
	&\leq c_{\alpha} \cdot \Tr\left[ A^{\alpha} B^{1-\alpha} \right]
    \quad\forall\, \alpha \in [\nicefrac{1}{2},1],
	\end{align}
    where
    \begin{align} \label{eq:c_alpha}
    c_\alpha^{(1)} \coloneqq\frac{1-\alpha}{\alpha}\frac{\pi}{\sin\left(\frac{1-\alpha}{\alpha}\pi\right)}, 
    \quad
    c_\alpha^{(2)}\coloneqq\left( 2\alpha\right)^{-\frac{1}{\alpha}}
    \left(1-\frac{1}{2\alpha}\right)^{2-\frac{1}{\alpha}}
    \frac{\alpha}{1-\alpha}, 
    \quad
    c_{\alpha}\coloneqq\min \left\{c_\alpha^{(1)}, c_\alpha^{(2)} \right\}<1.102,
    \end{align}
    with the endpoint values understood by one-sided limits, namely
\(c_{\nicefrac{1}{2}}=c_{\nicefrac{1}{2}}^{(2)}\coloneqq1\) and
\(c_1=c_1^{(1)}\coloneqq1\).

\end{prop}

Figure~\ref{figure:c_alpha} numerically plots the multiplicative coefficients $c_\alpha^{(1)}$ and $c_\alpha^{(2)}$.

\begin{figure}[h]
\centering
\includegraphics[width=0.6\linewidth]{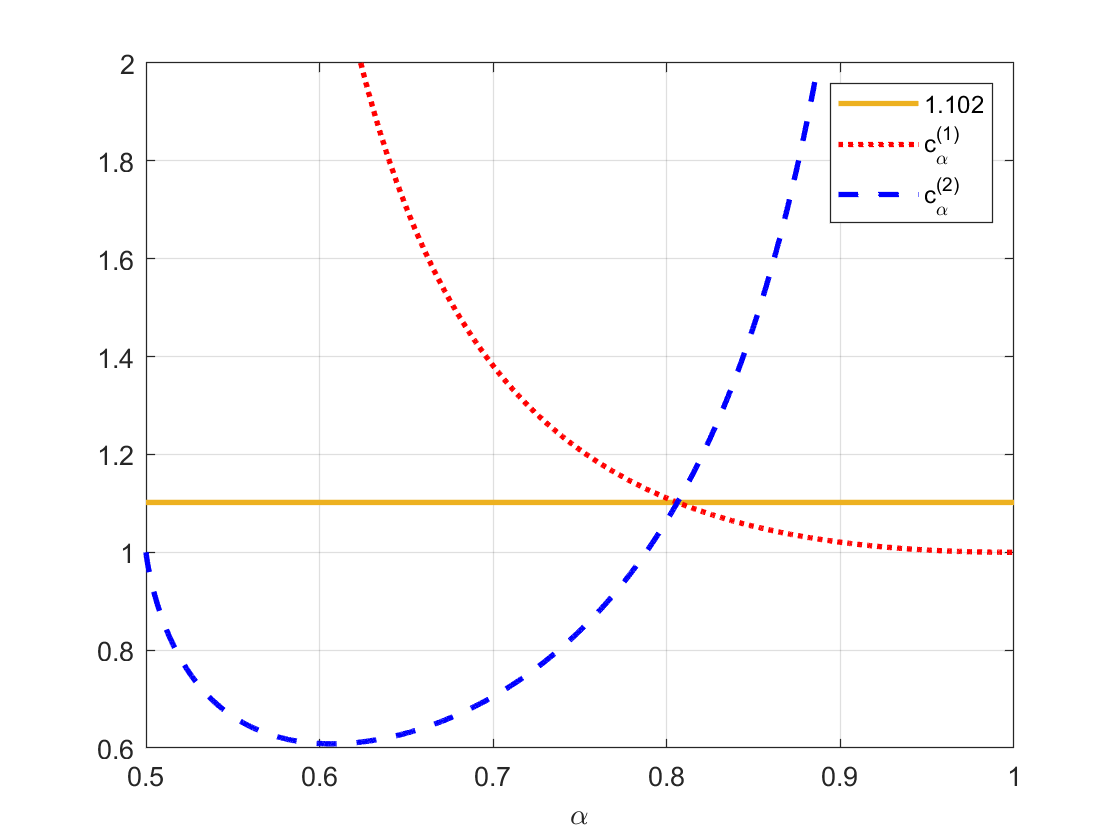}
	\caption{A numerical illustration of the coefficients $c_{\alpha}^{(1)}$ (dotted {\color{red}red} line) and $c_{\alpha}^{(2)}$ (dashed {\color{blue}blue} line) for $\alpha \in [\nicefrac{1}{2},1]$ in \eqref{eq:c_alpha} of Proposition~\ref{prop:key}.
    The maximum of the pointwise minimum $c_{\alpha}=\min\{c_{\alpha}^{(1)},c_{\alpha}^{(2)}\}$ over $\alpha\in[\nicefrac{1}{2},1]$ is $\approx 1.101911<1.102$.
	} \label{figure:c_alpha}
\end{figure} 

\begin{proof}
We will prove two bounds with different prefactors $c_{\alpha}^{(1)}$ and $c_{\alpha}^{(2)}$, respectively, and then take the smaller of the two, i.e.,~$c_{\alpha} =\min\left\{ c_{\alpha}^{(1)}, c_{\alpha}^{(2)}\right\}$.

We start with the first bound with the prefactor $c_{\alpha}^{(1)}$.
Via the extremal decomposition given in Theorem~\ref{theo:extremal}, we have, for any $\alpha \in (\nicefrac{1}{2},1]$,
\begin{align} \notag
		\Tr\left[ A \frac{B^{\alpha}}{A^{\alpha} + B^{\alpha}} \right]
		&= \int_{0}^{1} \Tr\left[ A \left\{  A^{\alpha} < \frac{1-u}{u} B^{\alpha} \right\}  \right] \mathrm{d} u
		\\
		&{\leq} \int_{0}^{1} \left(\frac{1-u}{u}\right)^{\frac{1-\alpha}{\alpha}} \Tr\left[ A^{\alpha} B^{1-\alpha} \left\{ A^{\alpha} < \frac{1-u}{u} B^{\alpha} \right\}\right] \mathrm{d} u, \quad \forall \,\alpha \in (\nicefrac{1}{2}, 1],
\end{align}
where we employ the inequality of Audenaert \textit{et al.} (Lemma~\ref{lemm:Audenaert} below).
The prefactor can be calculated by the residue theorem:
\begin{align} \notag
    c_{\alpha}^{(1)} 
    \coloneq
    \int_{0}^{1} \left(\frac{1-u}{u} \right)^{\frac{1-\alpha}{\alpha}} \d u
    =\int_0^\infty \frac{v^\frac{1-\alpha}{\alpha}}{\left(v+1\right)^2}\,\d v
    = \frac{1-\alpha}{\alpha}\frac{\pi}{\sin\left(\frac{1-\alpha}{\alpha}\pi\right)}\quad \forall \,\alpha \in (\nicefrac{1}{2}, 1)
\end{align} and $c_{1}^{(1)} \coloneqq \lim_{\alpha \nearrow 1} c_{\alpha}^{(1)} = 1$.
Furthermore, let $v = \frac{1-u}{u} \geq 0 $.
By the cyclic property of trace and its invariance under conjugate transpose,
\begin{align*}
\Tr\left[ A^{\alpha} B^{1-\alpha} \left\{ A^{\alpha} < v B^{\alpha} \right\} \right]
&= \Tr\left[ B^{1-\alpha}  A^{\alpha} \left\{ A^{\alpha} < v B^{\alpha} \right\} \right]
\\
&=   \Tr\left[ B^{1-\alpha} (A^{\alpha} - v B^{\alpha})\left\{ A^{\alpha} < v B^{\alpha} \right\} \right] + 
v\Tr\left[ B \left\{ A^{\alpha} < v B^{\alpha} \right\} \right]
\\
&\leq \Tr\left[ B^{1-\alpha} (A^{\alpha} - v B^{\alpha}) \right] +
v\Tr\left[ B \left\{ A^{\alpha} < v B^{\alpha} \right\} \right]
\\
&= \Tr\left[ A^{\alpha} B^{1-\alpha}\right] - v\Tr\left[ B \left( \I - \left\{  A^{\alpha} < vB^{\alpha}\right\}\right) \right]
\\
&\leq \Tr\left[ A^{\alpha} B^{1-\alpha}\right],
\end{align*}
proving the first bound with prefactor $c_{\alpha}^{(1)}$.

Next, we proceed to the second bound with the prefactor $c_{\alpha}^{(2)}$ for $\alpha\in [\nicefrac{1}{2},1) $.
For $t>0$, let \[X(t)=\frac{\I}{\sqrt{t\I+A^\alpha}}B^\alpha\frac{\I}{\sqrt{t\I+A^\alpha}},\] then by Lieb's integral formula,
\begin{align}
\Tr\left[A\frac{B^{\alpha}}{A^{\alpha} + B^{\alpha}}\right] \notag
&=\Tr\left[A \int_0^\infty\!\! \frac{\I}{t\I+A^\alpha+B^\alpha}B^\alpha\frac{\I}{t\I+A^\alpha+B^\alpha}\d t\right]\notag
\\
&\overset{\textnormal{(a)}}{=}\Tr\left[A \int_0^\infty\!\! \frac{\I}{\sqrt{t\I+A^\alpha}}\frac{\I}{\I+X(t)}X(t)\frac{\I}{\I+X(t)}\frac{\I}{\sqrt{t\I+A^\alpha}}\d t\right]\notag
\\
&\overset{\textnormal{(b)}}{\le} \kappa_\alpha\Tr\left[A \int_0^\infty \frac{\I}{t\I+A^\alpha}X(t)^{\frac1\alpha-1}\d t\right]\notag
\\
&=\kappa_\alpha\Tr\left[ \int_0^\infty \frac{A}{t\I+A^\alpha}\left( \frac{1}{\sqrt{t\I+A^\alpha}} B^\alpha \frac{1}{\sqrt{t\I+A^\alpha}}\right)^{\frac{1}{\alpha}-1}\d t\right] \label{eq:before_Araki}
\\
&\overset{\textnormal{(c)}}{\leq} \kappa_\alpha\Tr\left[ \int_0^\infty \frac{A}{t\I+A^\alpha}\left(B^\alpha\right)^{\frac{1}{\alpha}-1}\left(\frac{\I}{t\I+A^\alpha}\right)^{\frac{1}{\alpha}-1}\d t\right]\notag
\\
&=\kappa_\alpha\Tr\left[AB^{1-\alpha} \int_0^\infty \left(\frac{\I}{t\I+A^\alpha}\right)^{\frac{1}{\alpha}}\d t\right]\notag
\\
&=\kappa_\alpha\frac{\alpha}{1-\alpha}\Tr\left[A^\alpha B^{1-\alpha}\right],
\end{align}
where we denote $\kappa_\alpha=(2\alpha)^{-\frac{1}{\alpha}}(1-\frac{1}{2\alpha})^{2-\frac{1}{\alpha}}$ and write $\kappa_{\nicefrac{1}{2}} \coloneqq  \lim_{\alpha \searrow \nicefrac{1}{2}} \kappa_{\alpha} =  1$.
In (a), we used the identity 
\begin{align}
\frac{\I}{t\I+A^\alpha+B^\alpha}=\frac{\I}{\sqrt{t\I+A^\alpha}}\frac{\I}{\I+X(t)}\frac{\I}{\sqrt{t\I+A^\alpha}}.
\end{align}
The first inequality (b) comes from Young's inequality: for any $v\geq 0$ and $\alpha \in (\nicefrac{1}{2},1]$,
\begin{align}
1+v
&= \frac{1}{2\alpha}\cdot2\alpha+\left(1-\frac{1}{2\alpha}\right)\left(\frac{1}{1-\frac{1}{2\alpha}}v\right)\\
&\geq (2\alpha)^{\frac{1}{2\alpha}} \left(\frac{1}{1-\frac{1}{2\alpha}}\right)^{1-\frac{1}{2\alpha}}\cdot v^{1-\frac{1}{2\alpha}},
\end{align}
which implies that
\begin{align} \notag
\frac{v}{(1+v)^2}\leq(2\alpha)^{-\frac{1}{\alpha}}\left(1-\frac{1}{2\alpha}\right)^{2-\frac{1}{\alpha}}\cdot v^{\frac{1}{\alpha}-1}
=\kappa_\alpha v^{\frac{1}{\alpha}-1} 
\end{align}
for $\alpha \in (\nicefrac{1}{2},1]$ and $\frac{v}{(1+v)^2} \leq v$ for $\alpha = \nicefrac{1}{2}$.

At this stage, the expression \eqref{eq:before_Araki} is the desired bound with prefactor $c_{\alpha}^{(2)} \coloneqq\kappa_{\alpha} \cdot \frac{\alpha}{1-\alpha}$ in the commuting scenario, i.e., provided that $A$ and $B$ commute.
Hence, it remains to handle the noncommutativity via an Araki-type inequality,
which is precisely the role of step (c).
For (c), we fix $t>0$ and set a positive definite operator $C=(t\I+A^\alpha)^{-1}\leq t^{-1}\I$. Then $\frac{A}{t\I+A^\alpha}=g(C)$, for $g(x)=x^{1-1/\alpha}(1-tx)^{1/\alpha}$, with domain $(0,1/t]\supseteq \spec(C)$.
Since $s=\nicefrac{1}{\alpha}-1\in [0,1]$ and 
\begin{align}
x^sg(x)=(1-tx)^{\frac{1}{\alpha}}
\end{align}
is nonincreasing in $x$ on its domain $(0,1/t]$, the Araki-type inequality (Lemma~\ref{lemm:Araki} below with $X\leftarrow C$ and $Y\leftarrow B^\alpha$) gives
\begin{align}
\Tr\left[ \frac{A}{t\I+A^\alpha}\left( \frac{1}{\sqrt{t\I+A^\alpha}} B^\alpha \frac{1}{\sqrt{t\I+A^\alpha}}\right)^{\frac{1}{\alpha}-1}\right]
&=\Tr\left[ g(C)\left( C^{\nicefrac{1}{2}} B^\alpha C^{\nicefrac{1}{2}}  \right)^{\frac{1}{\alpha}-1}\right]
\\
&\leq \Tr\left[ g(C)(B^\alpha)^{\frac{1}{\alpha}-1}C^{\frac{1}{\alpha}-1} \right]
\\
&=\Tr\left[ \frac{A}{t\I+A^\alpha}\left(B^\alpha\right)^{\frac{1}{\alpha}-1}\left(\frac{1}{t\I+A^\alpha}\right)^{\frac{1}{\alpha}-1}\right],
\end{align}
which completes the proof.
\end{proof}

\begin{lemm}[Audenaert \textit{et al.}'s inequality {\cite[Lemma 1]{ACM+07},  \cite[Lemma 2]{ANS+08}}] \label{lemm:Audenaert}
	For positive semidefinite operators $A$ and $B$,
	\begin{align} \notag
		\Tr\left[ A^{1+s} \left\{ A< B \right\} \right]
		&\leq \Tr\left[ A B^s \left\{ A< B \right\} \right], \quad \forall\, s\in[0,1].
	\end{align}
	Equivalently,
	\begin{align}
        \label{eq:Audenaert0}
		\Tr\left[ A \left\{ A^\alpha < B^{\alpha} \right\} \right]
		&\leq \Tr\left[ A^{\alpha} B^{1-\alpha} \left\{ A^{\alpha} < B^{\alpha} \right\} \right],
        \quad \forall\, \alpha \in [\nicefrac{1}{2},1].
	\end{align}
\end{lemm}

\begin{lemm}[Araki-type inequality {\cite[Proposition 5]{LC_Araki}}] \label{lemm:Araki}
    Let $X,Y\geq 0$.
    For any $s\in[0,1]$ and any function $g$ on an interval $\mathcal{J}$, where $\spec(X)\subseteq \mathcal{J}$, such that $x\mapsto x^sg(x)$ is nonnegative and nonincreasing, we have 
\begin{align*}
\Tr\left[ g(X) \left(X^{\nicefrac{1}{2}} Y X^{\nicefrac{1}{2}} \right)^s \right]
\leq 
\Tr\left[ g(X) X^s Y^s \right].
\end{align*}
\end{lemm}


\section{Classical-Quantum Channel Coding} \label{sec:CQ}

In this section, we first establish a one-shot bound for classical-quantum channel coding (Theorem~\ref{theo:CQ}), which resolves Burnashev and Holevo's conjecture \cite{BH98}.
We extend the result to infinite dimensions in Section~\ref{sec:inf-dim}.
In Sections~\ref{sec:large_deviation}, \ref{sec:moderate_deviation}, and \ref{sec:small_deviation}, we demonstrate the corresponding asymptotic results in the large-deviation, moderate-deviation, and small-deviation regimes, respectively, and compare them with existing results.

\begin{defn}[Classical-quantum channel coding] \label{defn:CQ}
	Let $\mathscr{N}_{\X \to \B}: x\mapsto \rho_{\B}^x \in \mathcal{S}(\B)$ be a classical-quantum channel, where each channel output $\rho_{\B}^x$ is a density operator (i.e.,~a positive semidefinite operator with unit trace).
	\begin{enumerate}[1.]
		\item Alice has classical registers $\mathsf{M}$ and $\X$, and Bob has a quantum register $\B$.
		
		
		\item An encoder at Alice, $m \mapsto x(m)$, maps each (equiprobable) message in $\mathsf{M}$ to a codeword in $\X$.
		
		\item 
        Alice's codeword in $\X$ undergoes the classical-quantum channel $\mathscr{N}_{\X\to \B}$ and outputs a state on Bob's quantum register $\B$.
		
		\item Bob applies a decoding measurement described by a positive operator-valued measure (POVM) $\left\{ {\Lambda}_{\B}^{m} \right\}_{m\in\mathsf{M}}$ (i.e.,~$0\leq \Lambda_{\B}^{m} \leq \I_{\B}$ and $\sum_{m\in\mathsf{M}} \Lambda_{\B}^m = \I_{\B}$) on his  quantum register $\B$ to obtain an estimated message $\hat{m} \in \mathsf{M}$.
	\end{enumerate}
	We consider the conventional random coding strategy as follows.
    For each message $m\in \mathsf{M}$, a codeword $x({m})$ is drawn pairwise independently according to the common input distribution $p_{\X}$.
    Then, the expected \emph{random coding error probability} with optimal decoding for sending $|\mathsf{M}|$ messages through the channel $\mathscr{N}_{\X\to\B}$ with input distribution $p_{\X}$ is defined as\footnote{
    We only write $p$ in the subscript of $\varepsilon(\X:\B)_{p}$ in \eqref{eq:error_CQ} because the channel $\mathscr{N}_{\X\to\B} : x\mapsto \rho_{\B}^x$ is usually fixed in classical-quantum channel coding.
    Hence, the random coding error $\varepsilon(\X:\B)_{p}$ depends only on the input distribution $p_{\X}$.}
	\begin{align} \label{eq:error_CQ}
    \varepsilon(\X:\B)_{p}
    &\coloneqq   \mathds{E}_{\{x(m)\}\sim p_{\X}}  \left[ 
    \inf_{ \{ {\Lambda}_{\B}^{m} \}_{m\in\mathsf{M}} } \frac{1}{|\mathsf{M}|}  \sum_{m\in \mathsf{M}}	
     \Tr\left[ \rho_{\B}^{x(m)} \left(\I_{\B} -  {\Lambda}_{\B}^{m} \right) \right]     \right],
	\end{align}
    where the minimization is over all POVMs $\{ {\Lambda}_{\B}^{m} \}_{m\in\mathsf{M}}$, and
    the joint input-output state is denoted by 
    \begin{align} \label{eq:CQ-state}
    \rho_{\X\B} \coloneqq\sum_{x\in\X} p_{\X}(x)|x\rangle \langle x|_{\X} \otimes \rho_{\B}^x.
    \end{align}
\end{defn}

Given a realization of the random codebook $\left\{ x(1), x(2), \ldots, x(|\mathsf{M}|) \right\}$, we adopt the integral $\alpha$-PGM:
\begin{align}
    \begin{split} \label{eq:alpha-PGM}
    \mathring{\Pi}_{\B}^{x(m)} 
    \coloneqq\frac{\left( \rho_{\B}^{x({m})} \right)^{\alpha}}{\sum_{\bar{m}\in \mathsf{M} } \left( \rho_{\B}^{x(\bar{m})} \right)^{\alpha} }
    + \frac{1}{|\mathsf{M}|}\I_{\B^{\perp}}, \quad \forall \, m \in \mathsf{M}, \; \alpha \in [\sfrac{1}{2},1]
    \end{split}
\end{align}
according to the associated channel output states $\big\{ \rho_{\B}^{x(m)} \big\}_{m\in \mathsf{M}}$.
Here, $\I_{\B^{\perp}}$ denotes the projection onto the common kernel of channel outputs, i.e.,~$\rho_{\B}^{x(m)} \I_{\B^{\perp}} = 0$ for all $m\in\mathsf{M}$.

\begin{theo} \label{theo:CQ} 
	For any finite-dimensional classical-quantum channel $\mathscr{N}_{\X \to \B}: x \mapsto \rho_{\B}^x$ and a given input distribution $p_{\X}$, the random coding error probability \eqref{eq:error_CQ} for sending $|\mathsf{M}|$ messages is upper bounded by
	\begin{align}
		\varepsilon(\X:\B)_{p}
		&\leq c_{\alpha}
        (|\mathsf{M}|-1)^{\frac{1-\alpha}{\alpha}}
        \Tr \left[
        \left( \sum_{x\in\X} p_{\X}(x) \left(\rho_{\B}^x\right)^{\alpha}  \right)^{\nicefrac{1}{\alpha}}
        \right]
        \notag
        \\
        &= c_{\alpha} \cdot
        2^{
        -\frac{1-\alpha}{\alpha} \left[ I_{\alpha} (\X : \B)_{\rho} - \log_2 (|\mathsf{M}|-1) \right]
        }, \quad \forall\, \alpha \in [\nicefrac{1}{2},1]. 
        \label{eq:main}
	\end{align}
    Here, the joint state $\rho_{\X\B}$ is defined in \eqref{eq:CQ-state},
    $I_{\alpha} (\X:\B)_{\rho} \coloneqq\inf_{\sigma_{\B} \in \mathcal{S}(\B) } D_{\alpha}(\rho_{\X\B}\Vert \rho_{\X} \otimes \sigma_{\B}) $ is the order-$\alpha$ Petz--\Renyi information, and the factor $c_{\alpha}$ defined in \eqref{eq:c_alpha} admits a universal upper bound $1.102$.
\end{theo}
\begin{remark}[Maximal-error criterion]
While we only consider the average error criterion in this work, Theorem~\ref{theo:CQ} with the standard expurgation argument (e.g.~\cite[p.~140]{Gal68}) implies a maximal-error achievability bound with only a one-bit message-size loss and a factor-two error loss.
\end{remark}

The quantum Sibson identity \cite[Lemma II.6]{KW09}, \cite[(3.10)]{HT14}, \cite{CGH18} shows that the minimizer in $I_{\alpha} (\X:\B)_{\rho}$ is attained by
\[
\sigma_{\B}^{\star} = \frac{
        \left( \sum_{x\in\X} p_{\X}(x) \left(\rho_{\B}^x\right)^{\alpha}  \right)^{\nicefrac{1}{\alpha}}
        }{
        \Tr \left[
        \left( \sum_{x\in\X} p_{\X}(x) \left(\rho_{\B}^x\right)^{\alpha}  \right)^{\nicefrac{1}{\alpha}}
        \right]
        },
\]
and, hence, the order-$\alpha$ Petz--\Renyi information with respect to the state $\rho_{\X\B}$
admits a closed-form expression:
\[
I_{\alpha} (\X:\B)_{\rho} 
= D_{\alpha}\left(\rho_{\X\B}\Vert \rho_{\X} \otimes \sigma_{\B}^{\star}\right)
=
\frac{\alpha}{\alpha-1} \log_2 \Tr \left[
        \left( \sum_{x\in\X} p_{\X}(x) \left(\rho_{\B}^x\right)^{\alpha}  \right)^{\nicefrac{1}{\alpha}}
        \right].
\]

The proof of Theorem~\ref{theo:CQ} is given below and only uses the tilting inequality in Proposition~\ref{prop:key}.
\begin{proof}
By the symmetry of random coding, it suffices to calculate the error probability of sending $m=1$ without loss of generality: For all $\alpha \in [\nicefrac{1}{2},1]$,
\begin{align}\notag
    \varepsilon(\X:\B)_{p} 
    &\leq \mathds{E}_{ \{x(m)\}\sim p_{\X} }
    \Tr\left[ \rho_{\B}^{x({1})}  \cdot \frac{
    \sum_{\bar{m}\neq 1} \left(\rho_{\B}^{x({\bar{m}})}\right)^{\alpha}
    }{    \left(\rho_{\B}^{x({1})}\right)^{\alpha}
    + \sum_{\bar{m}\neq 1} \left(\rho_{\B}^{x({\bar{m}})}\right)^{\alpha}
    } \right]
    \\ \notag
    &\leq \mathds{E}_{ \{x(m)\}\sim p_{\X} } c_{\alpha} \Tr\left[ \left( \rho_{\B}^{x({1})} \right)^{\alpha} \left( \sum_{\bar{m}\neq 1} \left( \rho_{\B}^{x(\bar{m})} \right)^{\alpha} \right)^{\frac{1-\alpha}{\alpha}} \right],
\end{align}
where the inequality follows from Proposition~\ref{prop:key} with $A \leftarrow \rho_{\B}^{x({1})} $ and $B \leftarrow  \left[ \sum_{\bar{m}\neq 1} \left( \rho_{\B}^{x(\bar{m})} \right)^{\alpha} \right]^{\nicefrac{1}{\alpha}} $.

Now we evaluate the expectation over the random codebook $\{x(m)\} \sim p_{\X}$ to obtain
\begin{align*}
    &\mathds{E}_{x(1)} \mathds{E}_{x(\bar{m})\mid x({1})}
    \Tr\left[ \left( \rho_{\B}^{x({1})} \right)^{\alpha} \left( \sum_{\bar{m}\neq 1} \left( \rho_{\B}^{x(\bar{m})} \right)^{\alpha} \right)^{\frac{1-\alpha}{\alpha}} \right]
    \\
    &\overset{\text{(a)}}{\leq} 
    \mathds{E}_{x(1)}
    \Tr\left[ \left( \rho_{\B}^{x({1})} \right)^{\alpha} \left( \mathds{E}_{x(\bar{m})\mid x({1})}\sum_{\bar{m}\neq 1} \left( \rho_{\B}^{x(\bar{m})} \right)^{\alpha} \right)^{\frac{1-\alpha}{\alpha}} \right]
    \\
    &\overset{\text{(b)}}{=} (|\mathsf{M}|-1)^{\frac{1-\alpha}{\alpha}} \cdot
    \mathds{E}_{x\sim p_{\X}} \Tr\left[ \left( \rho_{\B}^x \right)^{\alpha } \left( \mathds{E}_{\bar{x}} \left( \rho_{\B}^{\bar{x}} \right)^{\alpha} \right)^{\frac{1-\alpha}{\alpha}}
    \right]
    \\
    &= (|\mathsf{M}|-1)^{\frac{1-\alpha}{\alpha}} \cdot \Tr\left[ \left( \sum_{x\in \X} p_{\X}(x) \left( \rho_{\B}^x \right)^{\alpha}
    \right)^{\nicefrac{1}{\alpha}}\right],
\end{align*}
where inequality (a) is because the power function $0\leq x\mapsto x^{\frac{1-\alpha}{\alpha}}$ is operator concave for $\frac{1-\alpha}{\alpha} \in [0,1]$;
equality (b) follows from the pairwise independence of the random codebook.
\end{proof}

\subsection{Extension to Infinite Dimensions} \label{sec:inf-dim}

Theorem~\ref{theo:CQ} is established for any classical-quantum channel with finite-dimensional output Hilbert space.
In the following, we show that, by employing the technique of \emph{finite-rank approximations} developed by Hiai \cite{Hia21} and Mosonyi \cite[\S III.C]{Mos23}, the result holds in infinite dimensions as well.
This proves, for the first time, the optimal error exponent (for rates higher than the critical rate) for infinite-dimensional (separable) Hilbert spaces, as will be shown shortly in Section~\ref{sec:large_deviation}.

Consider a classical-quantum channel $\mathscr{N}_{\X\to\B}: \X \to \mathcal{S}(\B)$. Here, the input alphabet $\X$ is finite, while the output Hilbert space $\mathcal{H}_{\B}$ may be infinite dimensional.
Let $(\I_{\B_k})_{k\in\mathds{N}}$ be an increasing sequence of finite-rank projections onto subspaces $\mathcal{H}_{\B_k}\subset\mathcal{H}_{\B}$ such that $\I_{\B_k} \nearrow \I_{\B}$ strongly and $\Tr[\I_{\B_k}] = k$.
We choose an $\alpha$-PGM according to the ``truncated'' channel outputs as follows:
\begin{align}
	\begin{split}
		\mathring{\Pi}_{\B_k }^{x(m)} 
		\coloneqq\frac{\left( \I_{\B_k} \rho_{\B}^{x({m})} \I_{\B_k} \right)^{\alpha}}{\sum_{\bar{m}\in \mathsf{M} } \left( \I_{\B_k} \rho_{\B}^{x(\bar{m})} \I_{\B_k} \right)^{\alpha} } + \frac{1}{|\mathsf{M}|} \I_{\B_k^{\perp}}, \quad \forall \, m \in \mathsf{M}, \; \alpha \in [\sfrac{1}{2},1],
	\end{split}
\end{align}
and $\I_{\B_k^{\perp}} \coloneqq \I_{\B}-\I_{\B_k}$ is the complement of the truncation subspace. 

Then, the random coding error is upper bounded by 
\begin{align}
	\varepsilon(\X:\B)_{p} 
	&\leq \mathds{E}_{ x(m)\sim p_{\X} }
	\Tr\left[ \rho_{\B}^{x({1})}  \cdot \frac{
		\sum_{\bar{m}\neq 1} \left(\I_{\B_k}\rho_{\B}^{x({\bar{m}})}\I_{\B_k}\right)^{\alpha}
	}{    \left(\I_{\B_k}\rho_{\B}^{x({1})}\I_{\B_k}\right)^{\alpha}
		+ \sum_{\bar{m}\neq 1} \left(\I_{\B_k}\rho_{\B}^{x({\bar{m}})}\I_{\B_k}\right)^{\alpha}
	} \right]
	+ \frac{|\mathsf{M}|-1}{|\mathsf{M}|}\Tr\left[ \rho_{\B} \I_{\B_k^{\perp}} \right]
	\notag
	\\
	&= \mathds{E}_{ x(m)\sim p_{\X} }
	\Tr\left[ \I_{\B_k}\rho_{\B}^{x({1})}\I_{\B_k}  \frac{
		\sum_{\bar{m}\neq 1} \left(\I_{\B_k}\rho_{\B}^{x({\bar{m}})}\I_{\B_k}\right)^{\alpha}
	}{    \left(\I_{\B_k}\rho_{\B}^{x({1})}\I_{\B_k}\right)^{\alpha}
		+ \sum_{\bar{m}\neq 1} \left(\I_{\B_k}\rho_{\B}^{x({\bar{m}})}\I_{\B_k}\right)^{\alpha}
	} \right]
	+ \frac{|\mathsf{M}|-1}{|\mathsf{M}|}\Tr\left[ \rho_{\B}\I_{\B_k^{\perp}} \right]
	\notag
	\\
	&\leq c_{\alpha} \cdot (|\mathsf{M}|-1)^{\frac{1-\alpha}{\alpha}} \cdot \Tr\left[ \left( \sum_{x\in \X} p_{\X}(x) \left( \I_{\B_k} \rho_{\B}^x \I_{\B_k} \right)^{\alpha}
	\right)^{\nicefrac{1}{\alpha}}\right]
	+ \frac{|\mathsf{M}|-1}{|\mathsf{M}|}\Tr\left[ \rho_{\B} \I_{\B_k^{\perp}} \right].
\end{align}
The last inequality follows from the proof of Theorem~\ref{theo:CQ} for finite dimensions by the substitution $\rho_{\B}^x \leftarrow \I_{\B_k} \rho_{\B}^x \I_{\B_k}$.

Since the above inequality holds for all $k\in\mathds{N}$, we let $k\to \infty$.
For the second term, $\Tr\left[ \rho_{\B} \I_{\B_k^{\perp}}\right] \to 0$, and for the first term, we apply Proposition~\ref{prop:continuity} below with $\A\leftarrow \X$ and $\tau_{\A} \leftarrow \rho_{\X}$
to obtain the following result for infinite dimensions.

\begin{theo}[Infinite dimensions] \label{theo:CQ-inf} 
	For any (possibly infinite-dimensional) classical-quantum channel $\mathscr{N}_{\X \to \B}: x \mapsto \rho_{\B}^x$ and a given input distribution $p_{\X}$ on a finite alphabet $\X$, the random coding error probability \eqref{eq:error_CQ} for sending $|\mathsf{M}|$ messages is upper bounded by
	\begin{align}
		\varepsilon(\X:\B)_{p}
		&\leq c_{\alpha}
		(|\mathsf{M}|-1)^{\frac{1-\alpha}{\alpha}}
		\Tr \left[
		\left( \sum_{x\in\X} p_{\X}(x) \left(\rho_{\B}^x\right)^{\alpha}  \right)^{\nicefrac{1}{\alpha}}
		\right]
		\notag
		\\
		&= c_{\alpha} \cdot
		2^{
			-\frac{1-\alpha}{\alpha} \left[ I_{\alpha} (\X : \B)_{\rho} - \log_2 (|\mathsf{M}|-1) \right]
		}, \quad \forall\, \alpha \in [\nicefrac{1}{2},1]. 
		\label{eq:main-inf}
	\end{align}
	Here, the joint state $\rho_{\X\B}$ is defined in \eqref{eq:CQ-state},
	$I_{\alpha} (\X:\B)_{\rho} \coloneqq\inf_{\sigma_{\B} \in \mathcal{S}(\B) } D_{\alpha}(\rho_{\X\B}\Vert \rho_{\X} \otimes \sigma_{\B}) $ is the order-$\alpha$ Petz--\Renyi information, and the factor $c_{\alpha}$ defined in \eqref{eq:c_alpha} admits a universal upper bound $1.102$.
\end{theo}

\begin{prop} \label{prop:continuity}
	Let $\mathcal{H}_{\A}$ be a finite-dimensional Hilbert space and $\mathcal{H}_{\B}$ be a possibly infinite-dimensional Hilbert space.
	Let $(\I_{\B_k})_{k\in\mathds{N}}$ be an increasing sequence of finite-rank projections onto $\mathcal{H}_{\B_k}$ such that $\I_{\B_k} \nearrow \I_{\B}$ strongly.
	For any trace-class operator $\rho_{\A\B} \geq 0$, $\tau_{\A}\geq 0$, and $\alpha \in (0,1]$, 
	\begin{align}
		\Tr\left[\left(\Tr_{\A}\left[ \tau_{\A}^{1-\alpha} (\I_{\B_k} \rho_{\A\B} \I_{\B_k})^{\alpha} \right] \right)^{\nicefrac{1}{\alpha}} \right]
		\to
		\Tr\left[\left(\Tr_{\A}\left[ \tau_{\A}^{1-\alpha} \rho_{\A\B}^{\alpha} \right] \right)^{\nicefrac{1}{\alpha}} \right]
	\end{align}
	as $k\to \infty$.
\end{prop}
\begin{proof}
	First note that
	\begin{align}
		\Tr\left[\left(\Tr_{\A}\left[ \tau_{\A}^{1-\alpha} \rho_{\A\B}^{\alpha} \right] \right)^{\nicefrac{1}{\alpha}} \right]
		= \left\|\Tr_{\A}\left[ \tau_{\A}^{1-\alpha} \rho_{\A\B}^{\alpha} \right]  \right\|_{\nicefrac{1}{\alpha}}^{\nicefrac{1}{\alpha}}.
	\end{align}
	To show its continuity on system $\B$, it is sufficient to show the continuity of $\left\|\Tr_{\A}\left[ \tau_{\A}^{1-\alpha} \rho_{\A\B}^{\alpha} \right]  \right\|_{\nicefrac{1}{\alpha}}$.
	
	Let $\tau_{\A} = \sum_i \lambda_i |i\rangle_{\A} \langle i |_{\A}$ be the  spectral decomposition of $\tau_{\A}$.
	Then, for all $\alpha \in (0,1]$,
	\begin{align}
		&\left|
		\left\|\Tr_{\A}\left[ \tau_{\A}^{1-\alpha} (\I_{\B_k} \rho_{\A\B} \I_{\B_k})^{\alpha} \right]  \right\|_{\nicefrac{1}{\alpha}}
		-
		\left\|\Tr_{\A}\left[ \tau_{\A}^{1-\alpha} \rho_{\A\B}^{\alpha} \right]  \right\|_{\nicefrac{1}{\alpha}}
		\right|
		\notag
		\\
		&\leq
		\left\| \Tr_{\A}\left[ \tau_{\A}^{1-\alpha} \left( (\I_{\B_k} \rho_{\A\B} \I_{\B_k})^{\alpha} - \rho_{\A\B}^\alpha \right) \right] \right\|_{1/\alpha}
		\\
		&=
		\left\| \sum_{i} \lambda_i^{1-\alpha} \langle i |_{\A}  \left( (\I_{\B_k} \rho_{\A\B} \I_{\B_k})^{\alpha} - \rho_{\A\B}^\alpha \right) |i \rangle_{\A} \right\|_{1/\alpha}
		\\
		&\overset{\textnormal{(a)}}{\leq}
		\sum_{i} \lambda_i^{1-\alpha}
		\left\| \langle i |_{\A} \left(  (  \I_{\B_k} \rho_{\A\B} \I_{\B_k})^{\alpha} - \rho_{\A\B}^\alpha \right) |i\rangle_{\A} \right\|_{1/\alpha}
		\\
		&\overset{\text{(b)}}{\leq}
		\sum_{i} \lambda_i^{1-\alpha}
		\left\| (  \I_{\B_k} \rho_{\A\B} \I_{\B_k})^{\alpha} - \rho_{\A\B}^\alpha \right\|_{1/\alpha}
		\\
		&\overset{\text{(c)}}{\leq}
		\sum_{i} \lambda_i^{1-\alpha}
		\left\| \I_{\B_k} \rho_{\A\B} \I_{\B_k} - \rho_{\A\B} \right\|_{1}^\alpha
		\\
		&\to 0,
	\end{align}
	where (a) follows from the triangle inequality of the Schatten $p$-norm with $p = {1}/{\alpha} \geq 1$;
	(b) is because the Schatten $1/\alpha$-norm is contractive under the projection $\langle i |_{\A} \otimes \I_{\B} (\cdot) \I_{\B} \otimes |i\rangle_{\A}$ (which follows from H\"older's inequality $\| P H P^\dagger \|_{p} \leq \|P\|_{\infty} \|H\|_p \|P^\dagger\|_{\infty} = \|H\|_p$ for any projection $P$);
	and (c) follows from the Powers--Størmer inequality \cite[Corollary 4]{And88}, \cite{HN89}, \cite[(2.8)]{Hia97}.
\end{proof}
\subsection{Large Deviation Regime and Comparisons} \label{sec:large_deviation}

When the underlying channel is product, say $\mathscr{N}_{\X_1 \to \B_1} \otimes \mathscr{M}_{\X_2 \to \B_2}$, the right-hand side of \eqref{eq:main} is multiplicative up to a universal constant factor that is independent of the blocklength and the system dimensions (see e.g.~\cite[Lemma 7]{HT14}) for \emph{any} $\alpha \in [\nicefrac{1}{2},1]$ and \emph{any} product input distribution, say $p_{\X_1} \!\otimes q_{\X_2}$.

If the channel is i.i.d., i.e., $\mathscr{N}_{\X\to \B}^{\otimes n}$, 
one may choose an arbitrary i.i.d.~input distribution $p_{\X}^{\otimes n}$; the resulting random codebook is called the \emph{i.i.d.~random codebook}.
By optimizing the order $\alpha \in [\nicefrac{1}{2},1]$, Theorem~\ref{theo:CQ} implies
\begin{align} \label{eq:CQ-exponent-iid}
    -\log_2  \varepsilon\left(\X^n : \B^n\right)_{ p^{\otimes n} }
    \geq
    n \cdot \sup_{\alpha \in [\nicefrac{1}{2},1]} \frac{1-\alpha}{\alpha} \left[ I_{{\alpha}}\left( \X:\B \right)_{\rho} - R \right]
    - \log_2 (1.102), \quad \forall\, n \in \mathds{N}.
\end{align}
If one further chooses an optimal input distribution, the above bound leads to an achievable error exponent:\footnote{The maximal order-$\alpha$ Petz--\Renyi information is called the order-$\alpha$ Petz--\Renyi radius \cite[\S 4]{MO14} or Petz--\Renyi capacity.}
\begin{align} \label{eq:CQ-exponent}
    \liminf_{n\to \infty} - \frac{1}{n} \log_2 \inf_{p_{\X}} \varepsilon\left(\X^n : \B^n\right)_{ p^{\otimes n} }
    \geq
    \sup_{p_{\X}} \sup_{\alpha \in [\nicefrac{1}{2},1]} \frac{1-\alpha}{\alpha} \left[ I_{{\alpha}}\left( \X:\B \right)_{\rho} - R \right].
\end{align}
Here, the asymptotic \emph{transmission rate} is
\begin{align} \label{eq:rate}
R \coloneqq\lim_{n\to \infty} \frac{1}{n} \log_2 |\mathsf{M}|,
\end{align}
which is a \emph{fixed} constant below the channel capacity $\sup_{p_{\X}} I_{1}\left( \X:\B \right)_{p}$.
Such a setting is conventionally called the \emph{large deviation regime} because the rate deviates from the fundamental limit $\sup_{p_{\X}} I_{1}\left( \X:\B \right)_{p}$ by a constant amount (independent of the blocklength $n$).
(In some contexts, it is also called the small-error regime as the error vanishes exponentially fast.)

The above error exponent is positive if and only if $R < \sup_{p_{\X}} I_{1}\left( \X:\B \right)_{\rho} $.
It coincides with the results independently obtained by Renes  \cite{Ren25} and Li--Yang \cite{LY25}, and also matches Dalai's sphere-packing exponent for classical-quantum channels \cite[Theorem 5]{Dal13} (see also \cite[Theorem 8]{CHT19} for the non-asymptotic refined sphere-packing bound) for $R$ higher than the so-called \emph{critical rate} $\sup_{p_{\X}} \left. \frac{\d}{\d s} s I_{\frac{1}{1+s}}\left( \X:\B \right)_{\rho} \right|_{s=1}$.
The advantages of our result in Theorem~\ref{theo:CQ} are that 
(i) we do not have the polynomial prefactor $(n+1)^{|\B|}$, which could be quite large in short blocklengths $n$,
and 
(ii) our error bound accommodates any input distribution $p_{\X}$ available at the encoder, which is of practical importance for implementations in communication protocols assisted by shared randomness or shared entanglement.

\medskip
It is worth mentioning that the first exponential-decay error bound for mixed classical-quantum channels was proved by Hayashi \cite{Hay07}:
\begin{align} \label{eq:Hay07}
\varepsilon(\X:\B)_{p}
		&\leq 4
        (|\mathsf{M}|-1)^{\frac{1-\alpha}{\alpha}}
        \sum_{x\in\X} p_{\X}(x)  \Tr \left[
        \left( \rho_{\B}^x\right)^{2-\frac{1}{\alpha}} \rho_{\B}^{\frac{1-\alpha}{\alpha}}
        \right], \quad \forall\, \alpha \in [\nicefrac{1}{2},1].
\end{align}
Later, the constant $4$ was improved to $1$ by one of the present authors \cite[Proposition 1]{Cheng_simple}.
For $n$-fold product channels, Hayashi's bound leads to the following achievable error exponent (using any i.i.d.~codebook $p_{\X}^{\otimes n}$):
\begin{align} \label{eq:Hayashi-exponent}
    \sup_{\alpha \in [\nicefrac{1}{2},1]} \frac{1-\alpha}{\alpha} \left[ D_{2-\frac{1}{\alpha}}\left( \rho_{\X\B} \Vert \rho_{\X} \otimes \rho_{\B} \right) - R \right].
\end{align}

In Proposition~\ref{prop:comparison_Hayashi} below, we directly show that the established error exponent in \eqref{eq:CQ-exponent} for any fixed $p_{\X}$ is at least as large as \eqref{eq:Hayashi-exponent}.
The proof immediately follows from the duality relations for the conditional \Renyi entropies \cite[Corollary 4]{TBH14}, \cite[Corollary 5.3]{Tom16}, and \cite[Lemma 4]{HT14}, established by Tomamichel, Berta, and Hayashi \cite{TBH14}.

\begin{prop} \label{prop:comparison_Hayashi}
For any quantum state $\rho_{\A\B} \in \mathcal{S}(\A\otimes \B)$ and $\tau_{\A} \geq 0$ such that $\supp(\rho_{\A}) \subseteq \supp(\tau_{\A})$, we have
\begin{align}
    D_{2-\frac{1}{\alpha} } \left( \rho_{\A\B} \Vert \tau_{\A} \otimes \rho_{\B} \right)
    \leq \inf_{\sigma_{\B} \in \mathcal{S}(\B)} D_{\alpha } \left( \rho_{\A\B} \Vert \tau_{\A} \otimes \sigma_{\B} \right),
    \quad \forall\,\alpha \geq \nicefrac{1}{2}.
\end{align}

In particular, for any quantum ensemble, $\{ p_{\X}(x), \rho_{\B}^x \}_{x\in\X}$,
\begin{align}
    \Tr \left[
        \left( \sum_{x\in\X} p_{\X}(x) \left(\rho_{\B}^x\right)^{\alpha}  \right)^{\nicefrac{1}{\alpha}}
        \right]
    \leq 
    \sum_{x\in\X} p_{\X}(x)  \Tr \left[
        \left( \rho_{\B}^x\right)^{2-\frac{1}{\alpha}} \rho_{\B}^{\frac{1-\alpha}{\alpha}}
        \right],
        \quad \forall\,\alpha \geq \nicefrac{1}{2}.
\end{align}
\end{prop}
\begin{proof}
    For any purification $\rho_{\A\B\C}$ of $\rho_{\A\B}$, we invoke \cite[Lemma 4]{HT14} as follows: for any $\alpha \geq \nicefrac{1}{2}$,
    \begin{align*}
        D_{2-\frac{1}{\alpha} } \left( \rho_{\A\B} \Vert \tau_{\A} \otimes \rho_{\B} \right)
        &=
        - D_{\frac{1}{\alpha} } \left( \rho_{\A\C} \Vert \tau_{\A}^{-1} \otimes \rho_{\C} \right) && \text{(by \cite[(3.14)]{HT14})}
        \\
        &\leq - \widetilde{D}_{\frac{1}{\alpha} } \left( \rho_{\A\C} \Vert \tau_{\A}^{-1} \otimes \rho_{\C} \right)
        \\
        &=
        \inf_{\sigma_{\B} \in \mathcal{S}(\B)} D_{\alpha } \left( \rho_{\A\B} \Vert \tau_{\A} \otimes \sigma_{\B} \right),
        && \text{(by \cite[(3.13)]{HT14})}
    \end{align*}
    where $\widetilde{D}_{{1}/{\alpha} }$ is the \emph{sandwiched \Renyi divergence}
    \cite{MDS+13, WWY14} (whose exact definition is not important here), which is always less than the Petz--\Renyi divergence  ${D}_{2-\nicefrac{1}{\alpha} }$ by the Araki--Lieb--Thirring inequality {\cite{Ara90, LT76}}.
\end{proof}

\subsection{Moderate Deviation Regime} \label{sec:moderate_deviation}

In communication, one endeavors to operate at the highest possible rate $R$ (cf.~\eqref{eq:rate}).
Hence, we can let the transmission rate approach channel capacity as the blocklength $n$ increases and study the asymptotic minimum error.
As $R_n$ approaches channel capacity at a rate slower than $\nicefrac{1}{\sqrt{n}}$, the error probability does not decay exponentially but still vanishes \cite{PV10, AW14b, CH17, CTT2017}.
Such a situation is called the \emph{moderate deviation regime} because the rate deviates from the fundamental limit by a moderate amount, and it is also called the medium-error regime.

In \cite[Theorem 4]{CH17}, the moderate deviation analysis was done by using Hayashi's bound \eqref{eq:Hay07}.
In the following, we show that Theorem~\ref{theo:CQ} also directly leads to vanishing errors in the moderate deviation regime.
This can be considered the quantum analogue of \cite[Theorem 2.1]{AW14b}.

Let $(a_n)_{n\in\mathds{N}}$ be any real sequence such that 
\begin{align}
    \lim_{n\to \infty} a_n = 0, \quad \text{and} \quad
    \lim_{n\to \infty} \sqrt{n} \cdot a_n = \infty.
\end{align}
Suppose that the channel $\mathscr{N}_{\X\to \B} $ satisfies\footnote{Note that the global constraint (over $s\in[0,1]$) in \eqref{eq:Upsilon} can be slightly weakened to being only around $s=0$ for a capacity-achieving input distribution.}
\begin{align} \label{eq:Upsilon}
    \Upsilon \coloneqq\sup_{s\in[0,1]} \sup_{p_{\X}} \left| \frac{\d^3}{\d s^3} s I_{\frac{1}{1+s}} (\X:\B)_{\rho} \right| < \infty.
\end{align}
Denote the channel capacity by $C_{\mathscr{N}} \coloneqq\sup_{p_{\X}} I_1(\X:\B)_{\rho}$, and define the \emph{channel dispersion} as:
\begin{align}
    V_{\mathscr{N}}
    &\coloneqq \inf_{p_{\X} : I_1(\X:\B)_{\rho} = C_{\mathscr{N}} } \mathds{E}_{x \sim p_{\X}} V\left( \rho_{\B}^x \Vert \rho_{\B} \right), \label{eq:V}
    \\
    V(\rho\Vert\sigma)
    &\coloneqq \Tr\left[ \rho \left( \log_2 \rho - \log_2 \sigma - D_1(\rho\Vert\sigma)\I \right)^2 \right].
\end{align}
For channels $\mathscr{N}$ with positive $V_{\mathscr{N}}>0$,
we choose $\alpha_n = (1 + \frac{a_n}{V_{\mathscr{N}}})^{-1}$ in \eqref{eq:CQ-exponent-iid} along with the Taylor series expansion of $D_{\alpha}$ around $1$ (see e.g.~\cite[Proposition 2]{CH17}) to obtain
\begin{align}
\sup_{\alpha \in [\nicefrac{1}{2},1]} \frac{1-\alpha}{\alpha} \left[ I_{{\alpha}}\left( \X:\B \right)_{\tilde{\rho}} - (C_{\mathscr{N}} - a_n) \right]
&\geq \frac{a_n^2}{2 V_{\mathscr{N}} \log(2)} -  \frac{a_n^3}{ 6 V_{\mathscr{N}}^3 \log(2)} \Upsilon,
\end{align}
where $\tilde{\rho}_{\X\B} = \sum_{x\in\X} \tilde{p}_{\X}(x) |x\rangle \langle x|\otimes \rho_{\B}^x$ and $\tilde{p}_{\X}$ is any dispersion-achieving input distribution (i.e., $\tilde{p}_{\X}$ achieves \eqref{eq:V}).
Then, Theorem~\ref{theo:CQ} implies
\begin{align}
    - \frac{1}{n a_n^2} \log_2 \varepsilon (\X^n : \B^n)_{\tilde{p}^{\otimes n}} 
    \geq  \frac{1}{2 V_{\mathscr{N}}\log(2)} \left( 1 - \frac{a_n}{3V_{\mathscr{N}}^2}\Upsilon \right)
    - \frac{\log_2 c_{\alpha_n}}{ n a_n^2 }
    \to \frac{1}{2 V_{\mathscr{N}}\log(2)}
\end{align}
as $n\to \infty$ (by noting that $c_{\alpha_n} = \mathcal{O}(1)$).
Moreover, the sub-exponential vanishing rate $\frac{1}{2 V_{\mathscr{N}}\log(2)}$ is tight \cite[Theorem 5]{CH17}.

\subsection{Small Deviation Regime} \label{sec:small_deviation}
For the large-error regime, also known as the non-vanishing-error regime, one asks for the largest achievable transmission rate $R$ with a prescribed error tolerance, that is, subject to the minimum error probability being no larger than a fixed constant, say $\epsilon \in (0,1)$ \cite[Proposition 2]{Cheng_simple}.
In the i.i.d.~asymptotics, the largest achievable rate converges to the channel capacity at a speed of $\mathcal{O}(\nicefrac{1}{\sqrt{n}})$
\cite{TH13, Li14, TT15, DPR16, OMW19, PW20, Cheng_simple}.
Hence, it is also called the \emph{small deviation regime}, as the rate deviates from the fundamental limit by a small amount $\nicefrac{1}{\sqrt{n}}$.

Given that the optimizing order $\alpha$ approaches $1$, it is natural to choose the $1$-PGM in decoding.
As shown in \cite{Cheng_simple}, the key ingredient giving the tightest one-shot characterization for $R$ is the following trace inequality when using the conventional $1$-PGM:
\begin{align}
\Tr\left[A (A+B)^{-\nicefrac{1}{2}} B (A+B)^{-\nicefrac{1}{2}} \right] \leq \Tr\left[ A \wedge B \right], \quad \forall\, A,B\geq 0,
\end{align}
where the noncommutative minimum $A\wedge B$ is given in \eqref{eq:noncommutative_minimum}.
Later, Beigi and Tomamichel proved that the integral $1$-PGM also satisfies a similar inequality \cite[Lemma 4]{BT24}, i.e.,
\begin{align} \label{eq:BT24_trace}
\Tr\left[A \frac{B}{A+B} \right] \leq \Tr\left[ A \wedge B \right],  \quad \forall\, A,B\geq 0
\end{align}
by adapting the analysis in \cite[Lemma 1]{Cheng_simple}.
Hence, the decoder in \eqref{eq:alpha-PGM} for $\alpha = 1$ can also achieve the state-of-the-art performance in the small deviation regime.

In the following, we show that \eqref{eq:BT24_trace} is a natural consequence of the operator layer cake theorem.
Using Theorem~\ref{theo:Dlog_formula} for $\frac{B}{A+B}$, we have
\begin{align}
\Tr\left[A \frac{B}{A+B} \right]
&= \int_{0}^1 \Tr\left[ A \proj{ A < \frac{1-u}{u} B } \right] \d u
\leq \int_{0}^1 \Tr\left[ A \proj{ A < \frac{1}{u} B } \right] \d u
= \Tr\left[ A \wedge B \right].
\end{align}
Here, the inequality is because the map $\gamma \mapsto \Tr\left[ A \proj{A < \gamma B } \right]$ is monotone increasing for $A,B\geq 0$ \cite{LHC25_layer_cake}\footnote{Monotonicity of the map $\gamma \mapsto \Tr\left[ A \proj{A < \gamma B } \right]$ is equivalent to convexity of the hockey-stick divergence in $\gamma$, i.e., 
$\gamma \mapsto E_{\gamma}(A\Vert B) \coloneqq\Tr\left[(A-\gamma B)_+\right]$, which is also equivalent to the monotonicity of the
\emph{information-spectrum divergence} $D_\textnormal{S}^{\epsilon}(\rho \,\|\, \sigma )\coloneqq\sup\left\{ \gamma \in \mathds{R} : \Tr\left[ \rho \left\{ \rho\leq \mathrm{e}^\gamma \sigma \right\} \leq \eps \right]  \right\} $ \cite{VH94, Han03, HN03, NH07, TH13} in the error parameter $\epsilon$.
Note that the known smooth entropic quantities such as the \emph{$\epsilon$-hypothesis-testing divergence}
$D_\textnormal{H}^\epsilon(\rho \,\|\, \sigma) \coloneqq\sup_{0\leq T\leq \I } \left\{ -\log \Tr[\sigma T] : \Tr[\rho T] \geq 1 - \varepsilon \right\}$
\cite{TH13, WR13, Li14} and related quantities; see e.g.~\cite{RLD25}, all satisfy this property; namely, the distinguishability increases once the error tolerance is relaxed.}, which may be viewed as a relaxation.
The last equality is from the integral representation in Proposition~\ref{prop:min} below.

\begin{prop}[Integral representation for the tracial noncommutative minimum] \label{prop:min}
For positive semidefinite $A, B\geq 0$, we have
\begin{align}
\Tr\left[ A \wedge B \right]
&= \int_0^1 \Tr\left[ A \proj{ u A < B } \right] \d u,
\end{align}
where $A\wedge B \coloneqq\frac12\left[ A+B- |A-B| \right]$.
\end{prop}
\begin{proof}
Note that $A\wedge B = B - (B-A)_+$.
We have
\begin{align*}
	\Tr\left[ A \wedge B \right]
	&= \Tr[B-(B-A)_+]
    =-\Tr[(B-uA)_+]\big\vert_{u=0}^{u=1}
    \\
    &=\int_{0}^1 \left(-\frac{\d}{\d u}\Tr\left[ (B-u A)_+ \right]\right) \d u
    \\
    &=\int_{0}^1 \Tr\left[ A \left\{ B-u A > 0 \right\} \right] \d u,
\end{align*}
where we use the Lebesgue integral and note that $u\mapsto \Tr[(B-u A)_+]$ is continuous and piecewise differentiable almost everywhere (see e.g.~\cite[Lemma 2.3]{LHC25_layer_cake}).
\end{proof}

\section{Constrained Classical-Quantum Channel Coding} \label{sec:CC}

The goal of this section is still classical-quantum channel coding, but now with a constrained input set $\Z \subset \X$.
In the $n$-shot setting, we will show later in Section~\ref{sec:cc_codes} that constant-composition codes are optimal in the sense that the established error exponent matches the sphere-packing exponent for \emph{each} composition.\footnote{Since constant-composition codes are finite-blocklength objects, the error exponent here also refers to the finite-blocklength exponent, namely, the leading term of $-\log \varepsilon$ after ignoring other sublinear terms. We will discuss this issue in Remark~\ref{remark:CC} later.}

We follow the notation given in Section~\ref{sec:CQ} and fix a classical-quantum channel $\mathscr{N}_{\X\to\B}: x\mapsto \rho_{\B}^x$ and an input distribution $p_{\X}$.
Since the codewords are now constrained to the set $\Z$ with $p_{\X}(\Z)>0$, we instead employ the following conditional probability distribution for random coding:
\begin{align} \label{eq:p_breve}
    \breve{p}_{\X}(x)
    &\coloneqq\frac{\{x \in \Z \}\cdot  p_{\X}(x) }{ p_{\X}(\Z) }
\end{align}
(here $\{x \in \Z \}$ is understood as the indicator function of the event $x \in \Z$).
We use the notation $\breve{p}_{\X}$ to highlight that it is induced from the original distribution $p_{\X}$.
Additionally, we still write $\X$ in the subscript of $\breve{p}_{\X}$ instead of $\Z$ for notational convenience, since we can always regard $\Z$ as embedded in the larger system $\X$.

Our goal is to bound the random coding error probability $\varepsilon({\X}:\B)_{\breve{p}}$ as in \eqref{eq:error_CQ}, and characterize it in terms of the c-q channel $\mathscr{N}_{\X\to\B}$ and input distribution $p_{\X}$.
    
To characterize $\varepsilon(\X:\B)_{\breve{p}}$, we introduce the \emph{order-$\alpha$ Petz--Augustin information} \cite{Aug69, Aug78, Csi95, nakiboglu_augustin_2018, MO18, CGH18} for the input distribution $p_{\X}$ as
\begin{align} \label{eq:Augustin}
{I}^{\text{Aug}}_{\alpha}(p;\mathscr{N})
\coloneqq\inf_{\sigma_{\B} \in \mathcal{S}(\B)} \sum_{x\in\X} p_{\X}(x) D_{\alpha}\left( \rho_{\B}^x \Vert \sigma_{\B} \right), \quad \alpha \in (0,1].
\end{align}
It is known that when $|\X|<\infty$, the minimizer in ${I}^{\text{Aug}}_{\alpha}(p_{\X};\mathscr{N})$ is always attained by a unique state, called the \emph{Augustin mean}, and we denote it by $\breve{\sigma}_{\B}^{\star}(\alpha,p_{\X}) \in \mathcal{S}(\B)$.
(We will drop the dependence $(\alpha,p_{\X})$ for simplicity when there is no possibility of confusion.) 
Moreover, it satisfies a fixed-point property (see \cite[Proposition 2-(b)]{CHT19}, \cite[Theorem IV.14]{MO18} for the finite-dimensional case and \cite{Cheng2021d}, \cite[(32)]{CB24} for the infinite-dimensional case):
\begin{align} \label{eq:fixed-point}
\breve{\sigma}_{\B}^{\star}(\alpha,p_{\X})
=  \left( \sum_{x\in\X} p_{\X}(x) 2^{ (1-\alpha) D_{\alpha}\left(\rho_{\B}^x \Vert \breve{\sigma}_{\B}^{\star}(\alpha, p_{\X})\right)  } \left(\rho_{\B}^x\right)^{\alpha}  
\right)^{\nicefrac{1}{\alpha}}, \quad \alpha \in (0,1].
\end{align}
This further implies that the Augustin mean has the tensorization property, i.e., 
\begin{align} \label{eq:tensor_Augustin_mean}
\breve{\sigma}_{\B^n}^{\star}\left(\alpha, p_{\X}^{\otimes n}\right) = \left(\breve{\sigma}_{\B}^{\star}(\alpha, p_{\X}) \right)^{\otimes n}    
\end{align} 
for $n$-fold product $p_{\X}^{\otimes n}$ and $\mathscr{N}_{\X\to\B}^{\otimes n}$.
Note that, for $\alpha=1$, the Augustin information coincides with the usual mutual information, i.e.,
${I}^{\text{Aug}}_{1}(p;\mathscr{N}) = {I}_{1}(\X:\B)_{\rho} = D_1(\rho_{\X\B}\Vert\rho_{\X}\otimes \rho_{\B})$.

\begin{theo}[One-shot bound for constrained classical-quantum channel coding] \label{theo:CC-one-shot} 
For any classical-quantum channel $\mathscr{N}_{\X \to \B}: x\mapsto \rho_{\B}^x$ and distribution $p_{\X}$, the random coding error probability  for sending $|\mathsf{M}|$ messages with an induced input distribution $\breve{p}_{\X}$ on the constrained set $\Z \subset \X$ with $p_{\X}(\Z)>0$, i.e., \eqref{eq:p_breve}, is upper bounded by
	\begin{align} \notag
		{\varepsilon}(\X:\B)_{\breve{p}}
		&\leq \frac{c_{\alpha}}{ p_{\X}(\Z)^{\nicefrac{1}{\alpha}} } \cdot
        2^{
        -\frac{1-\alpha}{\alpha} \left[  \inf_{x \in \Z } D_{\alpha}\left(\rho_{\B}^{x}\Vert \breve{\sigma}_{\B}^{\star}(\alpha,p_{\X}) \right)  - \log_2(|\mathsf{M}|-1) \right]
        }, \quad \forall\, \alpha \in [\nicefrac{1}{2},1].
	\end{align}
    Here, $\breve{\sigma}_{\B}^{\star}(\alpha,p_{\X})$ is the minimizer in \eqref{eq:Augustin}.
\end{theo}
\begin{remark}
The analysis of Theorem~\ref{theo:CC-one-shot} is largely inspired by Nakibo\u{g}lu's variant of Gallager's inner bound for classical channels \cite[Lemma 5]{Nak20}.
Our contribution here is to show that, via the tilting inequality (Proposition~\ref{prop:key}) and the analysis established in Section~\ref{sec:CQ}, the idea extends to classical-quantum channels with a constrained set as well.
\end{remark}

\begin{proof}
Fix any $\alpha \in [\nicefrac{1}{2},1]$, any $p_{\X}$, and the corresponding $\breve{p}_{\X}$.
For simplicity, we write $ \breve{\sigma}_{\B}^{\star} = \breve{\sigma}_{\B}^{\star}(\alpha,p_{\X})$ subsequently.
We directly apply Theorem~\ref{theo:CQ} to obtain
\begin{align*}
\varepsilon(\X:\B)_{\breve{p}}
		&\leq c_{\alpha}
        (|\mathsf{M}|-1)^{\frac{1-\alpha}{\alpha}}
        \Tr \left[
        \left( \sum_{x\in \Z} \breve{p}_{\X}(x) \left(\rho_{\B}^x\right)^{\alpha}  \right)^{\nicefrac{1}{\alpha}}
        \right]
        \\
        &\overset{\text{(a)}}{\leq}
        c_{\alpha}
        (|\mathsf{M}|-1)^{\frac{1-\alpha}{\alpha}}
        \Tr \left[
        \left( \sum_{x\in\X} \frac{\{x \in \Z \}\cdot  p_{\X}(x) }{ p_{\X}(\Z) } 2^{ (1-\alpha) D_{\alpha}\left(\rho_{\B}^x \Vert \breve{\sigma}_{\B}^{\star}\right)} \left(\rho_{\B}^x\right)^{\alpha}  \right)^{\nicefrac{1}{\alpha}}
        \right] 
        \cdot \sup_{x\in\Z} 2^{ \frac{\alpha-1}{\alpha} D_{\alpha}\left(\rho_{\B}^x \Vert \breve{\sigma}_{\B}^{\star}\right)}
        \\
        &\leq
        c_{\alpha}
        (|\mathsf{M}|-1)^{\frac{1-\alpha}{\alpha}}
        \Tr \left[
        \left( \sum_{x\in\X} \frac{ p_{\X}(x) }{ p_{\X}(\Z) } 2^{ (1-\alpha) D_{\alpha}\left(\rho_{\B}^x \Vert \breve{\sigma}_{\B}^{\star}\right)} \left(\rho_{\B}^x\right)^{\alpha}  \right)^{\nicefrac{1}{\alpha}}
        \right] 
        \cdot \sup_{x\in\Z} 2^{ \frac{\alpha-1}{\alpha} D_{\alpha}\left(\rho_{\B}^x \Vert \breve{\sigma}_{\B}^{\star}\right)}
        \\
        &\overset{\text{(b)}}{=}
        \frac{c_{\alpha}}{ p_{\X}(\Z)^{\nicefrac{1}{\alpha}}  }
        (|\mathsf{M}|-1)^{\frac{1-\alpha}{\alpha}} \Tr\left[ \breve{\sigma}_{\B}^{\star}\right]\cdot \sup\nolimits_{x\in\Z} 2^{ \frac{\alpha-1}{\alpha} D_{\alpha}\left(\rho_{\B}^x \Vert \breve{\sigma}_{\B}^{\star}\right)}
        \\
        &=
        \frac{c_{\alpha}}{ p_{\X}(\Z)^{\nicefrac{1}{\alpha}}  }
        (|\mathsf{M}|-1)^{\frac{1-\alpha}{\alpha}} \sup\nolimits_{x\in\Z} 2^{ \frac{\alpha-1}{\alpha} D_{\alpha}\left(\rho_{\B}^x \Vert \breve{\sigma}_{\B}^{\star}\right)},
\end{align*}
where in (a) we changed the priors $\breve{p}_{\X}$, and 
we invoked the fixed-point property \eqref{eq:fixed-point} in (b).
\end{proof}

\subsection{Constant-Composition Codes} \label{sec:cc_codes}
Theorem~\ref{theo:CC-one-shot} is useful in bounding the error exponent for \emph{constant-composition codes}.
Fix an integer $n\in\mathds{N}$ such that the input distribution $p_{\X}$ is an \emph{$n$-type} (whose probability values are integer multiples of $\frac{1}{n}$).
Now, the constrained set is the typical set, namely, the set of sequences $x^n \in \X^n$ having the same empirical distribution or composition of $p_{\X}$, i.e.,~
\begin{align*}
\Z^n &\leftarrow  \mathcal{T}_{p_{\X}}^n
\coloneqq\left\{ x^n \in \X^n : p_{\X}(x) = \frac{ \sum_{i=1}^n \{ x_i = x \}}{n}, \quad \forall\, x\in\X \right\},
\end{align*}
and hence
\begin{align} \label{eq:CC}
\breve{p}_{\X^n}(x^n) = \frac{ 1 }{\left| \mathcal{T}^n_{p_{\X}} \right|} \left\{ x^n \in  \mathcal{T}^n_{p_{\X}}  \right\} , \quad \forall\, x^n \in \X^n.
\end{align}

Notice that the length-$n$ joint input-output state 
\begin{align}
\breve{\rho}_{\X^n \B^n} = \sum_{x^n \in \X^n} \breve{p}_{\X^n}(x^n) |x^n\rangle\langle x^n |_{\X^n} \otimes \rho_{\B_1}^{x_1} \otimes \cdots \otimes \rho_{\B_n}^{x_n}
\end{align}
is no longer a product state.

\begin{theo}[Random coding bound for constant-composition codes] \label{theo:CC}
Let $\mathscr{N}_{\X\to\B} \colon x\mapsto \rho_{\B}^x$ be a classical-quantum channel.
For any integer $n\in\mathds{N}$ for which $p_{\X}$ is an $n$-type, and any rate $R\coloneqq\frac{1}{n}\log_2 |\mathsf{M}|$, the following bound holds for the constant-composition code $\breve{p}_{\X^n}$ in \eqref{eq:CC}:
\begin{align} \label{eq:CC_bound1}
\log_2 \varepsilon(\X^n:\B^n)_{\breve{p}} \leq - n \frac{1-\alpha}{\alpha} \left[ {I}^{\textnormal{Aug}}_{\alpha}(p;\mathscr{N}) - R \right] + \frac{|\X|}{\alpha}\log_2 (n+1) + \log_2 c_{\alpha}, \quad \forall\, \alpha \in [\sfrac12, 1].
\end{align}
Here, the order-$\alpha$ Petz--Augustin information ${I}^{\textnormal{Aug}}_{\alpha}(p;\mathscr{N})$ is defined in \eqref{eq:Augustin} with respect to the state $\rho_{\X\B} = \sum_{x\in\X} p_{\X}(x)|x\rangle \langle x|_{\X} \otimes \rho_{\B}^x$.

Moreover, for $R \geq \frac{\d}{\d s} s {I}^{\textnormal{Aug}}_{\frac{1}{1+s}}(p;\mathscr{N})\big|_{s=1}$,
\begin{align} \label{eq:CC_bound2}
 -\log_2 \varepsilon(\X^n:\B^n)_{\breve{p}}
= n \cdot \sup_{\alpha \in [\nicefrac{1}{2}, 1] }\frac{1-\alpha}{\alpha} \left[ {I}^{\textnormal{Aug}}_{\alpha}(p;\mathscr{N}) - R \right]
+ \mathcal{O}(\log_2 n)
\end{align}
(where the exact factors in $\mathcal{O}(\log_2 n)$ can be found in \eqref{eq:CC_bound1} and in the converse bound \cite{CHT19}).
\end{theo}

\begin{remark} \label{remark:CK}
In the commuting case, Theorem~\ref{theo:CC} recovers (the non-asymptotic version of) Csisz{\'a}r--K{\"o}rner's random coding bound \cite[Theorem 10.2]{CK11}.
Note that our result in terms of the parametric \Renyi divergence, $\sup_{\alpha\in[\nicefrac{1}{2},1]} \frac{1-\alpha}{\alpha} \left[ {I}^{\textnormal{Aug}}_{\alpha}(p;\mathscr{N}) - R \right]$, is commonly called the \emph{dual-domain} expression \cite{Bla74}, while Csisz{\'a}r--K{\"o}rner's result is expressed in the \emph{primal domain} (which matches the Haroutunian form of the sphere-packing bound \cite{Har68}):
\begin{align} \label{eq:CK}
\inf_{ \mathscr{M}_{\X\to\B}\colon x\mapsto \sigma_{\B}^x  } \left\{
\mathds{E}_{x\sim p_{\X}} D_1( \sigma_{\B}^x \Vert \rho_{\B}^x ) + 
\left( {I}^{\text{Aug}}_1(p;\mathscr{M}) - R
\right)_+
\right\}.
\end{align}
However, in the noncommuting case, the dual-domain expression of \eqref{eq:CK} corresponds to the Augustin information defined via the \emph{log-Euclidean} \Renyi divergence \cite{MO14, Dal16, CHT19}, instead of the Petz--\Renyi divergence.
This means that the random coding error exponent of the form \eqref{eq:CK} is generally not achievable because of the sphere-packing bound established in \cite{DW14, CHT19}.
In light of this, our result in Theorem~\ref{theo:CC} is the only such expression currently known to us; we do not know if there is a primal-domain expression as in the classical case.
\end{remark}

\begin{remark} \label{remark:CC}
Theorem~\ref{theo:CC} estimates the optimal performance of a finite-blocklength constant-composition code, whose codewords, by design, have type $p_{\X}$ as in \eqref{eq:CC}.
Since the leading first-order term in \eqref{eq:CC_bound2} is optimal (for higher rates) for \emph{any} $n$-type, both the leading term and the state $\rho_{\X\B}$ still depend on the blocklength $n$.
From the operational point of view, such an analysis is useful for characterizing the performance of a finite-blocklength code.

When comparing the optimal performance of constant-composition codes with other codes, e.g.,~the i.i.d.~codebook in Section~\ref{sec:large_deviation} (see Remark~\ref{remark:Augustin_vs_Renyi} below), one may choose an $n$-type $p_{\X}^{(n)}$ to approximate an arbitrary input distribution $p_{\X}$ to any prescribed precision.\footnote{Recall that there exists an $n$-type $p_{\X}^{(n)}$ such that the total variation distance is $\frac12\big\| p_{\X}^{(n)} - p_{\X} \big\|_1 \leq \nicefrac{|\X|}{2n} =: \delta_n$ \cite{CK11}.
The continuity of the Augustin information further yields the following cost (\cite[Proposition 5-(c)]{CGH18}):
    $H_{\text{b}}(\delta_n) + \delta_n \cdot \sup\nolimits_{p_{\X}} {I}^{\text{Aug}}_{1}(p;\mathscr{N})$,
where $H_{\text{b}}(q) = - q \log_2 q - (1-q) \log_2 (1-q) \leq \sqrt{4q (1-q)}$ is the binary entropy function.
Overall, the additional cost (for allowing $p_{\X}$ to be a general distribution) on the right-hand side of \eqref{eq:CC_bound1} is bounded by
    $\frac{1-\alpha}{\alpha} \big[
    \sqrt{\frac{2|\X|}{n} } + \frac{|\X|}{2n} \sup\nolimits_{p_{\X}} {I}^{\text{Aug}}_{1}(p;\mathscr{N})
    \big]$.}
\end{remark}

\begin{remark}\label{remark:Augustin_vs_Renyi}
For $p_{\X}$ and $\mathscr{N}_{\X\to\B}:x\mapsto \rho_{\B}^x$, let $\rho_{\X\B}=\sum_{x\in\X} p_{\X}(x)|x\rangle\langle x|_{\X}\otimes \rho_{\B}^x$.
By Jensen's inequality, one has ${I}^{\text{Aug}}_{\alpha}(p;\mathscr{N}) \geq I_{\alpha}(\X:\B)_{\rho}$ \cite{nakiboglu_augustin_2018, MO14, MO18, CGH18} (after optimizing over $p_{\X}$, both quantities equal the order-$\alpha$ \Renyi radius \cite[\S IV]{MO14}).
Hence, constant-composition codes yield a faster exponential decay rate for any $p_{\X}$.
This phenomenon for classical-quantum channels again confirms Gallager's observation \cite{Gal94}: 
constant-composition codes act as a better ensemble than the ensemble with independently chosen letters.
\end{remark}

\begin{proof}[Proof of Theorem~\ref{theo:CC}]
First, the probability of the typical sequences under the i.i.d.~$p_{\X}^{\otimes n} $ is lower bounded by (\cite[\S 2]{CK11}):\footnote{The probability $p_{\X}^{\otimes n}(\mathcal{T}_{p_{\X}}^n)$ can be explicitly calculated by Stirling's formula; see e.g.~\cite[p.~26]{CK11}.
Hence, the polynomial prefactor in Theorem~\ref{theo:CC} can be tightened slightly.}
\begin{align}
p_{\X}^{\otimes n}(\mathcal{T}_{p_{\X}}^n) \geq  (n+1)^{- |\X|}.
\end{align}
Then, we apply the one-shot bound in Theorem~\ref{theo:CC-one-shot}, the tensorization property of the Augustin mean given in \eqref{eq:tensor_Augustin_mean}, and the additivity of the Petz--\Renyi divergence to conclude the proof, namely,
\begin{align}
\inf_{x^n \in \mathcal{T}_{p_{\X}}^n } \frac{1}{n} D_{\alpha}\left( \rho_{\B^n}^{x^n} \Vert \sigma_{\B^n}^{\star} \right)
= \inf_{x^n \in \mathcal{T}_{p_{\X}}^n } \frac{1}{n} D_{\alpha}\left( \rho_{\B^n}^{x^n} \Vert (\sigma_{\B}^{\star})^{\otimes n} \right)
= {I}^{\text{Aug}}_{\alpha}(p;\mathscr{N}).
\end{align}

The finite-blocklength converse (i.e., the so-called sphere-packing bound) was derived in \cite[Theorem 8]{CHT19}:
\begin{align}
- \log_2 \varepsilon(\X^n:\B^n)_{\breve{p}} \leq n \cdot  \sup_{\alpha \in (0,1]} \frac{1-\alpha}{\alpha} \left[ {I}^{\text{Aug}}_{\alpha}(p;\mathscr{N}) - R \right] + \mathcal{O}(\log_2 n)
\end{align}
(the asymptotic result was first shown in \cite{DW14}).
\end{proof}


\section{Classical Data Compression with Quantum Side Information} \label{sec:CQSW}

In this section, we first show a generic one-shot achievability bound for classical data compression with quantum side information.
In Section~\ref{sec:iid_fixed-length}, we consider the $n$-shot setting where the sources are generated in an i.i.d.~fashion, and the compressed indices have a fixed length.
In Section~\ref{sec:type_fixed-length}, we consider sources drawn uniformly from a constant-type class.
In Section~\ref{sec:iid_variable-length}, we return to $n$-shot i.i.d.~sources but now with variable-length coding.

\begin{defn}[Classical data compression with quantum side information]
	Let $\rho_{\X\B} = \sum_{x\in\X} p_{\X} (x) |x\rangle \langle x|_{\X} \otimes \rho_{\B}^x \in \mathcal{S}(\X\B)$ be a classical-quantum state.
	\begin{enumerate}[1.]
		\item Alice has classical registers $\X$ and $\mathsf{M}$, and Bob has a quantum register $\B$.
				
		\item An encoder $\mathcal{E}\colon x\mapsto m(x)$
        at Alice compresses the source in $\X$ to an index in $\mathsf{M}$.
		
		\item According to the compressed index $m$, Bob applies a decoding measurement described by a family of POVMs on register $\B$ indexed by $m\in\mathsf{M}$, i.e.,~$\{ {\Lambda}_{\B}^{x,m} \}_{x\in\X}$, to recover the source $x\in\X$.
	\end{enumerate}
	Again, we adopt the conventional random coding (also known as random binning in source coding) to compress each source letter $x\in\X$ to a uniform index $m\in\mathsf{M}$.
    Then, the random coding error probability for $\rho_{\X\B}$ with index size $|\mathsf{M}|$ is 
	\begin{align} \label{eq:error_CQSW}
    \varepsilon(\X\mid\B)_{\rho}
    \coloneq
    \mathds{E}_{\{m(x)\}\sim \frac{1}{|\mathsf{M}|} }\left[
		\inf_{ \{ \Lambda_{\B}^{x,m} \} } \sum_{ x\in\X }	p_{\X}(x) 
        \Tr\left[ \rho_{\B}^{x} \left(\I_{\B} -  {\Lambda}_{\B}^{x, m} \right) \right]
        \right].
	\end{align}
\end{defn}

Given a realization of the random codebook $\{\mathcal{E}(x)\}_{x\in\X}$, we adopt the integral $\alpha$-PGM for each $m\in\mathsf{M}$:
\begin{align} \notag
    \mathring{\Pi}_{\B}^{x,m} 
    \coloneqq\frac{ \left(p_{\X}(x) \rho_{\B}^x\right)^{\alpha} }{ \sum_{\bar{x}: \mathcal{E}(\bar{x}) = m } \left( p_{\X}(\bar{x}) \rho_{\B}^{\bar{x}}\right)^{\alpha} } + \frac{p_{\X}(x)}{ \sum_{\bar{x}: \mathcal{E}(\bar{x}) = m } p_{\X}(\bar{x}) } \I_{\B^{\perp}}, \quad \forall x\in \X: \mathcal{E}(x)=m, \; \alpha \in [\sfrac{1}{2},1].
\end{align}
(Here, $p_{\X}(x) \rho_{\B}^x \I_{\B^{\perp}} = 0$ for all $x\in\X:\mathcal{E}(x)=m$; we also omit $\I_{\B^{\perp}}$'s dependence on $m$ for brevity.)

\begin{theo} \label{theo:CQSW}
    For any classical-quantum state $\rho_{\X\B}$, the random coding error probability \eqref{eq:error_CQSW} for compressing the source into $|\mathsf{M}|$ indices is upper bounded by
    \begin{align*}
    \varepsilon(\X\mid\B)_{\rho}
    &\leq c_{\alpha} |\mathsf{M}|^{\frac{\alpha-1}{\alpha}} \Tr\left[ \left( \sum_{x\in\X} \left[ p_{\X}(x) \rho_{\B}^x \right]^{\alpha} \right)^{\nicefrac{1}{\alpha}} \right]
    \\
    &= c_{\alpha} \cdot 2^{ - \frac{1-\alpha}{\alpha} \left[ \log_2 \left|\mathsf{M}\right| - H_{\alpha}(\X\mid\B)_{\rho} \right]  },
    \quad \forall\, \alpha \in [\nicefrac{1}{2},1].
    \end{align*}
    Here, $H_{\alpha} (\X\mid\B)_{\rho} \coloneqq-\inf_{\sigma_{\B} \in \mathcal{S}(\B) } D_{\alpha}(\rho_{\X\B}\Vert \I_{\X} \otimes \sigma_{\B}) $ is the order-$\alpha$ conditional Petz--\Renyi entropy.
\end{theo}

The quantum Sibson identity 
\cite[Lemma 3 in Supplementary Material]{SW12}, \cite[(3.10)]{HT14}, \cite{CGH18} shows that the minimizer in $H_{\alpha} (\X\mid\B)_{\rho}$ is attained by
\begin{align} \label{eq:Renyi_mean_CQSW}
\sigma_{\B}^{\star} = \frac{
        \left( \sum_{x\in\X} \left[p_{\X}(x) \rho_{\B}^x\right]^{\alpha}  \right)^{\nicefrac{1}{\alpha}}
        }{
        \Tr \left[
        \left( \sum_{x\in\X} \left[p_{\X}(x) \rho_{\B}^x\right]^{\alpha}  \right)^{\nicefrac{1}{\alpha}}
        \right]}, \quad \alpha > 0,
\end{align}
and, hence, the order-$\alpha$ conditional Petz--\Renyi entropy with respect to the state $\rho_{\X\B}$
admits a closed-form expression:
\[
H_{\alpha} (\X \mid \B)_{\rho} 
= -D_{\alpha}\left(\rho_{\X\B}\Vert \I_{\X} \otimes \sigma_{\B}^{\star}\right)
=
\frac{\alpha}{1-\alpha} \log_2 \Tr \left[
        \left( \sum_{x\in\X} \left[p_{\X}(x) 
 \rho_{\B}^x\right]^{\alpha}  \right)^{\nicefrac{1}{\alpha}}
        \right].
\]

\begin{proof}
We apply Proposition~\ref{prop:key} with $A\leftarrow p_{\X}(x) \rho_{\B}^x$ and $B \leftarrow  \left[ \sum_{\bar{x}\neq x, \, \mathcal{E}(\bar{x})= m } \left(p_{\X}(\bar{x}) \rho_{\B}^{\bar{x}}\right)^{\alpha} \right]^{\nicefrac{1}{\alpha}}$ to obtain, for every $x\in\X$,
\begin{align*}
		&\mathds{E}_{m \sim \frac{1}{|\mathsf{M}|} } \Tr\left[ p_{\X}(x) \rho_{\B}^x \frac{ \sum_{\bar{x}\neq x, \, \mathcal{E}(\bar{x})= m } \left(p_{\X}(\bar{x}) \rho_{\B}^{\bar{x}}\right)^{\alpha} }{ 
        \left( p_{\X}(x) \rho_{\B}^x \right)^{\alpha} + \sum_{\bar{x}\neq x, \, \mathcal{E}(\bar{x})= m } \left(p_{\X}(\bar{x}) \rho_{\B}^{\bar{x}}\right)^{\alpha} } \right]
        \\
		&\leq
		c_{\alpha}  \mathds{E}_{m \sim \frac{1}{|\mathsf{M}|} } \Tr\left[ \left( p_{\X}(x) \rho_{\B}^x\right)^{\alpha} \left( \sum_{\bar{x}\neq x, \, \mathcal{E}(\bar{x})= m } \left(p_{\X}(\bar{x}) \rho_{\B}^{\bar{x}}\right)^{\alpha} 
        \right)^{\frac{1-\alpha}{\alpha}} \right]
        \\
        &\overset{\text{(a)}}{\leq}
        c_{\alpha}   \Tr\left[ \left( p_{\X}(x) \rho_{\B}^x\right)^{\alpha} \left( \mathds{E}_{m \sim \frac{1}{|\mathsf{M}|} } \sum_{\bar{x}\neq x, \, \mathcal{E}(\bar{x})= m } \left(p_{\X}(\bar{x}) \rho_{\B}^{\bar{x}}\right)^{\alpha} 
        \right)^{\frac{1-\alpha}{\alpha}} \right]
        \\
        &= c_{\alpha}   \Tr\left[ \left( p_{\X}(x) \rho_{\B}^x\right)^{\alpha} \left( \sum_{\bar{x}\neq x} \frac{1}{|\mathsf{M}|} \left(p_{\X}(\bar{x}) \rho_{\B}^{\bar{x}}\right)^{\alpha} 
        \right)^{\frac{1-\alpha}{\alpha}} \right]
        \\
        &\overset{\text{(b)}}{\leq} c_{\alpha} |\mathsf{M}|^{\frac{\alpha-1}{\alpha}}  \Tr\left[ \left( p_{\X}(x) \rho_{\B}^x\right)^{\alpha} \left( \sum_{\bar{x} \in \X } \left(p_{\X}(\bar{x}) \rho_{\B}^{\bar{x}}\right)^{\alpha} 
        \right)^{\frac{1-\alpha}{\alpha}} \right],
	\end{align*}
    where inequality (a) is because the power function $0\leq x\mapsto x^{\frac{1-\alpha}{\alpha}}$ is operator concave for $\frac{1-\alpha}{\alpha} \in [0,1]$;
    inequality (b) is because the power function $0\leq x\mapsto x^{\frac{1-\alpha}{\alpha}}$ is operator monotone for $\frac{1-\alpha}{\alpha} \in [0,1]$.

    Summing the above inequality over all $x\in\X$ concludes the proof.
\end{proof}

\subsection{I.I.D.~Sources With Fixed-Length Coding} \label{sec:iid_fixed-length}

In the i.i.d.~scenario, each length-$n$ source $x^n = x_1 x_2 \ldots x_n \in \X^n$ is distributed according to a common i.i.d.~distribution $p_{\X}^{\otimes n}$, and each source is associated with quantum side information $\rho_{\B^n}^{x^n} \coloneqq\rho_{\B_1}^{x_1} \otimes \rho_{\B_2}^{x_2} \otimes \cdots \otimes \rho_{\B_n}^{x_n}$.
In fixed-length coding, each source $x^n$ is compressed to a shorter fixed-length sequence, i.e., $|\mathsf{M}| = \ceil{ 2^{nR} }$, where $R$ is called the \emph{compression rate}.

The resulting joint classical-quantum state for the i.i.d.~source is $\rho_{\X\B}^{\otimes n}$.
By applying the one-shot bound in Theorem~\ref{theo:CQSW}, we have the following characterization of the random coding error.

\begin{theo} \label{theo:CQSW_iid}
    The logarithmic random coding error for compressing the i.i.d.~source $\rho_{\X\B}^{\otimes n}$ under fixed-length coding with compression rate $R$ is characterized by
    \begin{align*}
    -\log_2 \varepsilon(\X^n\mid\B^n)_{\rho^{\otimes n}}
    =
    n\cdot \sup_{\alpha \in [\nicefrac{1}{2},1]} \frac{1-\alpha}{\alpha} \left[ R - H_{\alpha}(\X\mid\B)_{\rho} \right] 
    + \mathcal{O}(\log_2 n),
    \end{align*}
    for $H_1(\X\mid\B)_{\rho} < R \leq -\left.\frac{\d}{\d s}s {D}_{\frac{1}{1+s}}(\rho_{\X\B}\Vert\I_{\X}\otimes \rho_{\B})\right|_{s=1}$
    (where the exact factors in $\mathcal{O}(\log_2 n)$ can be found in the converse bound \cite{CHDH-2018}).
\end{theo}
\begin{remark}
Renes also obtained the characterization of the optimal error exponent \cite{Ren23}, while the achievability bound has a dimension-dependent prefactor (which is of order $(n+1)^{|\B|}$ in the i.i.d.~scenario).
\end{remark}

\begin{proof}
The converse bound has been derived in \cite[Theorem 2]{CHDH-2018}:
\begin{align*}
    - \log_2 \varepsilon(\X^n\mid\B^n)_{\rho^{\otimes n}}
    \leq  n \cdot \sup_{\alpha \in (0,1]} \frac{1-\alpha}{\alpha} \left[ R - H_{\alpha}(\X\mid\B)_{\rho} \right] +  \mathcal{O}(\log_2 n).
\end{align*}
The lower bound follows from Theorem~\ref{theo:CQSW} and the additivity:
$H_{\alpha}(\X^n \mid\B^n )_{\rho^{\otimes n}} = n H_{\alpha}(\X\mid\B)_{\rho}$ as the minimizer in \eqref{eq:Renyi_mean_CQSW} is multiplicative.
\end{proof}

\subsection{Constant-Type Sources With Fixed-Length Coding} \label{sec:type_fixed-length}

For any integer $n\in\mathds{N}$, we let $q_{\X}$  be an $n$-type on $\X$.
Now, we consider the scenario in which each source $x^n \in \X^n$ is distributed uniformly over the type class $\mathcal{T}_{q_{\X}}^n$, i.e.,
\begin{align} \label{eq:CC-source}
\breve{q}_{\X^n}(x^n) = \frac{ 1  }{\left| \mathcal{T}_{q_{\X}}^n \right|} \left\{ x^n \in  \mathcal{T}_{q_{\X}}^n  \right\}, \quad \forall\, x^n \in \X^n.
\end{align}
The resulting classical-quantum state
\begin{align} \label{eq:const-type}
    \breve{\rho}_{\X^n \B^n}
    = \sum_{x^n \in \X^n} \breve{q}_{\X^n}(x^n) |x^n\rangle \langle x^n |_{\X^n} \otimes \rho_{\B^n}^{x^n}
\end{align}
is no longer a product state as in Section~\ref{sec:cc_codes}.

\begin{theo} \label{theo:CQSW_cc}
    For compressing the source $\breve{\rho}_{\X^n \B^n}$ in \eqref{eq:const-type} induced by an $n$-type $q_{\X}$ with a compression rate $R$, we have
    \begin{align*}
    -\log_2 \varepsilon(\X^n\mid\B^n)_{\breve{\rho}}
    =
    n\cdot \sup_{\alpha \in [\nicefrac{1}{2},1]} \frac{1-\alpha}{\alpha} \left[ R - H(\X)_q + {I}^{\textnormal{Aug}}_{\alpha}(q;x\mapsto\rho_{\B}^x) \right]
    + \mathcal{O}(\log_2 n)
    \end{align*}
    for $H(\X)_q - {I}^{\textnormal{Aug}}_1(q;x\mapsto\rho_{\B}^x) < R \leq H(\X)_q - \frac{\d}{\d s} s {I}^{\textnormal{Aug}}_{\frac{1}{1+s}}(q;x\mapsto\rho_{\B}^x)\big|_{s=1}$.
    Here, the order-$\alpha$ Petz--Augustin information ${I}^{\textnormal{Aug}}_{\alpha}(q;x\mapsto\rho_{\B}^x)$ is defined in \eqref{eq:Augustin} for the input distribution $q_{\X}$ and the quantum side information $x\mapsto \rho_{\B}^x$,
    and $H(\X)_q \coloneqq- \sum_{x\in\X} q_{\X}(x) \log_2 q_{\X}(x)$ is the Shannon entropy.
\end{theo}

\begin{remark}
The higher-order term $\mathcal{O}(\log_2 n)$ in Theorem~\ref{theo:CQSW_cc} depends only on $|\X|$ and can be explicitly determined from \cite[Theorem 5]{CHDH2-2018} and Theorem~\ref{theo:CC} for the achievability part, and by referring to \cite[Theorem 5]{CHDH2-2018} for the converse part.
Here, we only focus on the fact that the higher-order term is independent of the dimension $|\B|$ and contributes only a polynomial factor (in $n$) to the error probability.
\end{remark}

\begin{remark}
The order-$\alpha$ Petz--Augustin information generally does not have a closed-form expression for $\alpha \neq 1$.
However, \cite{CB24} showed that it is still analytic in the order $\alpha$.
\end{remark}

\begin{proof}
The converse has been derived in 
\cite[Theorem 5]{CHDH2-2018}:
\begin{align}
    - \log_2 \varepsilon(\X^n\mid\B^n)_{\breve{\rho}}
    \leq n \cdot \sup_{\alpha \in [0,1]} \frac{1-\alpha}{\alpha} \left[ R - H(\X)_q + {I}^{\textnormal{Aug}}_{\alpha}(q;x\mapsto\rho_{\B}^x) \right]
    + \mathcal{O}(\log_2 n).
\end{align}
However, the achievability bound therein is not asymptotically tight.
Nonetheless, \cite[Theorem 2]{CHDH2-2018} showed that the minimum error of the constant-type source compression at rate $R$ with quantum side information is at most twice the minimum error of communication over the classical-quantum channel $x\mapsto \rho_{\B}^x$ via codes of the same composition $q_{\X}$ with rate $ H(\X)_q - R + \frac{1}{n} \log_2( 2n \log_2(|\X|+1) )$.
By invoking Theorem~\ref{theo:CC},
we conclude the proof.
\end{proof}

\subsection{I.I.D.~Sources With Variable-Length Coding} \label{sec:iid_variable-length}

We now return to the i.i.d.~source scenario, i.e., $\rho_{\X\B}^{\otimes n}$ for
\begin{align} \label{eq:state_variable}
    \rho_{\X\B} = \sum_{x\in\X} p_{\X}(x) |x\rangle\langle x|_{\X} \otimes \rho_{\B}^x.
\end{align}
In variable-length coding, each source is compressed to a binary string in $\{0,1\}^*$ with a possibly different length via an encoder $\mathscr{E}^n: \X^n \to \{0,1\}^*$.
We then define the \emph{average rate} of the code as the normalized expected length
\begin{align} \label{eq:rate_CQSW_variable}
    \bar{R} \coloneqq\frac{1}{n} \mathds{E}_{x^n \sim p_{\X}^{\otimes n} } \left[ \mathtt{length}\left( \mathscr{E}^n(x^n) \right) \right].
\end{align}

\begin{theo} \label{theo:CQSW_variable}
    The logarithmic random coding error for compressing the i.i.d.~source $\rho_{\X\B}^{\otimes n}$ in \eqref{eq:state_variable} via variable-length coding with average rate $\bar{R}$ is characterized by
    \begin{align*}
    \lim_{n\to \infty} - \frac{1}{n} \log_2 \varepsilon(\X^n\mid\B^n)_{\rho^{\otimes n}}
    =
    \sup_{\alpha \in [\nicefrac{1}{2},1]} \frac{1-\alpha}{\alpha} \left[ \bar{R} - H(\X)_{p} + {I}^{\textnormal{Aug}}_{\alpha}(p;x\mapsto\rho_{\B}^x) \right]
    \end{align*}
    for $H(\X\!\mid\!\B)_{\rho} < \bar{R} \leq H(\X)_p - \frac{\d}{\d s} s {I}^{\textnormal{Aug}}_{\frac{1}{1+s}}(p;x\mapsto\rho_{\B}^x)\big|_{s=1}$.
    Here, the order-$\alpha$ Petz--Augustin information ${I}^{\textnormal{Aug}}_{\alpha}(p;x\mapsto\rho_{\B}^x)$ is defined in \eqref{eq:Augustin} for the input distribution $p_{\X}$ and the quantum side information $x\mapsto \rho_{\B}^x$,
    and $H(\X)_p \coloneqq- \sum_{x\in\X} p_{\X}(x) \log_2 p_{\X}(x)$ is the Shannon entropy.
\end{theo}

\begin{proof}
Ref.~\cite[Theorem 13]{CHDH2-2018} showed that the optimal error exponent of variable-length coding with average rate $\bar{R} > H(\X\mid\B)_{\rho}$ is equal to that of classical-quantum channel coding with constant-composition $p_{\X}$ and rate $H(\X)_p - \bar{R} \geq 0$.
Our contribution here is to fill the achievability gap in \cite[Proposition 5.1]{CHDH2-2018} by invoking Theorem~\ref{theo:CC}.
\end{proof}


\section{Unassisted Classical Communication over Quantum Channels} \label{sec:unassisted}

\begin{defn}[Unassisted classical communication over quantum channels]
	Let $\mathscr{N}_{\A\to \B}\colon \mathcal{S}(\A) \to \mathcal{S}(\B)$ be a quantum channel.
	\begin{enumerate}[1.]
		\item Alice has a classical register $\mathsf{M}$ and a quantum register $\A$, and Bob has a quantum register $\B$.
		
		\item To send any (equiprobable) message $m\in\mathsf{M}$, Alice encodes it into an input state on her quantum register $\A$.
		
		\item 
        Alice's quantum state on the register $\A$ undergoes the quantum channel $\mathscr{N}_{\A\to \B}$ and produces an output state on Bob's quantum register $\B$.
		
		\item Bob applies a decoding measurement $\{ {\Pi}_{\B}^m\}_{m\in\mathsf{M}}$ on his quantum register $\B$ to obtain an estimated message $\hat{m} \in \mathsf{M}$.
	\end{enumerate}
    As in Section~\ref{sec:CQ}, we employ the conventional random coding strategy but now with a quantum ensemble $\{ p_{\X}(x), \rho_{\A}^x \}_{x \in \X} $, where the classical register $\X$ has an arbitrary finite alphabet.
    We also represent the ensemble by a classical-quantum state \begin{align} \label{eq:CQ-state_unassisted}
        \rho_{\X\A} = \sum_{x\in\X} p_{\X}(x)|x\rangle \langle x|_{\X} \otimes \rho_{\A}^x.
    \end{align}
    Specifically, each message $m\in \mathsf{M}$ is encoded into a quantum codeword $\rho_{\A}^{x(m)}$, drawn pairwise independently according to $p_{\X}$.
    Then, the minimum \emph{random coding error probability} for sending $|\mathsf{M}|$  messages through the channel $\mathscr{N}_{\mathsf{A}\to\B}$ with the quantum ensemble $\rho_{\X\A}$ is
	\begin{align} \label{eq:error_unassisted}
    \varepsilon(\X:\B)_{\mathscr{N}(\rho)}
    &\coloneqq \mathds{E}_{\{x(m)\}\sim p_{\X}} \left[ \inf_{ \{ {\Lambda}_{\B}^{m} \}_{m\in\mathsf{M}} } \frac{1}{|\mathsf{M}|}  \sum_{m\in \mathsf{M}}	
     \Tr\left[ \mathscr{N}_{\A\to \B}\left(\rho_{\A}^{x(m)}\right) \left(\I_{\B} -  {\Lambda}_{\B}^{m} \right) \right] \right].
    \end{align}
\end{defn}

Evidently, our analysis for classical-quantum channel coding in Section~\ref{sec:CQ} applies to classical communication over any quantum channel $\mathscr{N}_{\A\to \B}$ by the substitution $\rho_{\B}^{x(m)} \leftarrow \mathscr{N}_{\A\to \B}\left( \rho_{\A}^{x(m)} \right)$. Hence, the integral $\alpha$-PGM is constructed according to the channel images:
\begin{align} \notag
    \Pi_{\B}^{x(m)} 
    \coloneqq\frac{\left[ \mathscr{N}_{\A\to \B}\left( \rho_{\A}^{x(m)} \right) \right]^{\alpha}}{\sum_{\bar{m}\in \mathsf{M} } \left[ \mathscr{N}_{\A\to \B}\left( \rho_{\A}^{x(\bar{m})} \right)\right]^{\alpha} } + \frac{1}{|\mathsf{M}|} \I_{\B^{\perp}}, \quad \forall \, m \in \mathsf{M}, \; \alpha \in [\sfrac{1}{2},1],
\end{align}
where $\mathscr{N}_{\A\to \B}\big( \rho_{\A}^{x(m)} \big) \I_{\B^{\perp}} = 0$ for all $m\in\mathsf{M}$.

We then get the following result.
\begin{theo} \label{theo:CQ_unassisted} 
	For any quantum channel $\mathscr{N}_{\mathsf{A} \to \B}$, the  random coding error probability \eqref{eq:error_unassisted} for sending $|\mathsf{M}|$ messages with a quantum ensemble $\rho_{\X\A} = \sum_{x\in\X} p_{\X}(x)|x\rangle\langle x|_{\X} \otimes \rho_{\A}^x$ is upper bounded by
	\begin{align} \notag
		\varepsilon(\X:\B)_{\mathscr{N}(\rho)}
		&\leq c_{\alpha}
        (|\mathsf{M}|-1)^{\frac{1-\alpha}{\alpha}}
        \Tr \left[
        \left( \sum_{x\in\X} p_{\X}(x) \left[ \mathscr{N}_{\A\to\B}\left(\rho_{\mathsf{A}}^x\right)\right]^{\alpha}  \right)^{\nicefrac{1}{\alpha}}
        \right]
        \\
        \notag
        &= c_{\alpha} \cdot
        2^{
        -\frac{1-\alpha}{\alpha} \left[ I_{\alpha} (\X : \B)_{\mathscr{N}(\rho)} - \log_2 (|\mathsf{M}|-1)  \right]
        }, \quad \forall\, \alpha \in [\nicefrac{1}{2},1].
	\end{align}
\end{theo}

We may define the (one-shot) \emph{order-$\alpha$ Holevo quantity} by optimizing over the input ensemble $\rho_{\X\A}$ at the encoder:
\begin{align} \label{eq:Holevo_quantity}
\chi_{\alpha}(\mathscr{N}_{\A\to\B})
\coloneqq\sup_{ \rho_{\X\A} \in \mathcal{S}(\X\otimes\A) } I_{\alpha} (\X : \B)_{\mathscr{N}(\rho)}.
\end{align}
For any channel whose Holevo quantity is additive, i.e., $\chi_{\alpha}\left(\mathscr{N}_{\A\to\B}^{\otimes n}\right) = n \chi_{\alpha}(\mathscr{N}_{\A\to\B})$,
our result in Theorem~\ref{theo:CQ_unassisted} yields exponential decay for \emph{any}  blocklength $n\in\mathds{N}$.
Even if the Holevo quantity is not additive, Theorem~\ref{theo:CQ_unassisted} implies a finite-blocklength bound: \begin{align}  \varepsilon(\X^n:\B^n)_{\mathscr{N}^{\otimes n}(\rho^n)} \leq 1.102 \cdot 2^{ - \sup_{\alpha \in [\nicefrac{1}{2},1]} \frac{1-\alpha}{\alpha} \left[ \chi_{\alpha} (\mathscr{N}_{\A\to\B}^{\otimes n}) - \log_2 (|\mathsf{M}|-1)  \right]}, \quad \forall\, n \in \mathds{N}. \end{align}
Moreover, one may choose a sequence of entangled input states in the ensemble so as to asymptotically achieve the \emph{regularized} Holevo quantity:\footnote{As in the order-$1$ scenario, by definition, $\chi_{\alpha}(\cdot)$ is superadditive. Hence, taking supremum over all $n \in \mathds{N}$ is equivalent to the limit superior as $n\to \infty$.}
\begin{align} \label{eq:Holevo-reg} \chi_{\alpha}^{\text{reg}}(\mathscr{N}_{\A\to\B}) \coloneqq\sup_{n\in \mathds{N}} \frac{1}{n} \chi_{\alpha}\left(\mathscr{N}_{\A\to\B}^{\otimes n}\right).
\end{align}
For theoretical analysis, our result shows that, for any rate $R \coloneqq\lim_{n\to \infty} \frac{1}{n} \log_2 |\mathsf{M}|$,  the following regularized exponent is achievable in principle: 
\begin{align}
\sup_{\alpha \in [\sfrac12,1]} \frac{1-\alpha}{\alpha} \left[ \chi_{\alpha}^{\text{reg}}(\mathscr{N}_{\A\to\B}) - R \right].
\end{align}
This error exponent is positive if and only if the communication rate is below the order-$1$ regularized Holevo quantity $R < \chi_{1}^{\text{reg}}(\mathscr{N}_{\A\to \B})$, which is known as the (asymptotic) classical capacity of the quantum channel $\mathscr{N}_{\A\to \B}$.
Unfortunately, computationally finding the input ensembles that achieve $\chi_{\alpha}^{\text{reg}}$ and practically implementing such an asymptotic encoding strategy remain highly challenging.
Nevertheless, our one-shot result established in Theorem~\ref{theo:CQ_unassisted} applies to any input ensemble $\rho_{\X^n \A^n}$ for any $n\in\mathds{N}$ (some of which may be  easier to implement in practice).
Moreover, for product channels (not necessarily stationary) and using only product input states,  Theorem~\ref{theo:CQ_unassisted} already demonstrates exponential decay for any blocklength $n\in\mathds{N}$.

\section{Entanglement-Assisted Classical Communication} \label{sec:EA}

\begin{defn}[Entanglement-assisted classical communication over quantum channels] 
	Let $\mathscr{N}_{\A\to \B}\colon\mathcal{S}(\A) \to \mathcal{S}(\B)$ be a quantum channel.
	\begin{enumerate}[1.]
		\item Alice has a classical register $\mathsf{M}$ with cardinality $M\coloneqq|\mathsf{M}|$ representing the message set and has quantum registers $\A$ and $\A'$.
        Bob has quantum registers $\B$ and $\R'$.\footnote{The quantum register $\R'$ at Bob is viewed as the \emph{reference system} of $\A'$ that is left entirely untouched and evolves trivially (via the identity map $\id_{\R'}$). At the end of the protocol, Bob employs $\R'$ together with his received noisy system $\B$ to jointly decode the message. The correlation between $\R'$ and $\B$ thus determines the communication performance.}
		
		\item An arbitrary state $\theta_{\R'\A'}$ is shared between Bob (holding $\R'$) and Alice (holding $\A'$) as a resource for assisting communication.
		
		\item To send any (equiprobable) message $m\in\mathsf{M}$, Alice applies an encoding quantum operation $\mathcal{E}_{\A'\to \A}^m$ on $\theta_{\R'\A'}$.
		
		\item 
        Alice's quantum state on the register $\A$ undergoes the quantum channel $\mathscr{N}_{\A\to \B}$ and produces an output state on Bob's quantum register $\B$.
		
		\item Bob applies a decoding measurement $\{ {\Lambda}_{\R'\B}^m\}_{m\in\mathsf{M}}$ on registers $\R'$ and $\B$ to obtain an estimated message $\hat{m} \in \mathsf{M}$.
	\end{enumerate}
    We adopt the encoder of the \emph{position-based coding} \cite{AJW19a} as follows.
    Consider $\R'\A' = \R_1\cdots\R_M \A_1\cdots\A_M$.
    Alice and Bob pre-share an $M$-fold product state $\theta_{\R'\A'} \coloneqq\theta_{\R\A}^{\otimes M} = \theta_{\R_1\A_1}\otimes \cdots \otimes \theta_{\R_M\A_M}$, where $\A_m \cong \A$ and $\R_m \cong \R$ for each $m\in\mathsf{M}$, and Bob holds the reference systems $\R_1\cdots\R_M$ corresponding to Alice's systems $\A_1\cdots\A_M$.
	To send each $m\in\mathsf{M}$, 
	Alice sends her system $\A_m$, i.e., $\mathcal{E}_{\A^M\to \A}^m = \Tr_{\A^{\mathsf{M}\backslash \{m\}}}$, by tracing out systems $\A_{\bar{m}}$ for all $\bar{m}\neq  m$. 	
    The minimum error probability for sending $M$ messages through the channel $\mathscr{N}_{\A\to \B}$ with the assistance of ($M$ copies of) the state $\theta_{\R\A}$ is defined as
	\begin{align} \label{eq:error_EA}
    \varepsilon(\R:\B)_{\mathscr{N}(\theta)}
    \coloneq
    \inf_{ \left\{ {\Lambda}_{\R'\B}^m\right\}_{m\in\mathsf{M}} }
		\frac{1}{M} \sum_{m\in \mathsf{M}}	\Tr\left[  \mathscr{N}_{\A\to\B} \circ \mathcal{E}_{\A'\to \A}^m\left(\theta_{\R'\A'}\right)
        \left( \I_{\R'\B}-  {\Lambda}_{\R'\B}^m \right) \right].
	\end{align}
    The minimum error over all assisting entangled states $\theta_{\R\A}$ in the position-based coding is defined as
    \begin{align} \label{eq:error_EAC}
    \varepsilon_{\text{EAC}}^\star\left(M;\mathscr{N}\right)
    \coloneq
    \inf_{\theta_{\R\A}} \varepsilon(\R:\B)_{\mathscr{N}(\theta)}.
    \end{align}
\end{defn}

Bob applies the integral $\alpha$-PGM corresponding to the channel output states:
	\begin{align} 
        \notag
		\mathring{\Pi}_{\R_1 \ldots \R_M \B}^m 
        &\coloneqq\frac{ \left(\rho_{\R_1 \ldots \R_M  \B}^m\right)^{\alpha} }{ \sum_{\bar{m}\in\mathsf{M}} \left(\rho_{\R_1 \ldots \R_M \B}^{\bar{m}}\right)^{\alpha} } 
        + \frac{1}{|\mathsf{M}|} \I_{(\R_1 \ldots \R_M  \B)^{\perp}} \quad \forall \, m\in\mathsf{M},
        \; \alpha \in [\nicefrac{1}{2},1],
        \\
        \notag
		\rho_{\R_1 \ldots \R_M \B}^m 
        &\coloneqq\theta_{\R}^{\otimes (m-1)} \otimes \mathscr{N}_{\A\to\B}(\theta_{\R_m\A_m}) \otimes \theta_{\R}^{\otimes (M-m)}, \quad \forall\, m\in\mathsf{M},
	\end{align}
where $\rho_{\R_1 \ldots \R_M \B}^m \I_{(\R_1 \ldots \R_M  \B)^{\perp}} = 0$ for all $m\in\mathsf{M}$.

\begin{theo} \label{theo:EA} 
	For any quantum channel $\mathscr{N}_{\A \to \B}$, the error probability \eqref{eq:error_EA} for sending $M$ messages with the assistance state $\theta_{\R\A}$ is upper bounded by
	\begin{align} \notag
		\varepsilon(\R:\B)_{\rho}
		&\leq c_{\alpha}
        (M-1)^{\frac{1-\alpha}{\alpha}}
        \Tr_{\B} \left[
        \left( \Tr_{\mathsf{R}} \left[ \rho_{\mathsf{RB}}^{\alpha} \theta_{\mathsf{R}}^{1-\alpha} \right]
        \right)^{\nicefrac{1}{\alpha}}
        \right]
        \\
        \notag
        &= c_{\alpha} \cdot
        2^{
        -\frac{1-\alpha}{\alpha} \left[ I_{\alpha} (\R : \B)_{\mathscr{N}(\theta)} - \log_2 (M-1) \right]
        }, \quad \forall\, \alpha \in [\nicefrac{1}{2},1].
	\end{align}
    Here, $\rho_{\mathsf{RB}} \coloneqq\mathscr{N}_{\A \to \B}\left(\theta_{\R\A}\right)$ with $\rho_{\R} = \theta_{\R}$
    and $I_{\alpha} (\R:\B)_{\rho} \coloneqq\inf_{\sigma_{\B} \in \mathcal{S}(\B) } D_{\alpha}(\rho_{\R\B}\Vert \rho_{\R} \otimes \sigma_{\B}) $ is the order-$\alpha$ Petz--\Renyi information.
\end{theo}

The quantum Sibson identity \cite[(3.10)]{HT14}, \cite{CGH18} shows that the minimizer in $I_{\alpha} (\R:\B)_{\rho}$ is attained by
\[
\sigma_{\B}^{\star} = \frac{
        \left( \Tr_{\mathsf{R}} \left[ \rho_{\mathsf{RB}}^{\alpha} \rho_{\mathsf{R}}^{1-\alpha} 
        \right]\right)^{\nicefrac{1}{\alpha}}
        }{
        \Tr_{\B} \left[
        \left( \Tr_{\mathsf{R}} \left[ \rho_{\mathsf{RB}}^{\alpha} \rho_{\mathsf{R}}^{1-\alpha} \right]
        \right)^{\nicefrac{1}{\alpha}}
        \right]
        },
\]
and, hence, the order-$\alpha$ Petz--\Renyi information admits a closed-form expression:
\[
I_{\alpha} (\R:\B)_{\rho} 
= D_{\alpha}\left(\rho_{\R\B}\Vert \rho_{\R} \otimes \sigma_{\B}^{\star}\right)
=
\frac{\alpha}{\alpha-1} \log_2 \Tr_{\B} \left[
        \left( \Tr_{\mathsf{R}} \left[ \rho_{\mathsf{RB}}^{\alpha} \rho_{\mathsf{R}}^{1-\alpha} \right]
        \right)^{\nicefrac{1}{\alpha}}
        \right].
\]

\begin{remark} \label{remark:first_exp_EA}
The early developments of entanglement-assisted classical communication can be traced back to 
\cite{BSS+99, BSS+02, Hol02, DH13, MW14}.
The first i.i.d.~asymptotic exponential-decay bound of the error probability for the entanglement-assisted setting was implied by the second-order asymptotics of the maximal achievable rate given in \cite[Proposition 14]{DTW16}.
Later, \cite[Theorem 6]{QWW18} showed a one-shot bound using the position-based coding \cite{AJW19a} (see also \cite[Theorem 2]{Cheng_simple}), whose achievable error exponent has a form similar to classical-quantum channel coding \cite{Hay07}, i.e.,~for any $\theta_{\R\A}$,
\begin{align} \notag
\sup_{\alpha \in [\nicefrac{1}{2},1]} \frac{1-\alpha}{\alpha} \left[ D_{2-\frac{1}{\alpha}}\left( \mathscr{N}_{\A\to\B}(\theta_{\R\A}) \Vert \theta_{\R} \otimes \mathscr{N}_{\A\to\B}(\theta_{\A}) \right) - R \right].
\end{align}
The above quantity can be related to Theorem~\ref{theo:EA} by applying Proposition~\ref{prop:comparison_Hayashi} as discussed in Section~\ref{sec:large_deviation}.
\end{remark}

\begin{proof}
By symmetry of the position-based encoding, we calculate the error probability for sending $m=1$ without loss of generality:
\begin{align*}
\varepsilon(\R:\B)_{\rho}
&\leq \Tr\left[ \rho_{\R_1 \ldots \R_M \B}^1 \frac{ \sum_{\bar{m}\neq 1} \left(\rho_{\R_1 \ldots \R_M \B}^{\bar{m}}\right)^{\alpha} }{ 
\left(\rho_{\R_1 \ldots \R_M \B}^{1}\right)^{\alpha}
+ \sum_{\bar{m}\neq 1} \left(\rho_{\R_1 \ldots \R_M \B}^{\bar{m}}\right)^{\alpha}}
\right]
\\
&\leq c_{\alpha} \cdot \Tr\left[ 
\rho_{\mathsf{R}_1 \B}^{\alpha} \otimes \theta_{\mathsf{R}_2}^{\alpha} \otimes \theta_{\mathsf{R}_3}^{\alpha} \otimes \cdots \otimes \theta_{\mathsf{R}_M}^{\alpha}
\cdot
\left( 
\sum_{\bar{m}= 2}^M \rho_{\mathsf{R}_{\bar{m}} \B}^{\alpha} \bigotimes\limits_{ m'\neq \bar{m} } \theta_{\mathsf{R}_{m'}}^{\alpha}
\right)^{\frac{1-\alpha}{\alpha}}
\right].
\end{align*}
Here, we again employ Proposition~\ref{prop:key} with
\begin{align*}
A &\leftarrow \rho_{\R_1 \ldots \R_M \B}^1 = \rho_{\mathsf{R}_1 \B }\otimes \theta_{\mathsf{R}_2} \otimes \theta_{\mathsf{R}_3} \otimes \cdots \otimes \theta_{\mathsf{R}_M},
\\
B^{\alpha} &\leftarrow \sum_{\bar{m}\neq 1} \left(\rho_{\R_1 \ldots \R_M \B}^{\bar{m}}\right)^{\alpha}
= 
\sum_{\bar{m}= 2}^M \rho_{\mathsf{R}_{\bar{m}} \B}^{\alpha} \bigotimes\limits_{ m'\neq \bar{m} } \theta_{\mathsf{R}_{m'}}^{\alpha}.
\end{align*}

To evaluate the trace term, we apply an operator Jensen inequality, proved in Lemma~\ref{lemm:operator_Jensen} below with 
${\mathsf{A}} \leftarrow \mathsf{R}_1 \B$ and
$\B \leftarrow \mathsf{R}_2 \mathsf{R}_3 \ldots \mathsf{R}_M$,
\begin{align*}
    Y_{\mathsf{AB}} 
    &\leftarrow \sum_{\bar{m}= 2}^M \rho_{\mathsf{R}_{\bar{m}} \B}^{\alpha} \bigotimes\limits_{ m'\neq \bar{m} } \theta_{\mathsf{R}_{m'}}^{\alpha},
    \\
    \tau_{\B}
    &\leftarrow \theta_{\mathsf{R}_2} \otimes \theta_{\mathsf{R}_3} \otimes\cdots \otimes \theta_{\mathsf{R}_M},
\end{align*} 
to obtain
\begin{equation*}
\begin{aligned}[b]
&\Tr\left[ 
\rho_{\mathsf{R}_1 \B}^{\alpha} \otimes \theta_{\mathsf{R}_2}^{\alpha} \otimes \theta_{\mathsf{R}_3}^{\alpha} \otimes \cdots \otimes \theta_{\mathsf{R}_M}^{\alpha}
\cdot
\left( 
\sum_{\bar{m}= 2}^M \rho_{\mathsf{R}_{\bar{m}} \B}^{\alpha} \bigotimes\limits_{ m'\neq \bar{m} } \theta_{\mathsf{R}_{m'}}^{\alpha}
\right)^{\frac{1-\alpha}{\alpha}}
\right]
\\
&=
\Tr_{\mathsf{R}_1 \B} \left[
\rho_{\mathsf{R}_1 \B}^{\alpha} \cdot
\Tr_{\mathsf{R}_2 \ldots \mathsf{R}_M}\left[
\theta_{\mathsf{R}_2}^{\alpha} \otimes \cdots \otimes \theta_{\mathsf{R}_M}^{\alpha} \cdot
\left( 
\sum_{\bar{m}= 2}^M \rho_{\mathsf{R}_{\bar{m}} \B}^{\alpha} \bigotimes\limits_{ m'\neq \bar{m} } \theta_{\mathsf{R}_{m'}}^{\alpha}
\right)^{\frac{1-\alpha}{\alpha}}
\right]
\right]
\\
&\leq
\Tr_{\mathsf{R}_1 \B} \left[
\rho_{\mathsf{R}_1 \B}^{\alpha} \cdot
\left( \sum_{\bar{m}= 2}^M \theta_{\mathsf{R}_1}^{\alpha}\otimes \Tr_{ \mathsf{R}_{\bar{m}} } \left[ \rho_{\mathsf{R}_{\bar{m}} \B}^{\alpha} \theta_{\mathsf{R}_{\bar{m}}}^{1-\alpha} \right] \right)^{\frac{1-\alpha}{\alpha}}
\right]
\\
&=
(M-1)^{\frac{1-\alpha}{\alpha}} \cdot 
\Tr_{\mathsf{R}_1 \B} \left[ \rho_{\mathsf{R}_1 \B}^{\alpha} \theta_{\mathsf{R}_1}^{1-\alpha} \cdot
\left( \Tr_{\mathsf{R}}\left[ \rho_{\mathsf{R} \B}^{\alpha} \theta_{\mathsf{R}}^{1-\alpha} \right]
\right)^{\frac{1-\alpha}{\alpha}}
\right]
\\
&=
(M-1)^{\frac{1-\alpha}{\alpha}} \cdot 
\Tr_{\B} \left[
\left( \Tr_{\mathsf{R}} \left[ \rho_{\mathsf{RB}}^{\alpha} \theta_{\mathsf{R}}^{1-\alpha} \right]
\right)^{\frac{1}{\alpha}}
\right].
\end{aligned}\qedhere
\end{equation*}
\end{proof}

\begin{lemm} \label{lemm:operator_Jensen}
    For any bounded operator $Y_{\mathsf{AB}}  \geq 0$ on a Hilbert space $\mathcal{H}_{\mathsf{A}} \otimes 
    \mathcal{H}_{\B}$
    and
    any normalized state $\tau_{\B}$ on $\mathcal{H}_{\B}$,
    we have
    \begin{align*}
    \Tr_{\B} \left[ \tau_{\B}^{\alpha}\cdot 
    Y_{\mathsf{AB}}^{\frac{1-\alpha}{\alpha}}  
    \right]
    \leq 
    \left( \Tr_{\B} \left[ Y_{\mathsf{AB}} \cdot \tau_{\B}^{1-\alpha}\right] \right)^{\frac{1-\alpha}{\alpha}}, \quad \forall\, \alpha \in [\nicefrac{1}{2},1].
    \end{align*}
\end{lemm}
\begin{proof}
Denote the spectral decomposition of $\tau_{\B}$ by
$\tau_{\B} = \sum_i \lambda_i |i\rangle_{\B}\langle i|_{\B}$, where $\lambda_i\geq 0$ for each $i$, $\sum_i \lambda_i = 1$, and $\{|i\rangle_{\B}\}_i$ is an orthonormal basis of $\mathcal{H}_{\B}$.
We calculate\footnote{By convention, the power in $\lambda_i^{-\alpha}$ is understood as taken on the support of $\lambda_i$. If $\lambda_i=0$ for some $i$, then the term $\lambda_i^{-\alpha}$ is void.}
\begin{align*}\Tr_{\B} \left[ \tau_{\B}^{\alpha}\cdot 
    Y_{\mathsf{AB}}^{\frac{1-\alpha}{\alpha}}  
    \right]
&=
\sum_i \lambda_i^{\alpha} \cdot \mathbf{1}_{\mathsf{A}}\otimes \Big\langle i \Big|_{\B} Y_{\mathsf{AB}}^{\frac{1-\alpha}{\alpha}} \mathbf{1}_{\mathsf{A}}\otimes\Big|i\Big\rangle_{\B}
\\
&=
\sum_i \lambda_i \cdot \mathbf{1}_{\mathsf{A}}\otimes \Big\langle i \Big|_{\B} \left( Y_{\mathsf{AB}} \cdot \lambda_i^{-\alpha} \right)^{\frac{1-\alpha}{\alpha}} \mathbf{1}_{\mathsf{A}}\otimes \Big|i\Big\rangle_{\B}
\\
&\leq 
\left( 
\sum_i \lambda_i \cdot \mathbf{1}_{\mathsf{A}}\otimes \Big\langle i \Big|_{\B} Y_{\mathsf{AB}} \cdot \lambda_i^{-\alpha} \mathbf{1}_{\mathsf{A}}\otimes \Big|i\Big\rangle_{\B}
\right)^{\frac{1-\alpha}{\alpha}}
\\
&=
\left( \Tr_{\B} \left[ Y_{\mathsf{AB}} \cdot \tau_{\B}^{1-\alpha}\right] \right)^{\frac{1-\alpha}{\alpha}}.
\end{align*}
Here, the inequality follows from the operator concavity of $0\leq x\mapsto x^{\frac{1-\alpha}{\alpha}}$ for $\frac{1-\alpha}{\alpha} \in [0,1]$ and the
operator Jensen inequality \cite{HP03}, which states that,
for any operator concave function $f$,
any sequence $(X_1, X_2, \ldots)$
of bounded self-adjoint operators on a Hilbert space $\mathcal{H}$ supported on the domain of $f$,
and any sequence
$(C_1, C_2, \ldots)$ of bounded operators from $\mathcal{K}$ to $\mathcal{H}$ satisfying $\sum_i C_i^{\dagger} \mathbf{1}_{\mathcal{H}} C_i = \mathbf{1}_{\mathcal{K}}$, 
\begin{align*}
\sum_i C_i^{\dagger} f(X_i) C_i
\leq
f\left( \sum_i C_i^{\dagger}  X_i C_i\right).
\end{align*}
In view of 
$\mathcal{K} \leftarrow \mathcal{H}_\mathsf{A}$,
$\mathcal{H} \leftarrow \mathcal{H}_{\mathsf{A}}\otimes \mathcal{H}_{\B}$,
$C_i \leftarrow \mathbf{1}_{\mathsf{A}} \otimes \sqrt{\lambda_i} |i\rangle_{\B}$ with 
$\sum_i C_i^{\dagger} \cdot \mathbf{1}_{\mathsf{A}} \otimes \mathbf{1}_{\B} \cdot
C_i = \sum_i \mathbf{1}_{\mathsf{A}} \otimes \lambda_i = \mathbf{1}_{\mathsf{A}}$, 
and $X_i \leftarrow Y_{\mathsf{AB}} \cdot \lambda_i^{-\alpha}$,
the proof is completed.
\end{proof}

\section{Entanglement-Assisted Quantum Communication} \label{sec:EAQ}

In this section, we consider entanglement-assisted quantum communication over a quantum channel $\mathscr{N}_{\A\to\B}$. The goal is to transmit quantum information residing in a quantum register $\Q$ instead of sending classical messages in $\mathsf{M}$.
There are different notions of the error criteria that have been considered in the literature \cite{BKN00, KW04}.
For simplicity, we consider transmitting one part of the maximally entangled state, called \emph{entanglement transmission}, and discuss its relation to \emph{strong subspace transmission} in Remark~\ref{remark:quantum_criteria} later.
We refer the reader to the monograph \cite[\S 14]{KW20} for the relevant notation.

\begin{defn}[Entanglement-assisted quantum communication over quantum channels] 
	Let $\mathscr{N}_{\A\to \B}\colon\mathcal{S}(\A) \to \mathcal{S}(\B)$ be a quantum channel,
    and let $\Phi_{\underline{\Q}\Q} \coloneqq |\Phi\rangle\langle \Phi |_{\underline{\Q}\Q}$
    be the maximally entangled state, i.e.,
    \begin{align}
        |\Phi\rangle_{\underline{\Q}\Q}
        \coloneqq \frac{1}{\sqrt{|\Q|}} \sum_{i=0}^{|\Q|-1} |i\rangle_{\underline{\Q}} \otimes |i\rangle_{\Q}.
    \end{align}
	\begin{enumerate}[1.]
		\item Alice has a quantum register $\Q$ with dimension $|\Q|$
        and quantum registers $\A$ and $\A'$.
        Bob has quantum registers $\B$ and $\R'$.
		
		\item An arbitrary state $\theta_{\R'\A'}$ is shared between Bob and Alice as a resource for assisting communication.
        In addition, we suppose that shared randomness, a strictly weaker resource, is also available.
		
		\item 
        Alice applies an encoding quantum operation $\mathcal{E}_{\Q\A'\to \A}$ on the $\Q$ system of $\Phi_{\underline{\Q}\Q}$ and the $\A'$ system of $\theta_{\R'\A'}$.
		
		\item 
        Alice's quantum state on the register $\A$ undergoes the quantum channel $\mathscr{N}_{\A\to \B}$ and produces an output state on Bob's quantum register $\B$.
		
		\item Bob applies a decoding operation $\mathscr{D}_{\R'\B\to\Q}$ on registers $\R'$ and $\B$ to recover the $\Q$ system of $\Phi_{\underline{\Q}\Q}$.
	\end{enumerate}
    The minimum error of transmitting entanglement of rank $|\Q|$ over $\mathscr{N}_{\A\to\B}$ is defined as
    \begin{align} \label{eq:error_EAQ}
    \varepsilon_{\text{EAQ}}^\star(|\Q|;\mathscr{N})
    \coloneqq \inf_{\mathscr{E},\mathscr{D},\theta} 1 - \langle \Phi |_{\underline{\Q}\Q} (\mathscr{D} \circ \mathscr{N} \circ \mathscr{E} \otimes \id_{\underline{\Q}}) ( \Phi_{\underline{\Q}\Q} \otimes \theta_{\R'\A'}) | \Phi\rangle_{\underline{\Q}\Q},
\end{align}
where the second term is the (squared) \emph{entanglement fidelity} $F_{\text{e}}$ of the overall operation $(\mathscr{D} \circ \mathscr{N} \circ \mathscr{E} \otimes \id_{\underline{\Q}}) ( \,\cdot\, \otimes \theta_{\R'\A'})$.\footnote{The entanglement fidelity of a channel $\mathscr{M}_{\Q\to\Q}$ is defined as 
$F_{\text{e}}(\mathscr{M}) \coloneqq
\langle \Phi_{\underline{\Q}\Q} | \left( \id_{\underline{\Q}} \otimes \mathscr{M}_{\Q\to\Q} \right) (\Phi_{\underline{\Q}\Q})| \Phi \rangle_{\underline{\Q}\Q}
$.
}
\end{defn}

Leung and Matthews showed that the minimum error of entanglement transmission in \eqref{eq:error_EAQ} is related to the error probability of entanglement-assisted classical communication via the formula \cite[(46) and Appendix B]{LM15} (see also \cite{LY24b}):
\begin{align} \label{eq:EAQ_EAC}
     \varepsilon_{\text{EAQ}}^\star(|\Q|;\mathscr{N})
     =  \varepsilon_{\text{EAC}}^\star\left(|\Q|^2;\mathscr{N}\right).
\end{align}
Hence, by Theorem~\ref{theo:EA}, we directly have the following upper bound. 

\begin{theo} \label{theo:EAQ} 
	For any quantum channel $\mathscr{N}_{\A \to \B}$, the minimum error \eqref{eq:error_EAQ} for transmitting entanglement of rank $|\Q|$ with entanglement assistance is upper bounded by
	\begin{align} 
    \notag
        \varepsilon_{\textnormal{EAQ}}^\star(|\Q|;\mathscr{N})
		&\leq c_{\alpha}
        \left(|\Q|^2-1\right)^{\frac{1-\alpha}{\alpha}}
        \Tr_{\B} \left[
        \left( \Tr_{\mathsf{R}} \left[ \rho_{\mathsf{RB}}^{\alpha} \theta_{\mathsf{R}}^{1-\alpha} \right]
        \right)^{\nicefrac{1}{\alpha}}
        \right]
        \\
        &\leq c_{\alpha} \cdot
        2^{
        -\frac{1-\alpha}{\alpha} \left[ I_{\alpha} (\R : \B)_{\mathscr{N}(\theta)} - 2 \log_2 |\Q| \right]
        }, \quad \forall\, \theta_{\R\A} \in \mathcal{S}(\R\A), \alpha \in [\nicefrac{1}{2},1].
        \label{eq:EAQ}
	\end{align}
    Here, $\rho_{\mathsf{RB}} \coloneqq\mathscr{N}_{\A \to \B}\left(\theta_{\R\A}\right)$ with $\rho_{\R} = \theta_{\R}$
    and $I_{\alpha} (\R:\B)_{\rho} \coloneqq\inf_{\sigma_{\B} \in \mathcal{S}(\B) } D_{\alpha}(\rho_{\R\B}\Vert \rho_{\R} \otimes \sigma_{\B}) $ is the order-$\alpha$ Petz--\Renyi information.
\end{theo}

\begin{proof}
    To achieve the upper bound in \eqref{eq:EAQ} for any $\theta_{\R\A}$ and $\alpha \in [\nicefrac{1}{2},1]$, 
    we implement the position-based coding employed in Theorem~\ref{theo:EA} with assisting entanglement $\theta_{\R\A}^{\otimes M}$, $M = |\Q|^2$, the teleportation protocol \cite{BBC+93} with a maximally entangled state of rank $|\Q|$, and the symmetrization process via twirling (which only consumes shared randomness without extra entanglement) detailed in \cite[Appendix B]{LM15}.
    The resulting coded channel is a depolarizing channel whose infidelity is bounded by Theorem~\ref{theo:EA} \cite[(78)]{LM15}.
    We complete the proof.
\end{proof}

\begin{remark}[Equivalent one-shot criteria after shared-randomness twirling] \label{remark:quantum_criteria}
A stronger notion of quantum communication is \emph{strong subspace transmission}; namely, the entanglement fidelity $F_{\text{e}}$ in \eqref{eq:error_EAQ} is replaced by the following (squared) fidelity of a channel $\mathscr{M}_{\Q\to\Q}$:
\begin{align}
F(\mathscr{M}) \coloneqq \inf_{\psi_{\underline{\Q}\Q}}
\langle \psi |_{\underline{\Q}\Q} \left( \id_{\underline{\Q}} \otimes \mathscr{M}_{\Q\to\Q} \right) (\psi_{\underline{\Q}\Q})| \psi \rangle_{\underline{\Q}\Q},
\end{align}
where the infimum is over all pure states on a reference system
\(\underline{\Q}\) and the input system \(\Q\).
The standard references establish asymptotic quantum-capacity equivalence between entanglement transmission and strong subspace transmission; see Barnum--Knill--Nielsen \cite{BKN00} and Kretschmann--Werner \cite{KW04}.
In the entanglement-assisted setting, 
the exact one-shot equality for the two criteria follows from applying the twirling argument of Leung--Matthews \cite[Appendix B]{LM15} and observing that for depolarizing channels, the worst-case entangled-input ($\psi_{\underline{\Q}\Q}$) fidelity is attained on the maximally entangled state ($\Phi_{\underline{\Q}\Q}$).

This is analogous to the relation between average and maximal error probability in classical communication with shared randomness: a random relabeling of the messages symmetrizes the code, making all messages have the same error probability.
Here, the shared-randomness twirl symmetrizes the logical quantum
channel, making the worst input ($\psi_{\underline{\Q}\Q}$) no worse than the maximally entangled input ($\Phi_{\underline{\Q}\Q}$).
\end{remark}

\section{Conclusions} \label{sec:conclusions}

We resolved Burnashev and Holevo's 1998 conjecture for classical-quantum channels with a dimension-independent prefactor $c< 1.102$ and showed that a similar form of the random coding bound holds even beyond classical-quantum channels to include arbitrary fully quantum channels for sending classical information with and without entanglement assistance and transmitting quantum information with entanglement assistance.
The same reasoning naturally extends to constant-composition codes and classical data compression with quantum side information under fixed-length coding or variable-length coding.

The general proof recipe inherits Shannon and Gallager's random coding principle---employing random coding and a union-bound-type argument (via PGMs) to reduce the channel output ensemble to a proper binary quantum hypothesis testing problem.
Our key contribution is to show that the integral $\alpha$-PGM decomposes into a family of Holevo--Helstrom measurements over priors $(u,1-u)$ with the parameter $u$ integrated uniformly. This demonstrates that the resulting test is essentially optimal, up to a dimension-independent constant factor.
This advantage allows standard techniques developed in binary hypothesis testing to be naturally applied here, thereby yielding the optimal tilting in large deviation analysis.

The operator layer cake theorem (Theorem~\ref{theo:Dlog_formula}) not only serves as the main technique for proving the error exponents for various quantum packing-type problems (i.e.,~those problems in Table~\ref{table:survey_packing} whose error exponents are associated with $\alpha <1$), but also leads to quantum covering-type results (e.g., classical-quantum soft covering, convex splitting, privacy amplification, quantum information decoupling, and quantum channel simulation), whose relative-entropy error exponents are associated with $\alpha >1$; see Table~\ref{table:survey_covering} and the recent follow-up work \cite{sharp25}.
Finally, Theorem~\ref{theo:Dlog_formula} also provides an alternative proof of 
Frenkel's integral formula for quantum relative entropy \cite{Fre23}; see \cite{LHC25_layer_cake}.
Hence, Theorem~\ref{theo:Dlog_formula} established in this work may be of independent interest.


\section*{Acknowledgments}
We sincerely thank the anonymous reviewers of the \emph{13th Beyond IID in Information Theory} (2025) for giving us very insightful and helpful comments and suggestions, and for pointing out typos in our first version.
We are very grateful to Prof.~Fumio Hiai for teaching us the use of the Powers--Størmer inequality in Proposition~\ref{prop:continuity}.
We are supported by the Emerging Young Scholars Program of the National Science and Technology Council, Taiwan (R.O.C.) under Grant numbers~NSTC 114-2628-E-002-006, NSTC 114-2119-M-001-002, and NSTC 114-2124-M-002-003, by the Yushan Young Scholar Program of the Ministry of Education, Taiwan (R.O.C.) under Grant number~NTU-114V2016-1, and by the research project `Forefront Quantum Computing, Learning, and Engineering in Noisy Intermediate-Scale Quantum Era’ of National Taiwan University under Grant NTU-114L895005. H.-C.~Cheng acknowledges the support from the `Center for Advanced Computing and Imaging in Biomedicine (NTU-115L900702)' through The Featured Areas Research Center Program within the framework of the Higher Education Sprout Project by the Ministry of Education (MOE) in Taiwan.


\newpage 
\appendix
\section{Properties of the Operator Logarithm} \label{sec:log}

In this section, we record the properties of the operator logarithm.
We denote the \emph{principal logarithm} as follows:
\begin{align} \label{eq:defn:Log}
	\Log (z) \coloneqq\log (|z|) + \mathrm{i} \Arg(z) 
	= \log (r) + \mathrm{i} \theta, \quad \forall z = r \mathrm{e}^{\mathrm{i}\theta} \in \mathds{C}\backslash\mathds{R}_{\leq 0},
\end{align}
where $r>0$, $-\pi < \theta < \pi$, and $u\mapsto \log(u) $ on $(0,\infty)$ denotes the real natural logarithm.

For an operator $A$ such that $\spec(A)\subset\mathds C\setminus\mathds{R}_{\leq 0}$, the principal logarithm is well defined on \(\spec(A)\).
We then define the operator logarithm via the holomorphic functional calculus (see e.g.~\cite[\S 3]{HP14}, \cite[\S 11]{Hig08}, \cite[(5.47), p.~44]{Kat95}):
\begin{align}
	\Log(A) \coloneqq\frac{1}{2\pi \mathrm{i}}\ointctrclockwise_{C_A} \Log(z)\left( z \I - A \right)^{-1}\, \d z,
\end{align}
where \(C_A\) is a simple closed counterclockwise contour enclosing
\(\spec(A)\) and avoiding the branch cut.
If, furthermore, \(A>0\) is positive definite, then this coincides with
the usual operator logarithm, i.e.,~\(\Log(A)=\log(A)\).

We denote the G{\^{a}}teaux differential (i.e.,~the directional derivative) of the principal logarithm at $A$ along the direction $B$ by
\begin{align} \label{eq:defn:DLog}
	\mathrm{D}\Log[A](B)
	&\coloneqq\lim_{t\to 0} \frac{\Log(A+tB) - \Log(A)}{t}.
\end{align}
When $A>0$ is positive definite, we adopt the notation $\mathrm{D}\Log[A](B) = \mathrm{D}\log[A](B)$.

\begin{fact}
	Let $A$ be an operator with $\spec(A) \subset \mathds{C}\backslash \mathds{R}_{\leq 0}$.
	The G{\^{a}}teaux differential of the principal logarithm in \eqref{eq:defn:DLog} satisfies the following properties.
	\begin{enumerate}[(i)]
		\item \label{item:DLog_identity}
		It holds that $\mathrm{D} \Log[A](A) = \I$.
		
		

		\item \label{item:Dlog_linear}
		The principal logarithm is \emph{continuously Fr\'echet differentiable} 
        in a neighborhood of any $A$ with $\spec(A) \subseteq \mathds{C}\backslash \mathds{R}_{\leq0}$ and, hence, the \emph{Fr\'echet derivative}, $B \mapsto \mathrm{D}\Log[A](B)$, is a linear map.

        \item \label{item:DLog_positivity}
		For $A>0$, the \emph{Fr\'echet derivative}, $B \mapsto \mathrm{D}\log[A](B)$, is a (completely) positive map.
		Namely, $\mathrm{D}\log[A](B) \geq 0 $ for any $B\geq 0$.
		
		\item \label{item:Dlog_integral}
		(Lieb's formula) The G{\^{a}}teaux differential admits the following integral representations:
		\begin{align}
        \label{eq:log-formula1}
			\mathrm{D}\Log[A](B)
			&= \int_{0}^{\infty} \frac{1}{A+t\I} B \frac{1}{A+t\I} \,\mathrm{d}t
			\\
        \label{eq:log-formula2}    
			&= \int_{0}^{1} \frac{1}{tA+(1-t)\I} B \frac{1}{tA+(1-t)\I} \,\mathrm{d}t.
		\end{align}
		
		\item \label{item:Dlog_division}
		If $A>0$, we have
		\begin{align}
			\mathrm{D}\log[A](B)
			&=            A^{-1/2}\mathrm{D}\log\left[A^{-1}\right]\left(A^{-1/2} B A^{-1/2} \right)A^{-1/2}.
		\end{align}

        \item \label{item:Beigi-Tomamichel}
        (Beigi--Tomamichel's inequality)
        For $A>0$ and $B\geq 0$,
        \begin{align}
        \mathrm{D}\log[A+B](B) 
        \leq \mathrm{D} \log [A] (B).
        \end{align}

        \item \label{item:Dlog_norm}
        If $A>0$ and $B\geq 0$, then 
        \begin{align}
        \|\mathrm{D}\log[A](B)\|_\infty\leq \|A^{-1/2}BA^{-1/2}\|_\infty.
        \end{align}

	\end{enumerate}
\end{fact}
\begin{proof}
    Item~\ref{item:DLog_identity} follows from the definition.
    Item~\ref{item:Dlog_linear} follows from \cite{Ped00}, \cite[\S 6]{Peller2006} (see also \cite{Daleckii1956}).
    
    For item~\ref{item:Dlog_integral},
    the formula \eqref{eq:log-formula1} was first pointed out by Lieb \cite{Lie73}
    by using the integral formula
    $\Log(A) 
    = \int_0^{\infty} (1+t)^{-1}(A-\I)(A+t\I)^{-1}\d t
    =
    \int_0^{\infty} \left[ (1+t)^{-1}\I - (A+t\I)^{-1}\right]\d t$
    (see also \cite[(3.11)]{HP14}).
    The second formula follows from the substitution $t=\frac{1-s}{s}$, or by the integral representation:
	$\Log (A) =    \int_0^1  (A-\I)\left[ t A + (1-t) \I\right]^{-1} \d t $
    \cite{Wou65}, 
    \cite[(3.12)]{HP14},
    \cite[\S 11]{Hig08}.

    Item~\ref{item:DLog_positivity} follows from the formula given in Item~\ref{item:Dlog_integral}.

Item~\ref{item:Dlog_division} can be shown as follows:
Set $\Delta \coloneqq A^{-\frac12} B A^{-\frac12}$. Then, by Item~\ref{item:Dlog_integral}, we have 
 \begin{align*}
	{\rm D} \log[A^{-1}](\Delta)
	&=\int_0^\infty \frac{1}{u\I+A^{-1}}\,\Delta\,\frac{1}{u\I+A^{-1}}\,\d u
    \notag
	\\
	&=\int_0^\infty A^{\frac{1}{2}}\frac{1}{Au+\I}A^{\frac{1}{2}}\Delta A^{\frac{1}{2}}\frac{1}{Au+\I}A^{\frac{1}{2}}\,\d u
    \notag
	\\
	&=A^{\frac{1}{2}}\int_0^\infty \frac{1}{Au+\I}\,B\,\frac{1}{Au+\I}\,\d u\,A^{\frac{1}{2}} 
    \notag
	\\
	&\overset{(*)}{=} A^{\frac{1}{2}}\int_0^\infty \frac{1}{A+v\I}\,B\,\frac{1}{A+v\I}\,\d v\,A^{\frac{1}{2}}
    \notag
	\\
	&=A^{\frac{1}{2}} \cdot{\rm D}\log[A](B)\cdot A^{\frac{1}{2}}.
 \end{align*}
where ($*$) follows from the substitution $v=u^{-1}$.

Item~\ref{item:Beigi-Tomamichel} was first shown by Beigi and Tomamichel \cite[Lemma 3]{BT24}.
Here, we provide an alternative proof.
Using Theorem~\ref{theo:Dlog_formula}:
\begin{align*}
\mathrm{D}\log[A+B](B) 
&= \int_0^1 \left\{ u A < (1-u) B \right\} \, \d u
\notag
\\
&\overset{(*)}{=}\int_0^{\infty} \left\{ v A <  B \right\} \frac{1}{(1+v)^2} \, \d v
\notag
\\
&\leq
\int_0^{\infty} \left\{ v A <  B \right\} \, \d v
\notag
\\
&= \mathrm{D}\log[A](B),
\notag
\end{align*}
where ($*$) follows from the substitution $v = \frac{u}{1-u}$.

Item~\ref{item:Dlog_norm} can be shown by using Theorem~\ref{theo:Dlog_formula}. Let $r=\|A^{-1/2}BA^{-1/2}\|_{\infty}$, then
\begin{align}
    \|\mathrm{D}\log[A](B)\|_\infty=\left\|\int_{0}^{r} \{ uA < B \} \d u\right\|_{\infty} \leq\int_{0}^{r} \left\|\{ uA < B \}\right\|_{\infty} \d u=r.
\end{align}

\end{proof}


\section{Operator Layer Cake Theorem} \label{sec:layer-cake}

This section is devoted to establishing a layer cake representation for the operator logarithm introduced in Appendix~\ref{sec:log}.

\begin{theo}[Operator layer cake representation] \label{theo:Dlog_formula}
For any positive definite operator $A$ and any self-adjoint operator $B$ on a finite-dimensional Hilbert space,
the following representation holds:
\begin{align} \label{eq:Dlog_formula}
\mathrm{D} \log[A](B) = \int_{0}^{\infty} \{ uA < B \} \d u\,-\,\int_{-\infty}^{0} \{ uA > B \} \d u,
\end{align}
where $ \mathrm{D} \log[A](B) $ is the directional derivative of the operator logarithm (see Appendix~\ref{sec:log}),
and $ \{ uA < B \} \equiv \{ B - uA > 0 \} $ denotes the projection onto the positive spectral subspace of $B-uA$. Similarly, $ \{ uA > B \} \equiv \{ uA - B > 0 \} $ denotes the projection onto the negative spectral subspace of $B-uA$.
\end{theo}

\begin{remark}
Let $\Delta = A^{-\frac{1}{2}}BA^{-\frac{1}{2}}$.
Notice that $\{u\in \mathbb{R} \mid B-uA \; \text{singular}\}=\{u\in \mathbb{R} \mid \Delta -u\I \; \text{singular}\}=\spec\left(\Delta\right)$ as $B-uA=A^{\frac{1}{2}}(\Delta-u\I)A^{\frac{1}{2}}$ and $A$ has full support.
Since $\Delta$ has discrete spectrum (due to the finite-dimensional assumption), the set $\{u\in \mathbb{R} \mid B-uA \; \text{singular}\}$ has measure zero.
Hence, in the two integrals, one can replace $\{ uA < B \}$ or $\{ uA > B \}$ in \eqref{eq:Dlog_formula} by $\{ uA \leq B \}$ or $\{ uA \geq B \}$, respectively.
\end{remark}



We begin with a definition and a lemma that appear repeatedly in the proof of the theorem.

    Let $X$ be an operator on a finite-dimensional Hilbert space. We define the real part (Hermitian part) and the imaginary part (skew-Hermitian part) of $X$ by \[\mathrm{Re}(X)\coloneqq\frac{X+X^\dagger}{2}, \quad\mathrm{Im}(X)\coloneqq\frac{X-X^\dagger}{2\mathrm{i}}.\]

\begin{lemm}\label{lemm:real_part}
    If $\mathrm{Re}(X)>0$ (resp.~$<0$), then all the eigenvalues of $X$ have positive (resp.~negative) real part; similarly, if $\mathrm{Im}(X)>0$ (resp.~$<0$), then all the eigenvalues of $X$ have positive (resp. negative) imaginary part.
\end{lemm}

\begin{proof}
    We only prove the case when $\mathrm{Re}(X)>0$.
    Let $\lambda$ be an eigenvalue of $X$ and let $\ket{v}$ be a corresponding non-zero normalized eigenvector, i.e.,
\begin{align} X\ket{v} = \lambda \ket{v}, \quad \langle v | v \rangle = 1. 
\end{align}
Then,
\begin{align}
\mathrm{Re}(\lambda)=\frac{\lambda+\bar{\lambda}}{2}
=\frac{\bra{v}X\ket{v}+\overline{\bra{v}X\ket{v}}}{2}=\bra{v} \mathrm{Re}(X) \ket{v}>0.
\end{align}
The remaining cases follow from similar proofs.
\end{proof}

Now, we are ready to present the proof of Theorem~\ref{theo:Dlog_formula}.

\begin{proof}[Proof of Theorem~\ref{theo:Dlog_formula}]
	In the following argument, $\I$ denotes the identity operator.
	Define $\Delta=A^{-\frac{1}{2}}BA^{-\frac{1}{2}}$ and let $r, R>0$ satisfy $r\I>\Delta>-r\I$ and $R\I>B>-R\I$.
	Additionally, let $C_{R}^+$ (resp.~$C_{R}^-$) denote the boundary of the counterclockwise
	semidisc with center $z = 0$ and radius $R$ that lies entirely in the right (resp.~left) half-plane, let $C_{R,\text{arc}}^+$ (resp.~$C_{R,\text{arc}}^-$) denote the counterclockwise
	semicircular arc with center $z = 0$ and radius $R$ that lies entirely in the right (resp.~left) half-plane, and let $C_{R}^\circ$ denote the full circle centered at $z=0$ with radius $R$.
	
	Note that if $u>r$, then $\{B-uA> 0\}$ is the zero operator, because $r\I>\Delta$ implies $rA>B$, and thus $B-uA<B-rA<0$. Similarly, if $u<-r$, then $\{B-uA< 0\}=0$, since $-r\I<\Delta$ gives $B-uA>B+rA>0$.
	Therefore, for the first and second integration regions in \eqref{eq:Dlog_formula}, we only need to integrate over $[0,r]$ and $[-r, 0]$, respectively.
	
	To tackle $\{ B - uA > 0 \}$, let $C(u)$ be any simple closed counterclockwise contour that encloses exactly the positive eigenvalues of $B-uA$.
	We apply the residue theorem to $B-uA$ (see e.g.~\cite[Problem 5.9, p.~40]{Kat95}) to obtain
	\begin{align}
	\left\{ B - uA > 0 \right\} = \frac{1}{2\pi \mathrm{i}}\ointctrclockwise_{C(u)} \frac{1}{z\I-(B-uA)}\, \d z.
	\end{align}

	Suppose that \( u \notin \spec(\Delta) \). Then \( B - uA \) is non-singular.  
For \(u\geq 0\), we have
\begin{align}
    B-uA \leq B < R\I .
\end{align}
Hence all positive eigenvalues of \(B-uA\) lie in \((0,R)\). Moreover,
since \(u\notin \spec(\Delta)\), zero is not an eigenvalue of \(B-uA\).
Therefore, the right semidisc contour \(C_R^+\) encloses exactly the
positive eigenvalues of \(B-uA\), so we may use \( C_R^+ \) in place of \( C(u) \) for all such values of \( u \) and write:
\begin{align}\label{eq:residue-theorem}
	\left\{ B - uA > 0 \right\} = \frac{1}{2\pi \mathrm{i}}\ointctrclockwise_{C_R^+} \frac{1}{z\I-(B-uA)}\, \d z.
\end{align}

Although the above integral is well-defined in the Riemann sense, for the purposes of the following calculations, we shall interpret it in the sense of Lebesgue. Specifically, by writing
\begin{align}
\ointctrclockwise_{C_R^+} f(z)\, \mathrm{d}z,
\end{align}
we mean the Lebesgue integral
\begin{align}
\int_{R}^{-R} f(\mathrm{i}x)\, \mathrm{i}\, \mathrm{d}x + \int_{-\pi/2}^{\pi/2} f\left( R e^{\mathrm{i} \theta} \right) \mathrm{i}R e^{\mathrm{i} \theta} \, \mathrm{d}\theta,
\end{align}
where the contour \( C_R^+ \) is parameterized by  
\begin{align}
C_R^+ = \{ \mathrm{i}x \mid R \geq x \geq -R \} \cup \{ R e^{\mathrm{i} \theta} \mid -\pi/2 \leq \theta \leq \pi/2 \}.
\end{align}
Moreover, not only the integral over $C_R^+$ but all contour integrals appearing below are to be interpreted in the Lebesgue sense, with the following parameterizations:
\begin{align*}
    C_{R,\text{arc}}^+ &=  \{ R e^{\mathrm{i} \theta} \mid -\pi/2 \leq \theta \leq \pi/2 \},\\
    C_R^- &= \{ \mathrm{i}x \mid -R \leq x \leq R \} \cup \{ R e^{\mathrm{i} \theta} \mid \pi/2 \leq \theta \leq 3\pi/2 \},\\
    C_{R,\text{arc}}^- &= \{ R e^{\mathrm{i} \theta} \mid \pi/2 \leq \theta \leq 3\pi/2 \},\\
    C_R^\circ &= \{ R e^{\mathrm{i} \theta} \mid 0 \leq \theta \leq 2\pi \}.
\end{align*}

\medskip
\noindent\textbf{Step 1: Rewriting the integrals.}
Returning to the calculation, since the set $\spec(\Delta)$ has measure zero, we may integrate both sides of \eqref{eq:residue-theorem} over $[0,r]$ in the sense of Lebesgue to obtain:
	\begin{align}
	\int_0^r\left\{ B - uA > 0 \right\} \d u 
	&= \frac{1}{2\pi \mathrm{i}} \int_0^r\ointctrclockwise_{C_R^+} \frac{1}{z \I-(B-uA)} \d z \, \d u
	\\
	&\overset{\text{(a)}}{=} \frac{1}{2\pi \mathrm{i}} \ointctrclockwise_{C_R^+} \int_0^r \frac{1}{z \I-(B-uA)} \d u \, \d z
	\\
	&= \frac{1}{2\pi \mathrm{i}} A^{\frac{-1}{2}}\ointctrclockwise_{C_R^+}\int_0^r \frac{1}{u \I+zA^{-1}-\Delta}\, \d u\, \d z\,A^{\frac{-1}{2}} 
	\\
	&\overset{\text{(b)}}{=} \frac{1}{2\pi \mathrm{i}} A^{\frac{-1}{2}}\ointctrclockwise_{C_R^+} \Log(r\I+zA^{-1}-\Delta)-\Log(zA^{-1}-\Delta)\,\d z\,A^{\frac{-1}{2}}  
    \\\label{eq:integral_positive}
    &\overset{\text{(c)}}{=} -\frac{1}{2\pi \mathrm{i}} A^{\frac{-1}{2}}\ointctrclockwise_{C_R^+}\Log(zA^{-1}-\Delta)\,\d z\,A^{\frac{-1}{2}}.
	\end{align}
	Here, (a) follows from Fubini's theorem; the required absolute integrability is established in Lemma~\ref{lemm:boundedness_Fubini} below. 
    In (b), the integral is evaluated by the fundamental theorem of calculus, since the spectrum of $u\I+zA^{-1}-\Delta$ does not touch the branch cut $\mathds{R}_{\leq 0}$ as $u$ varies from $0$ to $r$, except possibly at $z=0$.
    \footnote{
    Note that as long as $z\not\in \{0,R\}$, the spectrum of $z A^{-1}-\Delta$ will never be real. This is because  $\mathrm{Im}(z A^{-1}-\Delta)\gtrless 0$ implies $\mathrm{Im}(\spec(z A^{-1}-\Delta))\gtrless 0$ by Lemma~\ref{lemm:real_part}. For $z=R$, $zA^{-1}-\Delta$ is positive.
    For $z\in C_R^+\setminus\{0\}$ as the real number $u$ varies from $0$ to $r$, the spectrum of $ u \I + z A^{-1}-\Delta$ only shifts to the right in the complex plane, and hence never touches the branch cut.
    }
	For (c), by hypothesis $r \I > \Delta$, $\Log(r\I+zA^{-1}-\Delta)$ is holomorphic in $\mathrm{Re}(z)>-\varepsilon_+$, where $\varepsilon_+=\left\Vert(A^{\frac{1}{2}}(r\I-\Delta)A^{\frac{1}{2}})^{-1}\right\Vert^{-1}>0$.\footnote{
	The spectrum of $r\I+zA^{-1}-\Delta$ will never have negative real parts as long as $\mathrm{Re}(z)>-\left\Vert(A^{\frac{1}{2}}(r\I-\Delta)A^{\frac{1}{2}})^{-1}\right\Vert^{-1}$ by checking that $\mathrm{Re}(r\I+zA^{-1}-\Delta)> 0$.
	}
    Since the region enclosed by $C_R^+$ lies in $\{z\in \mathds{C} \mid \mathrm{Re}(z)>-\varepsilon_+\}$, by Cauchy's integral theorem,    
	\[
	\ointctrclockwise_{C_R^+}\Log(r\I+zA^{-1}-\Delta)\, \d z =0.
	\]

    Similarly, for the negative part,
    \begin{align}
		\int_{-r}^0\{ B - uA < 0 \} \,\d u 
		&= \frac{1}{2\pi \mathrm{i}}\int_{-r}^0 \ointctrclockwise_{C_R^-}\frac{1}{z\I-(B-uA)}\,\d z\,\d u 
		\\
		&\overset{\text{(a)}}{=}\frac{1}{2\pi \mathrm{i}}\ointctrclockwise_{C_R^-}\int_{-r}^0 \frac{1}{z\I-(B-uA)}\,\d u\,\d z 
		\\
		&=\frac{1}{2\pi \mathrm{i}} A^{\frac{-1}{2}}\ointctrclockwise_{C_R^-}\int_{-r}^0 \frac{-1}{-u\I-zA^{-1}+\Delta}\,\d u\,\d z\,A^{\frac{-1}{2}}  
		\\
		&\overset{\text{(b)}}{=} \frac{1}{2\pi \mathrm{i}} A^{\frac{-1}{2}}\ointctrclockwise_{C_R^-}\Log(-zA^{-1}+\Delta)-\Log(r\I-zA^{-1}+\Delta)\,\d z\,A^{\frac{-1}{2}}  
        \\
        &\overset{\text{(c)}}{=}\frac{1}{2\pi \mathrm{i}} A^{\frac{-1}{2}}\ointctrclockwise_{C_R^-}\Log(-zA^{-1}+\Delta)\,\d z\,A^{\frac{-1}{2}}.
	\end{align}
Here, (a) is justified by the same Fubini argument as above, using Lemma~\ref{lemm:boundedness_Fubini}. 
In (b), the integral is again evaluated by the fundamental theorem of calculus, since the spectrum of
$-u\I-zA^{-1}+\Delta$ does not touch the branch cut $\mathds{R}_{\leq 0}$ for $u\in[-r,0]$, except possibly at $z=0$.
For (c), by the hypothesis $r \I > -\Delta$, $\Log(r\I-zA^{-1}+\Delta)$ is holomorphic in $\mathrm{Re}(z)<\varepsilon_-$, where $\varepsilon_-=\left\Vert(A^{\frac{1}{2}}(r\I+\Delta)A^{\frac{1}{2}})^{-1}\right\Vert^{-1}>0$. Since the region enclosed by $C_R^-$ lies in $\{z\in \mathds{C} \mid \mathrm{Re}(z)<\varepsilon_-\}$, (c) follows from Cauchy's integral theorem.

\medskip
\noindent\textbf{Step 2: Combining the contour integrals.}
Concluding the first step, we have
\begin{align}
&\int_0^r\left\{ B - uA > 0 \right\} \d u-\int_{-r}^0\left\{ B - uA < 0 \right\} \d u\notag
\\
&=
-\frac{1}{2\pi \mathrm{i}} A^{\frac{-1}{2}}\left(\ointctrclockwise_{C_R^+}\Log(zA^{-1}-\Delta)\,\d z+\ointctrclockwise_{C_R^-}\Log(-zA^{-1}+\Delta)\,\d z\right)\,A^{\frac{-1}{2}}.
\label{eq:step1}
\end{align}

We split the integrals on the right-hand side of \eqref{eq:step1}. 
Before doing so, we verify the logarithmic decompositions used on the two arcs.

For $z\in C_{R,\mathrm{arc}}^+ \cup C_{R,\mathrm{arc}}^-$, write $z=Re^{\mathrm{i}\theta}$ with
$-\pi\leq\theta\leq\pi$. Then
\begin{align}
\mathrm{Re}(A^{-1}-z^{-1}\Delta)
&=
A^{-1}-R^{-1}\cos(\theta)\Delta 
=
A^{-1/2}\left(\I-R^{-1}\cos(\theta)B\right)A^{-1/2}>0,
\end{align}
where the last inequality follows from the choice of $R$. Hence, the spectrum of
$A^{-1}-z^{-1}\Delta$ lies in the open right half-plane by Lemma \ref{lemm:real_part}. For $z\in C_{R,\mathrm{arc}}^+$, since $z$ is nonzero and lies in the closed right half-plane, the spectrum of $zA^{-1}-\Delta=z\cdot(A^{-1}-z^{-1}\Delta)$ stays off the branch cut and the principal arguments
add without crossing $\pm \pi$. Therefore, the principal logarithm satisfies
\begin{align}
\Log(zA^{-1}-\Delta)
=
\Log(z)\I+\Log(A^{-1}-z^{-1}\Delta)
\end{align}
on $C_{R,\mathrm{arc}}^+$.
Similarly, for $z\in C_{R,\mathrm{arc}}^-$, $-z$ is nonzero and lies in the closed right half-plane. Thus
\begin{align}
\Log(-zA^{-1}+\Delta)
=
\Log(-z)\I+\Log(A^{-1}-z^{-1}\Delta)
\end{align}
on $C_{R,\mathrm{arc}}^-$.

Using these decompositions, the positive part becomes
\begin{align*}
&\ointctrclockwise_{C_R^+}\Log(zA^{-1}-\Delta)\,\d z\\
&=\int_{C_{R,\text{arc}}^+}\Log(zA^{-1}-\Delta)\,\d z
+\int_R^{-R} \Log(x\mathrm{i}A^{-1}-\Delta)\,\mathrm{i}\d x\\
&=\int_{C_{R,\text{arc}}^+}\Log(A^{-1}-z^{-1}\Delta)\,\d z
+\int_{C_{R,\text{arc}}^+}\Log(z)\I\,\d z
+\int_R^{-R} \Log(x\mathrm{i}A^{-1}-\Delta)\,\mathrm{i}\d x.
\end{align*}
Similarly, the negative part becomes
\begin{align*}
&\ointctrclockwise_{C_R^-}\Log(-zA^{-1}+\Delta)\,\d z\\
&=\int_{C_{R,\text{arc}}^-}\Log(-zA^{-1}+\Delta)\,\d z
+\int_{-R}^R \Log(-x\mathrm{i}A^{-1}+\Delta)\,\mathrm{i}\d x\\
&=\int_{C_{R,\text{arc}}^-}\Log(A^{-1}-z^{-1}\Delta)\,\d z
+\int_{C_{R,\text{arc}}^-}\Log(-z)\I\,\d z
+\int_{-R}^R \Log(-x\mathrm{i}A^{-1}+\Delta)\,\mathrm{i}\d x.
\end{align*}


However,
\begin{align*}
\int_{C_{R,\text{arc}}^+}\Log(z)\I\,\d z+\int_{C_{R,\text{arc}}^-}\Log(-z)\I\,\d z
=\int_{C_{R,\text{arc}}^+}\Log(z)\I\,\d z-\int_{C_{R,\text{arc}}^+}\Log(w)\I\,\d w=0
\end{align*}
by the substitution $w=-z$.

Also, for \(x>0\), the operator \(x\mathrm{i}A^{-1}-\Delta\) has
strictly positive imaginary part. Hence its spectrum lies in the open
upper half-plane, and the scalar principal-log identity, applied through
the holomorphic functional calculus, gives
\[
    \Log(-x\mathrm{i}A^{-1}+\Delta)
    -
    \Log(x\mathrm{i}A^{-1}-\Delta)
    =
    -\pi\mathrm{i}\I .
\]
Similarly, for \(x<0\), one obtains
\[
    \Log(-x\mathrm{i}A^{-1}+\Delta)
    -
    \Log(x\mathrm{i}A^{-1}-\Delta)
    =
    \pi\mathrm{i}\I .
\]
Therefore,
\begin{align*}
&\int_R^{-R} \Log(x\mathrm{i}A^{-1}-\Delta)\,\mathrm{i}\d x
+\int_{-R}^R \Log(-x\mathrm{i}A^{-1}+\Delta)\,\mathrm{i}\d x\\
&=-\int_{-R}^{R} \Log(x\mathrm{i}A^{-1}-\Delta)\,\mathrm{i}\d x
+\int_{-R}^R \Log(-x\mathrm{i}A^{-1}+\Delta)\,\mathrm{i}\d x\\
&=\int_{0}^R
\left(
    \Log(-x\mathrm{i}A^{-1}+\Delta)
    -
    \Log(x\mathrm{i}A^{-1}-\Delta)
\right)\,\mathrm{i}\d x
+\int_{-R}^{0}
\left(
    \Log(-x\mathrm{i}A^{-1}+\Delta)
    -
    \Log(x\mathrm{i}A^{-1}-\Delta)
\right)\,\mathrm{i}\d x\\
&=\int_{0}^R \left(-\pi\mathrm{i}\I\right)\,\mathrm{i}\d x
+\int_{-R}^{0} \left(\pi\mathrm{i}\I\right)\,\mathrm{i}\d x
=0.
\end{align*}
Hence,
\begin{align}
&\ointctrclockwise_{C_R^+}\Log(zA^{-1}-\Delta)\,\d z+\ointctrclockwise_{C_R^-}\Log(-zA^{-1}+\Delta)\,\d z\notag\\
&=\int_{C_{R,\text{arc}}^+}\Log(A^{-1}-z^{-1}\Delta)\,\d z+\int_{C_{R,\text{arc}}^-}\Log(A^{-1}-z^{-1}\Delta)\,\d z=\ointctrclockwise_{C_{R}^\circ}\Log(A^{-1}-z^{-1}\Delta)\,\d z.
\label{eq:step2-cancel}
\end{align}
Combining \eqref{eq:step1} and \eqref{eq:step2-cancel}, we obtain
\begin{align}
\int_0^r\left\{ B - uA > 0 \right\} \d u-\int_{-r}^0\left\{ B - uA < 0 \right\} \d u=
-\frac{1}{2\pi \mathrm{i}} A^{\frac{-1}{2}}
\ointctrclockwise_{C_R^\circ}
\Log(A^{-1}-z^{-1}\Delta)\,\d z
A^{\frac{-1}{2}}.
\label{eq:step2}
\end{align}

\medskip
\noindent\textbf{Step 3: Integrating by parts.}
By the holomorphic functional calculus chain rule and the linearity of the Fr\'echet derivative, we have
\[
\frac{\d}{\d z}\Log(A^{-1}-z^{-1}\Delta)
=
\mathrm{D}\Log[A^{-1}-z^{-1}\Delta](\Delta)\cdot z^{-2}.
\]
By integration by parts, we have
\begin{align}
\ointctrclockwise_{C_R^\circ}
\Log(A^{-1}-z^{-1}\Delta)\,\d z
&=
-\ointctrclockwise_{C_R^\circ}
z^{-1}\cdot\mathrm{D}\Log[A^{-1}-z^{-1}\Delta](\Delta)\,\d z
\\
&=
-\int_{0}^{2\pi}
{\mathrm D} \Log\left[A^{-1}-(R \mathrm{e}^{\mathrm{i}\theta})^{-1}\Delta\right](\Delta)
\cdot \mathrm{i} \d\theta,
\label{eq:ibp-step}
\end{align}
where the last equality is from the substitution $z=Re^{\mathrm{i}\theta}$.

Since \eqref{eq:ibp-step} holds for every sufficiently large $R$ satisfying $R\I>B>-R\I$, and its left-hand side is actually independent of $R$ by \eqref{eq:step2}, we may let $R\to\infty$. Since
\[
    A^{-1} - (Re^{\mathrm{i}\theta})^{-1}\Delta
    \rightarrow A^{-1}
\]
uniformly in \(\theta\) as \(R\to\infty\), by the continuity of the Fr\'echet derivative $\mathrm{D}\Log[\,\cdot\,]$ (Fact~\ref{item:Dlog_linear}), we obtain
\begin{align}
\ointctrclockwise_{C_R^\circ}\Log(A^{-1}-z^{-1}\Delta)\,\d z
=-\int_{0}^{2\pi}\mathrm{D} \Log[A^{-1}](\Delta)\cdot  \mathrm{i} \d\theta
=-2\pi\mathrm{i}\cdot\mathrm{D} \log[A^{-1}](\Delta).
\label{eq:ibp-limit}
\end{align}
Therefore,
\begin{align}
\int_0^r\left\{ B - uA > 0 \right\} \d u-\int_{-r}^0\left\{ B - uA < 0 \right\} \d u
=
A^{\frac{-1}{2}}\cdot\mathrm{D} \log[A^{-1}](\Delta)\cdot A^{\frac{-1}{2}}.
\end{align}
Invoking 
\[
\mathrm{D}\log[A^{-1}](\Delta) = A^{\frac{1}{2}} \cdot\mathrm{D}\log[A](B)\cdot A^{\frac{1}{2}}
\]
(Fact~\ref{item:Dlog_division}), we complete the proof.
\end{proof}

\begin{lemm} \label{lemm:boundedness_Fubini}
Following the notation used in the proof above, we have
\begin{align}
\int_0^r\ointctrclockwise_{C_R^+}
\left\| \frac{1}{z \I-(B-uA)}\right\| |\d z|\,\d u < \infty,
\qquad
\int_{-r}^0\ointctrclockwise_{C_R^-}
\left\| \frac{1}{z \I-(B-uA)}\right\| |\d z|\,\d u < \infty.
\end{align}
\end{lemm}
\begin{proof}
	We break $C_R^+$ into two parts: the straight line $C_{R,\,\text{straight}}^+ \coloneqq\{\mathrm{i}x : R\geq x \geq -R\}$ and the arc $C_{R,\,\text{arc}}^+ \coloneqq\{R\mathrm{e}^{\mathrm{i}\theta}: -\pi/2\leq\theta\leq\pi/2\}$.
	Notice that the integrand is the reciprocal of the minimum distance between $z$ and the spectrum of $B-uA$. 
	For $z\in C_{R,\,\text{arc}}^+$, if $\beta$ is in the spectrum of $B-uA$ and $\beta\geq 0$, then
	\[
	|z-\beta|\geq |z| - |\beta| =|z|-\beta\geq R-\Vert B\Vert>0.
	\]
	On the other hand, if $\beta< 0$, then it is clear that $|z-\beta|>R\geq R-\Vert B\Vert$. As a result, the integrand is bounded by $(R-\Vert B\Vert)^{-1}$. Thus 
	\[
	\int_0^r\int_{C_{R,\,\text{arc}}^+} \left \Vert\frac{1}{z\I-(B-uA)}\right\Vert |\d z|\,\d u\leq r\cdot\pi R\cdot \frac{1}{R-\Vert B\Vert}< \infty.
	\]
	
	For the straight line part, we rewrite the integral:
	\[
	\int_0^r\int_{C_{R,\,\text{straight}}^+}  \left\Vert\frac{1}{z\I-(B-uA)}\right\Vert |\d z|\,\d u =\int_0^r\int_{-R}^R \left \Vert\frac{1}{\mathrm{i}x\I-(B-uA)}\right\Vert \d x\, \d u.
	\]
	Fix $x$ and $u$. 
    The only possible singularities occur when
\(\mathrm{i}x\I-(B-uA)\) is not invertible. Since \(B-uA\) is self-adjoint,
its spectrum is real. Thus this can happen only if
 \(x=0\). In this case, non-invertibility is equivalent to
\(0\in\spec(B-uA)\), or equivalently \(u\in\spec(\Delta)\). Therefore the possible
singularities are contained in
\[
    \spec(\Delta)\times\{0\}.
\]
Since \(\spec(\Delta)\) is finite, this set has two-dimensional Lebesgue
measure zero and hence does not affect the double integral.
    
    Let $\{\beta_j\}$ be the set of eigenvalues of $B-uA$; then $\{\mathrm{i}x-\beta_j\}$ is the set of eigenvalues of $\mathrm{i}x\I-(B-uA)$. Thus 
	\begin{align*}
		\left\Vert\frac{1}{\mathrm{i}x\I-(B-uA)}\right\Vert
		=\frac{1}{\min_{j}|\mathrm{i}x-\beta_j|}
        =\frac{1}{\sqrt{x^2+\min_{j}|\beta_j|^2}}
        =\frac{1}{\sqrt{x^2+\left\Vert(B-uA)^{-1}\right\Vert^{-2}}},
	\end{align*}
	as
	$$
	\min_{j}|\beta_j|=\left\Vert(B-uA)^{-1}\right\Vert^{-1}.
	$$
	However, 
	\begin{align*}
		\left\Vert(B-uA)^{-1}\right\Vert
		&=\left\Vert A^{-\frac{1}{2}}(\Delta-u\I)^{-1}A^{-\frac{1}{2}}\right\Vert
        \leq\left\Vert A^{-1}\right\Vert\left\Vert (\Delta-u\I)^{-1}\right\Vert.
	\end{align*}
	Hence
	\begin{align*}
		\left \Vert\frac{1}{\mathrm{i}x\I-(B-uA)}\right\Vert
		&=\frac{1}{\sqrt{x^2+\left\Vert(B-uA)^{-1}\right\Vert^{-2}}}\\
		&\leq\frac{1}{\sqrt{x^2+\left\Vert A^{-1}\right\Vert^{-2}\left\Vert (\Delta-u\I)^{-1}\right\Vert^{-2}}}\\
		&=\frac{1}{\sqrt{x^2+\left\Vert A^{-1}\right\Vert^{-2}\left(\min_{1\leq j\leq n}\left\vert\delta_j-u\right\vert\right)^2}}\\
		&=\max_{1\leq j\leq n}\frac{1}{\sqrt{x^2+\left\Vert A^{-1}\right\Vert^{-2}\left(\delta_j-u\right)^2}},
	\end{align*}
	where $\{\delta_j\}_{j=1}^n$ are the eigenvalues of $\Delta = A^{-\frac12} B A^{-\frac12}$.
	
	For the integral
	\[
	\int_0^r\int_{-R}^R \left \Vert\frac{1}{\mathrm{i}x \I-(B-uA)}\right\Vert \d x\,\d u
	\]
	to be finite, we only need each of
	\[
	\int_0^r\int_{-R}^R \frac{1}{\sqrt{x^2+\left\Vert A^{-1}\right\Vert^{-2}\left(\delta_j-u\right)^2}} \d x\,\d u
	\]
	to be finite.
    Indeed, the substitution \(y=\|A^{-1}\|^{-1}(u-\delta_j)\) reduces the integrand to 
\((x^2+y^2)^{-1/2}\), whose local integrability is immediate in polar coordinates.
    
The estimate for the negative part follows similarly.
\end{proof}

{\larger
\bibliographystyle{myIEEEtran}
\bibliography{reference.bib, operator.bib}
}

\end{document}